\documentclass[12pt]{article}

\usepackage{amsmath,amsfonts,amsthm}
\usepackage{hyperref}
\usepackage{fullpage}

\usepackage{lipsum}
\usepackage{graphicx}
\usepackage{epstopdf}
\usepackage{algorithmic}
\ifpdf
  \DeclareGraphicsExtensions{.eps,.pdf,.png,.jpg}
\else
  \DeclareGraphicsExtensions{.eps}
\fi

\usepackage{enumitem}
\setlist[enumerate]{leftmargin=.5in}
\setlist[itemize]{leftmargin=.5in}

\title{Sampling Theory for Super-Resolution with Implicit Neural Representations}

\author{Mahrokh Najaf, Gregory Ongie\thanks{Department of Mathematical and Statistical Sciences, Marquette University, Milwaukee, WI, USA.\\ Corresponding author: Gregory Ongie (gregory.ongie@marquette.edu)}}

\usepackage{amsopn}

\usepackage{amsmath}
\usepackage{graphicx} 
\usepackage[]{mdframed}
\usepackage{cleveref}

\usepackage{xcolor}

\usepackage{caption}
\usepackage{amsmath,amsfonts,bm}

\def\1{\bm{1}}

\def\vc{{\bm{c}}}

\def\ve{{\bm{e}}}

\def\vk{{\bm{k}}}

\def\vn{{\bm{n}}}

\def\vq{{\bm{q}}}

\def\vv{{\bm{v}}}
\def\vw{{\bm{w}}}
\def\vx{{\bm{x}}}
\def\vy{{\bm{y}}}
\def\vz{{\bm{z}}}

\def\mF{{\bm{F}}}

\def\mI{{\bm{I}}}

\DeclareMathAlphabet{\mathsfit}{\encodingdefault}{\sfdefault}{m}{sl}
\SetMathAlphabet{\mathsfit}{bold}{\encodingdefault}{\sfdefault}{bx}{n}

\def\gA{{\mathcal{A}}}

\def\gC{{\mathcal{C}}}

\def\gF{{\mathcal{F}}}

\def\gK{{\mathcal{K}}}
\def\gL{{\mathcal{L}}}
\def\gM{{\mathcal{M}}}
\def\gN{{\mathcal{N}}}

\def\gP{{\mathcal{P}}}
\def\gQ{{\mathcal{Q}}}

\def\gW{{\mathcal{W}}}
\def\gX{{\mathcal{X}}}

\newcommand{\R}{\mathbb{R}}

\DeclareMathOperator*{\argmin}{arg\,min}

\DeclareMathOperator{\sign}{sign}
\DeclareMathOperator{\rank}{rank}

\renewcommand{\S}{\mathbb{S}}

\newcommand{\T}{\top}

\newcommand{\F}{\mathcal{F}}

\newcommand{\Td}{\mathbb{T}^d}
\newcommand{\torus}{\mathbb{T}}
\newcommand{\TP}{\mathsf{TP}}
\newcommand{\Csym}{\mathbb{C}_{\mathrm{sym}}}
\newcommand{\Peta}{S_\eta}
\newcommand{\Zd}{\mathbb{Z}^d}
\newcommand{\C}{\mathbb{C}}

\renewcommand{\Im}{\mathrm{Im}}

\newtheorem{myLem}{Lemma}
\newtheorem*{myLemStar}{Lemma}
\newtheorem{myThm}{Theorem}

\begin{document}

\maketitle
\begin{abstract}
Implicit neural representations (INRs) have emerged as a powerful tool for solving inverse problems in computer vision and computational imaging. INRs represent images as continuous domain functions realized by a neural network taking spatial coordinates as inputs. However, unlike traditional pixel representations, little is known about the sample complexity of estimating images using INRs in the context of linear inverse problems. Towards this end, we study the sampling requirements for recovery of a continuous domain image from its low-pass Fourier samples by fitting a single hidden-layer INR with ReLU activation and a Fourier features layer using a generalized form of weight decay regularization. Our key insight is to relate minimizers of this non-convex parameter space optimization problem to minimizers of a convex penalty defined over an infinite-dimensional space of measures. We identify a sufficient number of Fourier samples for which an image realized by an INR is exactly recoverable by solving the INR training problem.  To validate our theory, we empirically assess the probability of achieving exact recovery of images realized by low-width single hidden-layer INRs, and illustrate the performance of INRs on super-resolution recovery of continuous domain phantom images.
\end{abstract}

\section{Introduction}
Many inverse problems in signal and image processing are naturally posed in continuous domain. For example, in magnetic resonance imaging, measurements are modeled as samples of the Fourier transform of a function defined over continuous spatial coordinates \cite{fessler2010model}. Often these continuous domain inverse problems are approximated using discrete image or signal representations, along with an appropriate discretization of the forward model, such as the discrete Fourier transform in place of the continuous Fourier transform. However, this can lead to undesirable discretization artifacts, such as Gibb's ringing \cite{gelb2002,veraart2016gibbs,ongie2015,ongie2016off}. Such artifacts may be reduced by increasing the resolution of the discretization grid, but at the cost of additional computational complexity and potentially making the problem ill-posed without additional regularization.

Recently, continuous domain image representations using neural networks have gained traction in a variety of inverse problems arising in computer vision and computational imaging (see, e.g., the recent review \cite{xie2022neural}). These so-called \emph{implicit neural representations} (INRs), also variously known as 
\emph{coordinate-based neural networks} or \emph{neural radiance fields}, parameterize a signal or image as a neural network defined over continuous spatial coordinates \cite{mildenhall2021nerf,sitzmann2020implicit,tancik2020fourier}. The main benefits of INRs over discrete representations are that (1) they offer a compact (non-linear) representation of images with potentially fewer parameters than a discrete representation, (2) they allow for more accurate continuous domain forward modeling, which can reduce the impact of discretization artifacts, and (3) standard neural network training can be used to fit an INR to measurements, sidestepping the need for custom iterative solvers.

However, the fundamental limits of image estimation using INRs have not been explored in detail. One basic unanswered question is: How many linear measurements are necessary and sufficient for accurate reconstruction of an image by fitting an INR? Additionally, how does the architecture of the INR influence the recovered image?  As a step in the direction of answering these questions, we consider the recovery of an image from its low-pass Fourier samples using INRs. In particular, we focus on a popular INR architecture originally introduced in \cite{tancik2020fourier} that combines a fully connected neural network with a Fourier features layer. 

Recent work has shown that training a single hidden-layer ReLU network with 
$\ell^2$ regularization on its parameters (also known as ``weight decay'' regularization) imparts a sparsity effect, such that the network learned has low effective width \cite{savarese2019infinite,ongie2019function,parhi2021banach,parhi2023deep}. Bringing this perspective to INRs, we study the conditions under which an image realizable by a single-hidden layer, low-width INR is exactly recoverable from low-pass Fourier samples by minimizing a generalized weight decay objective. Our main theoretical insight is to relate minimizers of this non-convex parameter space optimization problem to minimizers of a convex optimization problem posed over a space of measures.

Using this framework, we prove two sampling theorems. First, we identify a sufficient number of low-pass Fourier samples such that an image realizable by a width-1 INR is the unique minimizer of the INR training problem regularized with a modified form of weight decay. Additionally, under more restrictive assumptions on the frequencies used to define the Fourier features layer, we extend the result to images realizable by a general width-$s$ INR.  

Finally, we validate our theory with a series of experiments. First, we demonstrate that random images generated by low-width INRs are exactly recoverable from low-pass Fourier samples by solving the generalized weight decay regularized INR training problem. Next, we illustrate the performance of INRs on super-resolution recovery of continuous domain phantom images, highlighting the effects of generalized weight regularization in practical INR training. Lastly, we empirically explore the role of depth of INRs for super-resolution recovery. We show that deeper INRs trained with weight decay regularization perform better at recovering piecewise constant phantom images, suggesting that deep INRs may exhibit a bias towards piecewise constant representations.

Preliminary versions of the results in this paper were presented in \cite{asilomar}. The majority of proofs were omitted from \cite{asilomar}, which are included here. Additionally, this paper expands the sampling theory to the general width-$s$ INR case, and contains an expanded experiments section.

\subsection{Related Work}\label{sec:relatedwork}
Closely related to the Fourier sampling model studied in this work, several papers have explored INRs for MRI reconstruction, including super-resolution imaging \cite{wu2021irem,wu2022arbitrary,van2022scale}. Other applications of INRs to MRI reconstruction include: multi-contrast imaging \cite{mcginnis2023single}, dynamic imaging \cite{huang2023neural,feng2022spatiotemporal,kunz2024implicit}, and compressed sensing with prior-image constraints \cite{shen2022nerp}.

While this work focuses on one type of INR architecture, a wide variety of other INR architectures have been proposed. Early INR architectures focused on fully connected ReLU networks with an initial \emph{positional encoding} layer, inspired by transformer architectures \cite{mildenhall2021nerf}. Later work investigated alternative positional encoding layers given by random Fourier features \cite{tancik2020fourier} (the focus of this work), and a multiresolution hash encoding \cite{muller2022instant}. Other approaches to INR architecture design do not use an initial positional encoding layer, but instead consider fully connected networks with non-standard activation functions, including periodic \cite{sitzmann2020implicit} or variable-periodic \cite{liu2024finer} functions, Gaussian functions \cite{ramasinghe2022beyond}, sinc functions \cite{saratchandran2024sampling}, and wavelet-inspired activation functions \cite{saragadam2023wire,shenouda2024relus}. See the recent review \cite{essakine2024we} for a detailed comparison of various popular INR architectures.

Several recent works have also sought to mathematically characterize the expressive power and inductive bias of INRs. A neural tangent kernel (NTK) analysis of INRs is given in \cite{tancik2020fourier}, which shows that a tunable random Fourier features layer can help alleviate the low spatial-frequency spectral bias associated with fully connected ReLU networks. Relatedly, \cite{yuce2022structured} studies the expressive power of a variety of INR architectures (including those proposed in \cite{sitzmann2020implicit,tancik2020fourier}) through a generalized Fourier series analysis. Their analysis shows that only functions given by a linear combination of certain harmonics of the initial feature map, with additional structural constraints on the coefficients, are realizable with a given INR architecture. Follow-up work expands on the Fourier series perspective specifically for INRs with sinusoidal activations \cite{novello2022understanding} and its implications for robust INR initialization \cite{novello2025tuningfrequenciesrobusttraining}. Additionally, recent work investigates incorporating batch normalization in INR layers as a means to reduce low spatial-frequency spectral bias of INR architectures \cite{cai2024batch}. 

Most closely related to this work, \cite{shenouda2024relus} introduces a means to measure function space regularity of a shallow INR representation using the \emph{variation norm}, which is the function space regularizer induced by weight decay regularization. Similarly, in this work we consider regularizing INRs by quantities analogous to the variation norm. However, the present work goes beyond \cite{shenouda2024relus} in proving super-resolution recovery guarantees for images realized by low-width INRs.

Finally, the proof techniques used in this work draw from the ``off-the-grid'' compressed sensing literature. This includes works focused on super-resolution of sparse measures \cite{bhaskar2013atomic,chi2014compressive,candes2014towards,chi2020harnessing} and multidimensional generalizations \cite{poon2019multidimensional,eftekhari2021stable,poon2023geometry,kurmanbek2023multivariate}. In particular, the image model investigated in this work is closely related to (though distinct from) the continuous-domain piecewise constant image model investigated in \cite{ongie2015,ongie2016off,ongie2017convex}. Algebraic facts regarding the zero-sets of two-variable trigonometric polynomials established in \cite{ongie2016off} are central to key proofs in this work; higher-dimensional extensions of these results have  recently been explored in \cite{zou2021recovery}.

\section{Problem formulation}
\label{sec:Problem formulation}
We consider the recovery of a $d$-dimensional real-valued continuous domain signal/image $f:[0,1]^d\rightarrow \R$ from given a finite sample of its Fourier coefficients
\[
\hat{f}[\vk] = \int_{[0,1]^d}f(\vx)e^{-2\pi \mathrm{i} \vk^\top \vx}d\vx,
\]
where $\vk\in \mathbb{Z}^d$ is any $d$-tuple of integer spatial frequencies. In particular, suppose we sample frequencies belonging to a uniformly sampled square region in frequency domain
$
\Omega = \{\vk \in \mathbb{Z}^d : \|\vk\|_\infty \leq K \}
$
where $K$ is a positive integer.
Let $\gF_\Omega f = (\hat{f}[\vk])_{\vk \in \Omega}$ denote the vector of Fourier coefficients of $f$ restricted to $\Omega$. We study the inverse problem of recovering the function $f$ given $\vy = \gF_\Omega f$. 

Clearly, this is an ill-posed inverse problem if no further assumptions are made on $f$. One approach to resolve this ill-posedness is to assume $f$ is well-approximated by a function realized by a ``simple'' neural network $f_\theta:[0,1]^d\rightarrow \R$, i.e., an INR, with trainable parameters $\theta$. Under this assumption, one can attempt to recover $f$ from its measurements $\vy$ by solving the (non-linear) least squares problem
\begin{equation}\label{eq:lsfit_noreg}
    \min_{\theta} \tfrac{1}{2}\|\gF_\Omega f_\theta - \vy\|_2^2.
\end{equation}
However, this problem may still be ill-posed (i.e., \eqref{eq:lsfit_noreg} may have multiple solutions) unless the number of trainable parameters used to define the INR $f_\theta$ is heavily constrained. Though, even without constraining the number of parameters, useful solutions of \eqref{eq:lsfit_noreg} may be found by relying on implicit regularization induced by practical gradient-based algorithms \cite{gunasekar2018characterizing}. 

As an alternative to using an overly constrained network architecture or relying on implicit regularization to resolve ill-posedness, we propose incorporating an explicit parameter space regularizer $R(\theta)$ to the least squares objective:
\begin{equation}\label{eq:lsfit}
    \min_{\theta} \tfrac{1}{2}\|\gF_\Omega f_\theta - \vy\|_2^2 + \lambda R(\theta),
\end{equation}
where $\lambda > 0$ is a regularization parameter.
In particular, we focus on a class of regularizers $R(\theta)$ that generalize \emph{weight decay regularization} \cite{krogh1991simple}, the squared $\ell^2$-norm of all trainable parameters; this class is described in more detail below in \Cref{sec:weight_decay}. 

In a noise-free scenario, the regularization parameter $\lambda$ in \eqref{eq:lsfit} needs to be very small to ensure data consistency  is tightly enforced.  As a model for this situation, we will focus on the equality constrained problem
\begin{equation}\label{eq:opt_param_space1}
\min_{\theta} R(\theta)~~s.t.~~\gF_\Omega f_\theta = \vy,
\end{equation}
which can be thought of as the limiting case of \eqref{eq:lsfit} as $\lambda\rightarrow 0$. 

Our main goal is to characterize global minimizers of  \eqref{eq:opt_param_space1}. In particular, we are interested in its \emph{function space minimizers}, i.e., functions $f_{\theta^*}$ where $\theta^*$ is a global minimizer of \eqref{eq:opt_param_space1}. Given $\vy = \gF_\Omega f$ where $f$ is realizable as an INR, we ask: when is $f$ the unique function space minimizer of \eqref{eq:opt_param_space1}? Next, we describe the architectural assumptions we put on the INR $f_\theta$ and the associated regularizer $R(\theta)$.

\subsection{INR Architecture}

In this work, we focus on an INR architecture first proposed in \cite{tancik2020fourier} that combines a trainable fully connected network with an initial Fourier features layer. For simplicity, we primarily consider a ``shallow'' INR architecture, where the fully connected network has a single-hidden layer with ReLU activation.

\begin{figure*}
    \centering
    \includegraphics[width=0.8\textwidth]{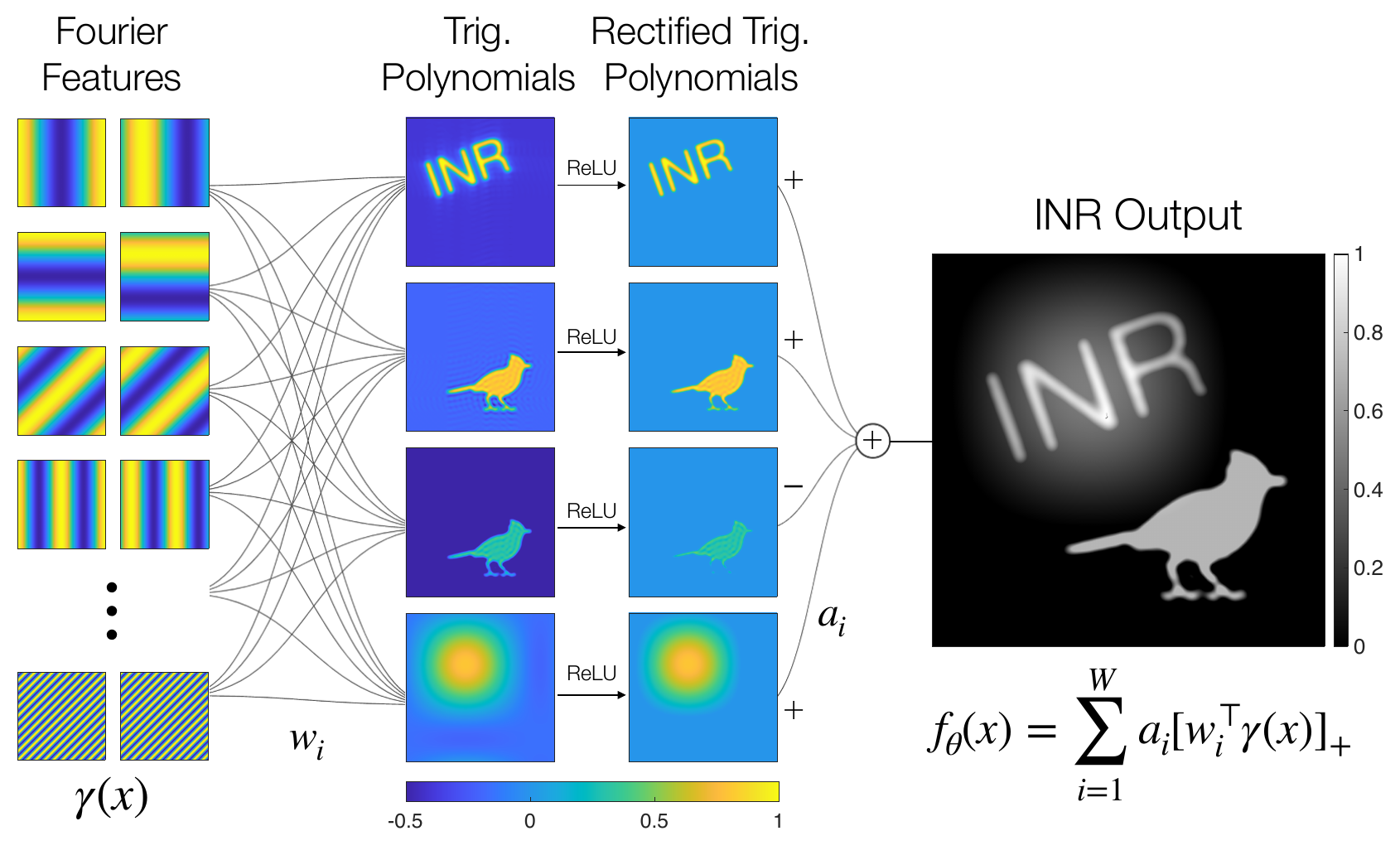}
    \caption{\footnotesize The shallow INR architecture considered in this work represents an image as a linear combination of rectified trigonometric polynomials. We study the sample complexity of estimating an image of this type from its low-pass Fourier samples.}
    \label{fig:RTP_orig}
\end{figure*}

Specifically, consider a shallow ReLU network $g_\theta:\R^D\rightarrow \R$ of width-$W$ and input dimension $D$ given by
\begin{equation}\label{eq:MLP}
g_\theta(\vv) = \sum_{i=1}^W a_i[\vw_i^\T \vv]_+,
\end{equation}
where $[t]_+ := \max\{0,t\}$ is the ReLU activation function, and $\theta = \left((a_i,\vw_i)\right)_{i=1}^W$ denotes the vector of all trainable parameters with  $a_i \in \R$ and $\vw_i \in \R^D$ for all $i\in [W]$. We denote the resulting parameter space by $\Theta_W \simeq \R^{W(D+1)}$.
By choosing $D = d$, this architecture could be used as an INR directly. However, as shown in \cite{tancik2020fourier}, it is challenging to recover high-frequency content in images using such an architecture.

Instead, to enhance recovery of high-frequency content, \cite{tancik2020fourier} proposed adding an initial Fourier features layer to the architecture: given a collection of non-zero frequency vectors $\Omega_0 = \{\vk_1,...,\vk_p\}\subset \mathbb{R}^d$ define the Fourier features embedding $\bm\gamma:\R^d\rightarrow \R^{D}$ with $D=2p+1$ by
\begin{equation}\label{eq:ffmap}
\bm\gamma(\vx) = 
[1,\sqrt{2}\cos(2\pi\vk_1^\T\vx),...,\sqrt{2}\cos(2\pi\vk_p^\T\vx),
\sqrt{2}\sin(2\pi\vk_1^\T\vx), ...,
\sqrt{2}\sin(2\pi\vk_p^\T\vx)]^\T.
\end{equation} 
This gives the shallow INR architecture:
\begin{equation}\label{eq:ftheta}
{f_\theta(\vx) = g_\theta(\bm\gamma(\vx)) = \sum_{i=1}^W a_i[\vw_i^\T \bm\gamma(\vx)]_+}.
\end{equation}
To better understand the impact of the Fourier features embedding,
for any weight vector $\vw \in \R^{D}$, let us write $\vw = [w^{(0)},\vw^{(1)}, \vw^{(2)}]$  where $w^{(0)} \in \R$ and $\vw^{(1)}, \vw^{(2)} \in \R^p$. Then defining $\tau(\vx) := \vw^\T \bm\gamma(\vx)$, we have
\begin{equation}\label{eq:trigpoly}
\tau(\vx) = w^{(0)} + \sqrt{2}\sum_{j=1}^{p} \left(w_{j}^{(1)} \cos(2\pi\bm\vk_{j}^\T\vx) + w_{j}^{(2)} \sin(2\pi\bm\vk_j^\T\vx)\right).
\end{equation}
This shows $\tau$ is a real-valued \emph{trigonometric polynomial}, i.e., $\tau$ is a function with frequency support contained in a finite set.
Therefore, a shallow INR with ReLU activation and a Fourier features embedding is a weighted linear combination of \emph{rectified trigonometric polynomials}, i.e., functions of the form $\vx \mapsto [\tau(\vx)]_+$ where $\tau$ is a real-valued trigonometric polynomial; see \Cref{fig:RTP_orig} for an illustration.

Originally, \cite{tancik2020fourier} proposed defining the frequency set $\Omega_0$ by sampling a fixed number of normally distributed random vectors $\vk \sim \gN(0,\sigma^2\mI)$, where the variance $\sigma^2$ is treated as a tunable parameter. To simplify the mathematical theory, in this work we focus on the case where $\Omega_0$ is a set of uniformly sampled integer-valued frequencies $\vk \in \Zd$. This implies that the INR \eqref{eq:ftheta} is a periodic function on $[0,1]^d$, or equivalently, a function defined over the $d$-dimensional torus $\mathbb{T}^d := (\mathbb{R}/\mathbb{Z})^d$. We note that using uniformly sampled integer frequencies in place of random frequencies for an INR Fourier features layer was also investigated in \cite{benbarka2022seeing}.

\subsection{Generalized Weight Decay Regularization}\label{sec:weight_decay}
We focus on a family of parameter space regularizers $R(\theta)$ that strictly generalize ``weight decay'' regularization \cite{krogh1991simple}, i.e., the squared $\ell^2$-norm of all trainable parameters. Consider any non-negative \emph{weighting function} $\eta$ defined over inner-layer weight vectors $\vw \in \R^D$. We say $\eta$ is an \emph{admissible} weighting function if it satisfies the following properties: 
\begin{enumerate}
    \item[(i)] $\eta$ is continuous,
    \item[(ii)] $\eta$ is positive 1-homogeneous: $\eta(\alpha \vw) = \alpha \eta(\vw)$ for all $\vw \in \R^d$ and $\alpha \geq 0$,
    \item[(iii)] $\eta(\vw) = 0$ implies $[\vw^\T\bm\gamma(\cdot)]_+ = 0$ (or, equivalently, $[\vw^\T\bm\gamma(\cdot)]_+ \neq 0$ implies $\eta(\vw)>0$).
\end{enumerate} 
Then, given any admissible weighting function $\eta$, we define the \emph{generalized weight decay penalty} applied to any parameter vector $\theta = ((a_i,\vw_i))_{i=1}^W$ by
\begin{equation}\label{eq:param_space_reg}
    R(\theta) = \frac{1}{2}\sum_{i=1}^W\left(|a_i|^2 + \eta(\vw_i)^2 \right).
\end{equation}
For example, if $\eta(\vw) = \|\vw\|_2$ we recover standard weight decay regularization. For the sampling theory given in \Cref{sec:mainthm}, we consider alternative weighting functions $\eta(\vw)$ defined in terms of penalties applied to the corresponding hidden-layer ReLU unit function $[\vw_i^\T\bm\gamma(\cdot)]_+$.

\section{Theory}
\label{sec:Theory}
Our main theoretical insight is to show that minimizers of the non-convex INR training problem \eqref{eq:opt_param_space1} coincide with minimizers of a \emph{convex} optimization problem defined over an infinite-dimensional space of measures. We use this equivalence to prove our main theorems, which identifies a sufficient number of Fourier samples needed to recover a continuous domain image realizable as shallow INR. Proofs of all results in this section are given in the Appendix.

\subsection{Generalized Weight Decay Induces Unit Sparsity}
Given an admissible weighting function $\eta$, and INR training width $W > 0$, the main optimization problem we consider in this work is
\begin{equation}\label{eq:opt_param_space}
    \min_{\theta \in \Theta_W}  \frac{1}{2}\sum_{i=1}^W\left(|a_i|^2 + \eta(\vw_i)^2\right)~~\mathrm{s.t.}~~\gF_\Omega f_\theta = \vy. \tag{\ensuremath{P_{\theta,W}}}
\end{equation}
First, we show that \eqref{eq:opt_param_space} is equivalent to optimizing a modified cost function over a restricted set of parameter vectors. In particular, let $\Theta'_{W} \subset \Theta_W$ denote the subset of parameter vectors $\theta = \left((a_i,\vw_i)\right)_{i=1}^W$ that satisfy the normalization constraint $\|\vw_i\|_2 = 1$ for all $i\in[W]$. Then we have the following equivalence:
\begin{myLem}\label{lem:ell2_to_ell1}
Let $\eta$ be any admissible weighting function. Then function space minimizers of \eqref{eq:opt_param_space}
coincide with function space minimizers of
\begin{equation}\label{eq:opt_finite_width_R}
\min_{\substack{\theta \in \Theta_{W}'}} \sum_{i=1}^W |a_i|\eta(\vw_i)~~\mathrm{s.t.}~~\gF_\Omega f_\theta = \vy.\tag{\ensuremath{P_{\theta,W}'}}
\end{equation}
\end{myLem}
The proof, given in \Cref{sec:app:proof_ell2_to_ell1}, follows by an argument similar to the proof of \cite[Lemma 1]{savarese2019infinite}, and relies on the positive 1-homogeneity of the ReLU activation and of the weighting function $\eta$. This result shows minimizing a generalized weight decay penalty implicitly minimizes a weighted $\ell^1$-norm on the outer-layer weights $a_i$ subject to a normalization constraint on the inner-layer weight vectors $\vw_i$.  Due to the sparsity-promoting property of the $\ell^1$-norm, intuitively this result shows that minimizing generalized weight decay penalty should promote functions that are realizable by a ``sparse'' single-hidden layer INR, i.e., an INR with few active units.

\subsection{Infinite-width Convexification}

One technical hurdle to characterizing minimizers of \eqref{eq:opt_finite_width_R} is that set of parameters $\theta \in \Theta_{W}'$ satisfying the constraint $\gF_\Omega f_\theta = \vy$ is non-convex. To sidestep this issue, we show how to reformulate \eqref{eq:opt_finite_width_R} as a convex optimization problem posed over a space of measures that represents the limit as the hidden-layer width $W\rightarrow \infty$. We note that similar convex formulations for single hidden-layer neural networks have been explored in \cite{bengio2005convex,bach2017breaking,ergen2021convex}.

First, fix any set of normalized INR parameters $\theta = \left((a_i,\vw_i)\right)_{i=1}^W\in \Theta_{W}'$. By linearity of $\gF_\Omega$ we have
\begin{equation}
\gF_\Omega f_\theta = \sum_{i=1}^W a_i \gF_\Omega[\vw_i^\T\bm\gamma(\cdot)]_+.
\end{equation}
Now we show the sum on the right above is expressible as the integration of a function against a discrete measure, i.e., a linear combination of Dirac deltas. In particular, given any admissible weight function $\eta(\vw)$, define the set
\begin{equation}
    \Peta := \{\vw \in \R^D : \|\vw\|_2 =1, \eta(\vw) > 0\},
\end{equation} 
which is an open subset of the $D-1$ dimensional unit sphere, $\mathbb{S}^{D-1} \subset \R^D$. Also, for any signed measure $\mu$ defined over $\Peta$, consider the linear operator $\gK_\Omega$ defined by
\begin{equation}
\gK_\Omega \mu := \int_{\Peta} \gF_\Omega[\vw^\T\bm\gamma(\cdot)]_+ d\mu(\vw).
\end{equation}
Then, for the measure $\mu^* = \sum_{i=1}^W a_i \delta_{\vw_i}$ where $\delta_\vw$ denotes a Dirac delta measure centered at $\vw \in \Peta$, we have $ \gK_\Omega \mu^* = \gF_\Omega f_\theta$. This shows the constraint $\gF_\Omega f_\theta = \vy$ in \eqref{eq:opt_finite_width_R} is equivalent to $\gK_\Omega \mu^* = \vy$.

Next, we define the $\eta$-weighted total variation norm $\|\mu\|_{TV,\eta}$ of any signed measure $\mu$ over $\Peta$ by\footnote{By continuity of $\eta$, there exists a constant $B$ such that $\eta(\vw) \leq B$ for all $\|\vw\|=1$, which ensures the integral above is well-defined.}
\[
\|\mu\|_{TV,\eta} := \int_{\Peta} \eta(\vw) d|\mu|(\vw),
\]
where $|\mu|$ is the total variation measure associated with $\mu$. For example, we have $|\mu^*| = \sum_{i=1}^W |a_i|\delta_{\vw_i}$, and so
$\|\mu^*\|_{TV,\eta} = \sum_{i=1}^W |a_i|\eta(\vw_i)$,
which coincides with the objective in \eqref{eq:opt_finite_width_R}. 

Therefore, letting $\gM_W(\Peta)$ denote the set of signed measures defined over $\Peta$ expressible as a weighted linear combination of at most $W$ Diracs, we see that \eqref{eq:opt_finite_width_R} is equivalent to the optimization problem
\begin{equation}\label{eq:opt_discrete_measures}
p^*_W = \min_{\mu \in \gM_W(\Peta)} \|\mu\|_{TV,\eta}~~s.t.~~\gK_\Omega \mu = \vy.\tag{\ensuremath{P_{\mu,W}}}
\end{equation}
In particular, the measure $\mu^* = \sum_{i=1}^W a_i \delta_{\vw_i}$ is a minimizer of \eqref{eq:opt_discrete_measures} if and only if the parameter vector $\theta^* = \left((a_i,\vw_i)\right)_{i=1}^W \in \Theta_{W}'$ is a minimizer of \eqref{eq:opt_finite_width_R}. Also, by \Cref{lem:ell2_to_ell1}, \eqref{eq:opt_finite_width_R} has the same function space minimizers as the original INR training problem \eqref{eq:opt_param_space}. This shows that we can describe all function space minimizers of \eqref{eq:opt_param_space} by identifying the minimizers of \eqref{eq:opt_discrete_measures}.

However, analyzing \eqref{eq:opt_discrete_measures} directly is still challenging since $\gM_W(\Peta)$ is not a closed vector space.
Instead, we pass to the larger Banach space $\gM(\Peta)$ of signed Radon measures defined over $\Peta$ with finite total variation norm, and consider the optimization problem
\begin{equation}\label{eq:opt_func_space}
    p^* = \min_{\mu \in \gM(\Peta)} \|\mu\|_{TV,\eta}~~s.t.~~\gK_\Omega \mu = \vy.\tag{\ensuremath{P_\mu}}
\end{equation}
This can be thought of as the ``infinite-width'' analogue of \eqref{eq:opt_discrete_measures}, or equivalently, \eqref{eq:opt_finite_width_R}. However, unlike these problems, \eqref{eq:opt_func_space} is a \emph{convex} optimization problem, albeit one posed over an infinite-dimensional space of measures.

Problem \eqref{eq:opt_func_space} also has a function space interpretation: for any $\mu \in \gM(\Peta)$, define
\[
f_{\mu}(\vx) := \int_{\Peta} [\vw^\T\bm\gamma(\vx)]_+ d\mu(\vw).
\]
Then $f_\mu$ can be thought of as ``infinite-width'' shallow INR, for which $\gF_\Omega f_\mu = \gK_\Omega \mu$. Similar to before, we call $f$ a function space minimizer of \eqref{eq:opt_func_space}, if $f = f_{\mu^*}$ where $\mu^*$ is any minimizer of \eqref{eq:opt_func_space}.

The existence of sparse minimizers of \eqref{eq:opt_func_space} has direct implications for \eqref{eq:opt_discrete_measures}. First, since $\gM_W(\Peta) \subset \gM(\Peta)$ we always have $p^* \leq p_W^*$. Additionally, if a minimizer $\mu^*$ of \eqref{eq:opt_func_space} is a finite linear combination of at most $W$ Diracs, i.e., $\mu^* \in \gM_W(\Peta)$, then $\mu^*$ must also be a minimizer of  \eqref{eq:opt_discrete_measures}, and so $p^* = p_W^*$. Our next result shows that, for sufficiently large training widths $W$, such a minimizer always exists assuming the set of frequencies used to define the Fourier features embedding contains a base set of frequencies. 

\begin{myLem}\label{lem:low_width_minimizers}
Suppose the set of frequencies $\Omega_0 \subset \Zd$ used to define the Fourier features embedding $\bm\gamma$ given in \eqref{eq:ffmap} is such that either $\ve_i \in \Omega_0$ or $-\ve_i \in \Omega_0$ for all $i=1,...,d$ where $\ve_i$ is the $i$th elementary basis vector. Assume $W \geq |\Omega|$. Then, given any $\vy \in \mathrm{Im}(\gF_\Omega)$, \eqref{eq:opt_func_space} has a minimizer $\mu^* \in \gM_W(\Peta)$.
\end{myLem}
The proof of \Cref{lem:low_width_minimizers}, which makes use of the abstract ``representer theorem'' of \cite{bredies2020sparsity}, is given in \Cref{sec:app:proof_lem2}.
This result implies that when $W \geq |\Omega|$ the minimums of \eqref{eq:opt_func_space} and \eqref{eq:opt_discrete_measures} must coincide, i.e.,  $p^* = p_W^*$, and a sparse measure $\mu^* \in \gM_W(\Peta)$ is a minimizer of \eqref{eq:opt_func_space} if and only if it is a minimizer of \eqref{eq:opt_discrete_measures}. Another immediate implication of \Cref{lem:low_width_minimizers} is that $W \geq |\Omega|$ is sufficient to guarantee the existence of feasible solutions for \eqref{eq:opt_param_space}, i.e., given arbitrary Fourier data $\vy \in \text{Im}(\F_\Omega)$ there is always a shallow INR $f_\theta$ of width $W\leq |\Omega|$ such that $\gF_\Omega f_\theta = \vy$.

\subsection{Exact Recovery of Shallow INRs}\label{sec:mainthm}
Now we give our main results, which identify a sufficient number of Fourier samples for which an image realizable as a shallow INR is exactly recoverable as the unique function space minimizer of the INR training problem \eqref{eq:opt_param_space}. 

Our first result focuses on recovery of a function $f$ realizable as \emph{width-1} shallow INR.

\begin{myThm}[Width-1 INR, general Fourier features, $d\geq 2$]\label{thm:main}
Consider the shallow INR architecture given in \eqref{eq:ftheta} with input dimension $d\geq 2$, where the set of frequencies $\Omega_0 \subset \Zd$ used to define the Fourier features embedding $\bm\gamma:\R^d\rightarrow\R^D$ given in \eqref{eq:ffmap} is such that $\|\vk\|_\infty \leq K_0$ for all $\vk \in \Omega_0$.
Let $f = a_1[\vw_1^\T\bm\gamma(\cdot)]_+$ for some $a_1 \in \R$ and $\vw_1 \in \R^D$. Suppose $\Omega\supseteq \{\vk \in \Zd : \|\vk\|_\infty \leq 3 K_0\}$, and let $\vy = \gF_\Omega f$. Then, for all training widths $W\geq 1$, $f$ is the unique function space minimizer of the INR training problem \eqref{eq:opt_param_space} with weighting function $\eta(\vw) = \|\gF_\Omega[\vw^\T\bm\gamma(\cdot)]_+\|_2$.
\end{myThm}

See \Cref{sec:app:thm1_proof} for the proof. This result shows that $O(K_0^d)$ low-pass Fourier samples are sufficient to uniquely recover a width-1 shallow INR with a Fourier features embedding having a maximum frequency $K_0$ (in an $\ell^\infty$ sense). In particular, if the Fourier sampling set is $\Omega = \{\vk \in \mathbb{Z}^d : \|\vk\|_\infty \leq K\}$, then $K \geq 3K_0$ is sufficient to ensure unique recovery.

Our second result focuses on the case of functions $f$ realizable as a width-$s$ shallow INR. However, we are unable to directly extend \Cref{thm:main} to this setting in all generality. Instead, we restrict our attention to input dimension $d=2$ (i.e., 2D images), and we consider an INR architecture defined in terms of a restricted set of Fourier features. In particular, we assume the Fourier features embedding $\bm\gamma:\mathbb{T}^2\rightarrow\R^4$ has the form
\begin{equation}\label{eq:gamma_restricted}
\bm\gamma(\vx) = \sqrt{2}[\cos(2\pi x_1),\sin(2\pi x_1), \cos(2\pi x_2),\sin(2\pi x_2)]^\T ~\text{for all}~\vx = [x_1,x_2]^\T \in \torus^2.
\end{equation}
Note here that the constant feature $1$ is omitted from $\bm\gamma$. 
Also, we assume that $f$ is a \emph{positive linear combination} of rectified trigonometric polynomials of the form $[\vw^\top \bm\gamma(\cdot)]_+$. Finally, for technical reasons, we will assume that none of the weight vectors $\vw\in\R^4$ belong to the set
\begin{equation}\label{eq:Vset}
V = \{\vw \in \R^4: w_1 = w_2 = 0,~\text{or}~w_3=w_4 = 0,~\text{or}~w_1^2 + w_2^2 = w_3^2+w_4^2\}.
\end{equation}
Observe that $V$ is a measure-zero subset of $\R^4$, i.e., our result still holds for almost all weight vectors $\vw \in \R^4$.

\begin{myThm}[Width-$s$ INR, restricted Fourier features, $d=2$]\label{thm:main2}
Consider the shallow INR architecture given in \eqref{eq:ftheta} with input dimension $d=2$ and restricted Fourier features embedding $\bm\gamma$ given in \eqref{eq:gamma_restricted}.
Let $f = \sum_{i=1}^s a_i[\vw_i^\T\bm\gamma(\cdot)]_+$ where $a_i > 0$ and $\vw_i \in \R^4/V$ for all $i\in[s]$. Suppose $\Omega\supseteq \{\vk \in \mathbb{Z}^2 : \|\vk\|_1 \leq 2s\}$, and let $\vy = \gF_\Omega f$. Then, for all training widths $W\geq s$, $f$ is the unique function space minimizer of the INR training problem \eqref{eq:opt_param_space} with weighting function $\eta(\vw) = \int_{\torus^2}[\vw^\T\bm\gamma(\vx)]_+ d\vx$.
\end{myThm}

See \Cref{sec:app:thm2_proof} for the proof. This result shows that $O(s^2)$ low-pass Fourier samples are sufficient to uniquely recover a 2D image realized by a width-$s$ shallow INR having restricted Fourier features. In particular, if the Fourier sampling set is $\Omega = \{\vk \in \mathbb{Z}^2 : \|\vk\|_\infty \leq K\}$, then $K \geq 2s$ is sufficient to ensure unique recovery.

The main idea of the proofs of \Cref{thm:main} and \Cref{thm:main2} is to pass to the convex dual of \eqref{eq:opt_func_space}, which reduces to a semi-infinite program, i.e., a finite dimensional linear program with infinitely many constraints. Similar to previous works that have focused on super-resolution of sparse measures \cite{candes2014towards,poon2019multidimensional,eftekhari2021stable}, we prove unique recovery is equivalent to the construction of a particular dual optimal solution known as a \emph{dual certificate}; see \Cref{app:dual} for more details. 

Note that both these results rely on a generalized form of weight decay regularization that differs from standard weight decay: standard weight decay corresponds to a weighting function $\eta(\vw) = \|\vw\|_2$, while in \Cref{thm:main} we assume $\eta(\vw) = \|\gF_\Omega[\vw^\T\bm\gamma(\cdot)]_+\|_2$ and in \Cref{thm:main2} we assume $\eta(\vw) = \int_{\torus^2}[\vw^\T\bm\gamma(\vx)]_+ d\vx$. Ultimately, these generalized weighting functions arise as a consequence of our proof techniques. However, their use as regularizers is still intuitive: they can be seen as (approximately) penalizing the $L^p$-norm ($p=1$ or $2$) contribution made by each unit $[\vw_i^\T\bm\gamma(\cdot)]_+$ to the overall image. This is immediately apparent in the case of $\eta(\vw) = \int_{\torus^2}[\vw^\T\bm\gamma(\vx)]_+ d\vx = \|[\vw^\T\bm\gamma(\cdot)]_+\|_{L^1(\torus^d)}$. And for $\eta(\vw) = \|\gF_\Omega[\vw^\T\bm\gamma(\cdot)]_+\|_2$, by Parseval's Theorem, we have $\eta(\vw) \approx \|[\vw^\T\bm\gamma(\cdot)]_+\|_{L^2(\Td)}$ assuming the low-pass frequency set $\Omega$ is sufficiently large.

Despite strong theoretical guarantees and possible benefits in terms of interpretability, these more general forms of weight decay are more complicated and computationally expensive to implement than standard weight decay, which may limit their practical utility. However, we conjecture that recovery guarantees similar to \Cref{thm:main} and \Cref{thm:main2} can be extended to the simpler case of standard weight decay regularization (i.e., weighting function $\eta(\vw)=\|\vw\|_2$), though doing so may require new proof techniques. This conjecture is verified empirically in the next section. 

\section{Experiments}\label{sec:Experiments}
Here we present a series of experiments to validate our sampling theory and to illustrate the effects of training INRs with (generalized) weight decay regularization in super-resolution recovery of 2D images. First, to validate the sampling guarantees given in Theorems 1 and 2, we analyze the probability of achieving exact recovery of random continuous-domain images realizable by shallow INRs from their low pass Fourier coefficients. Next, we illustrate the performance of shallow INRs in super-resolution recovery of various continuous-domain phantom images; additional experiments investigate the sparsity-inducing effects of (generalized) weight decay regularization. Finally, going beyond the theory developed in this work, we empirically investigate the role of depth in Fourier features INRs, demonstrating that deeper INRs trained with weight decay regularization achieve improved reconstruction on piecewise constant phantom images relative to shallow INRs, and compare the results with other INR architectures besides ReLU networks with Fourier features layers.

\paragraph{Implementation Details} In all our experiments, to approximate the continuous domain Fourier sampling operator $\gF_\Omega$, we implement a discretized version $\tilde{\gF}_\Omega$ using the discrete Fourier transform. In particular, given a continuous domain image $f:[0,1]^2\rightarrow \R$, we define $\tilde{\gF}_\Omega f = \mF_\Omega\mathcal{E}f$ where $\mathcal{E}f$ is a 2D pixel array obtained by evaluating $f$ on an $N \times N$ uniform grid within $[0,1]^2$, and $\mF_\Omega$ is the $N^2\times N^2$ 2D discrete Fourier transform matrix restricted to rows corresponding to the frequency set $\Omega$; we set $N=1024$. Additionally, the weight functions $\eta$ given in \Cref{thm:main} and \Cref{thm:main2} used to define generalized weight decay penalties are implemented via the approximations $\tilde{\eta}(\vw) = \|\tilde{\gF}_\Omega [\vw^\T \bm\gamma(\cdot)]_+\|_2$ and $\tilde{\eta}(\vw) = \frac{1}{N^2}\|\mathcal{E}[\vw^\T \bm\gamma(\cdot)]_+\|_1$, respectively. All experiments were run in the PyTorch numerical computing environment using a single NVIDIA V100 GPU with 32 GB RAM. Additional implementation details are described in the Supplementary Materials \ref{sec:supp:exp_details}. Code to reproduce all experiments is available at \url{https://github.com/gregongie/super_inrs}. 

\subsection{Super-Resolution with Shallow INRs}

\subsubsection{Exact Recovery Experiments}\label{sec:exact_recovery}
To validate the exact recovery guarantees  \Cref{thm:main} and \Cref{thm:main2}, we consider a ``student-teacher'' setup. Specifically, we generate random continuous domain images realizable as a ``teacher'' shallow INR of varying widths $W$ (\Cref{fig:RTP} shows an example for $W=3$), and investigate the maximum sampling frequency $K$ (in an $\ell^\infty$ sense) needed to recover the image by solving the optimization problem \eqref{eq:opt_param_space1} using a ``student'' shallow INR with the same architecture but larger width. As an attempt to solve the equality constrained problem \eqref{eq:opt_param_space1} we apply the Augmented Lagrangian (AL) method (see, e.g., \cite[Ch 7]{nocedal1999numerical}); see \Cref{sec:app:AL} for a detailed algorithmic description.

\begin{figure}[htbp!]
    \centering    
    \includegraphics[width=0.75\columnwidth]{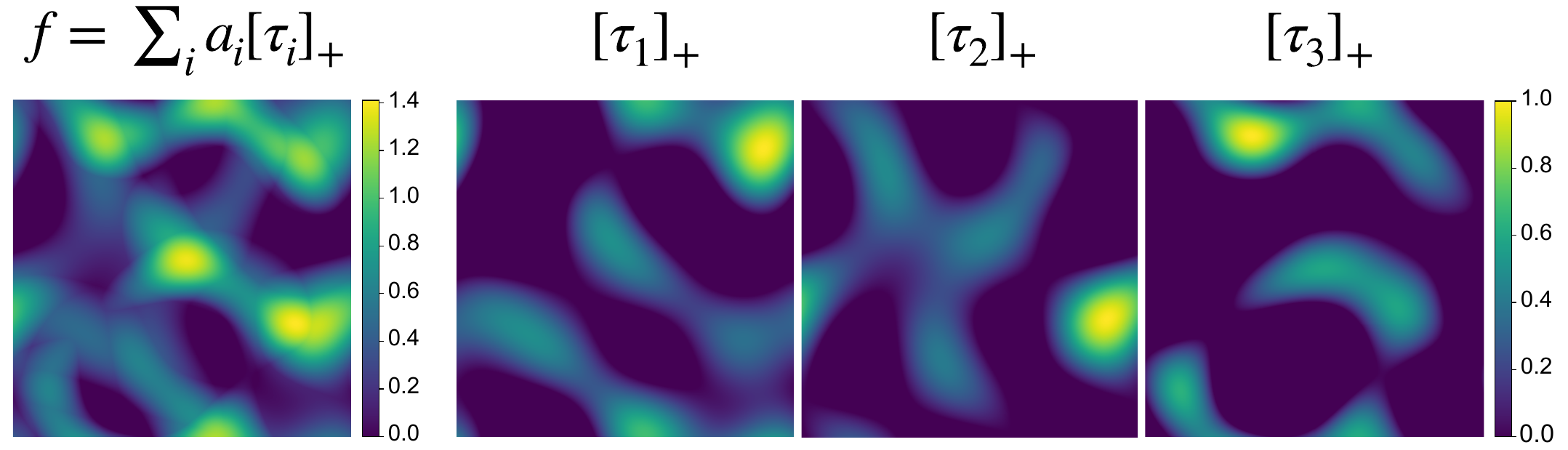}\\[1em]
    \includegraphics[width=0.75\columnwidth]{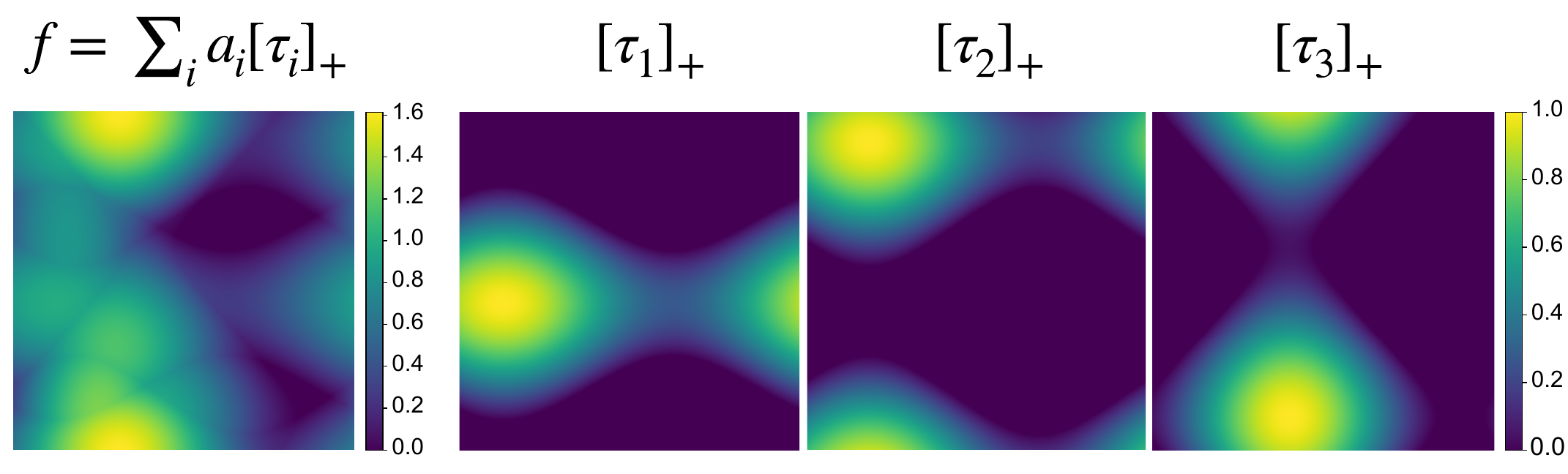}\\[1em]    
    \caption{\footnotesize Example of a random width 3 ``teacher'' INR used in our exact recovery experiments to validate \Cref{thm:main} (top panel) and \Cref{thm:main2} (bottom panel). The image $f$ is a weighted sum of the three randomly generated rectified trigonometric polynomials $[\tau_i]_+ =[\vw_i^\T \bm\gamma(\cdot)]_+$ shown on the right. Note that the two settings consider INRs defined with Fourier features embeddings  $\bm\gamma(\cdot)$ having different frequencies.}
    \label{fig:RTP}
\end{figure}
 
We consider two settings, matching the settings of \Cref{thm:main} and \Cref{thm:main2}, which differ in terms of the Fourier features used and the form of generalized weight decay penalty used. In the setting of \Cref{thm:main}, we use the general Fourier features embedding \eqref{eq:ffmap} with maximum frequency $K_0=2$ resulting in $D=25$ features, while in the setting of \Cref{thm:main2}, we use the restricted Fourier features embedding \eqref{eq:gamma_restricted} consisting of $D=4$ features only. The generalized weight decay penalties defined in \Cref{thm:main} and \Cref{thm:main2} we call \emph{modified WD-I} and \emph{modified WD-II}, respectively. We additionally compare against standard weight decay (i.e., the squared $\ell^2$-norm of all trainable parameters), which we call \emph{standard WD}.
In both settings, we vary both the maximum sampling frequency cutoff $K$ and the teacher width $W$, and run $10$ random trials for each $(K,W)$ pair to approximate the probability of exact recovery. A random teacher network is defined by drawing inner-layer weight vectors $\vw_1,...,\vw_W$ uniformly at random from the sphere, and drawing outer-layer weights $a_1,...,a_W$ uniformly at random from the interval $[1,5]$. The student network width is fixed at $100$ in all cases. We say the recovery is ``exact'' if the pixel-wise mean-squared error (MSE) between the ground truth image and student INR rasterized onto a $1024\times 1024$ pixel grid drops below $10^{-9}$ at any point during training. More details for these experiments are provided in Supplementary Materials \ref{sec:supp:exp_details}.

\Cref{fig:exact_rec} shows the results of these experiments as probability tables. In the probability tables corresponding to  \Cref{thm:main}, for standard WD regularization, we see that a linear trend between the teacher width $W$ and the minimal sampling cutoff $K$ needed to achieve exact recovery with high probability up to width 5, while exact recovery generally fails for teacher widths greater than 5. For modified WD-I regularization, although the linear trend is less pronounced, exact recovery is achieved in all random trials when the teacher width is 1 (consistent with \Cref{thm:main}) and remains attainable in many cases for teacher widths greater than one, but generally fails for widths greater than 5. Note that, for both regularization methods, the parameter space optimization problem is non-convex, and so convergence of the AL method to a global minimizer is not guaranteed. Therefore, it is unclear if the failure at achieving exact recovery for large widths in the setting of Theorem 1 is due to a difficulty in reaching a global minimizer with the proposed optimization method, or a more fundamental lack of identifiability.

\begin{figure*}[htbp!]

  \centering     \includegraphics[width=\textwidth]{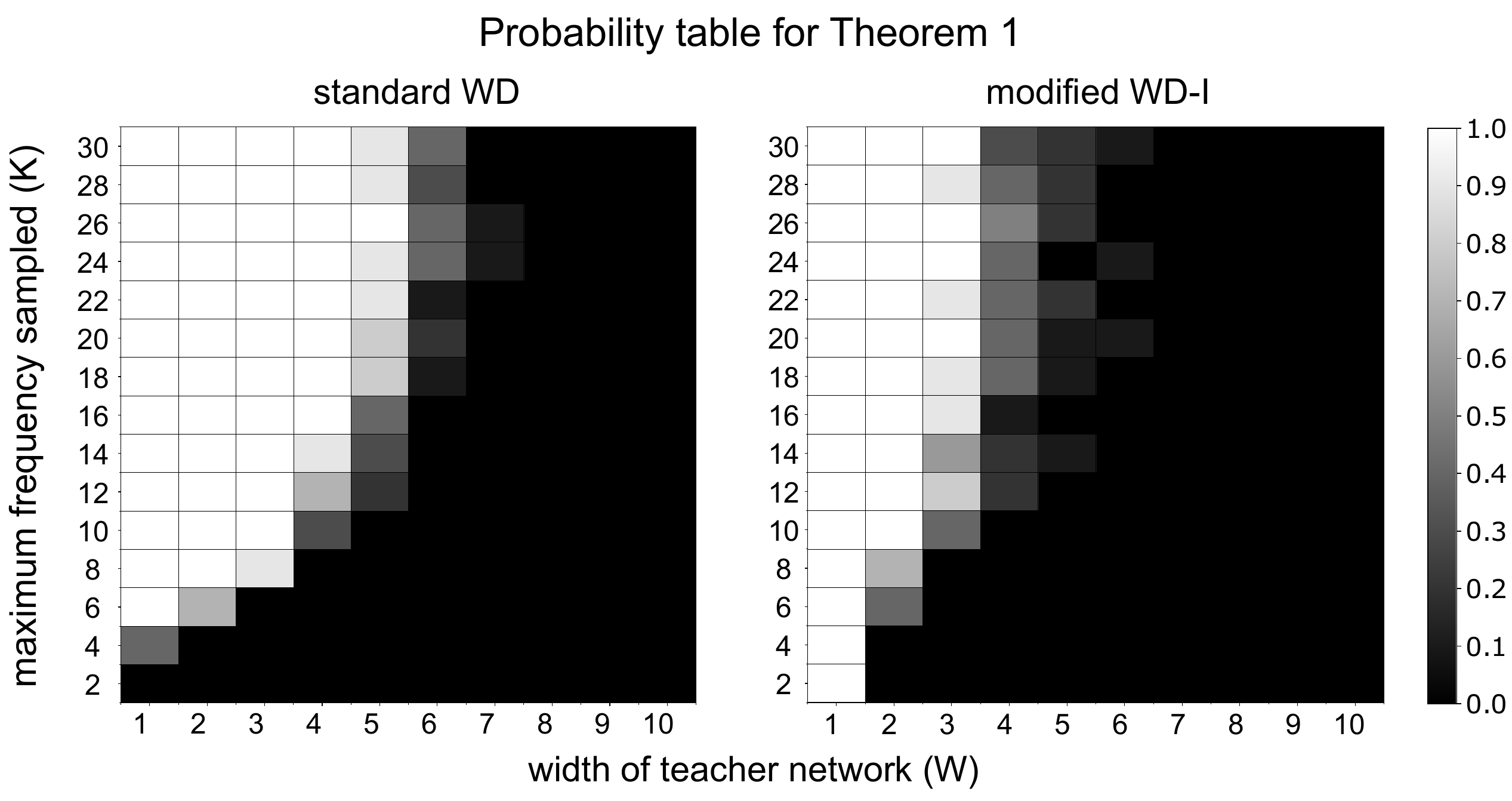}\\[1em]
  \centering     \includegraphics[width=\textwidth]{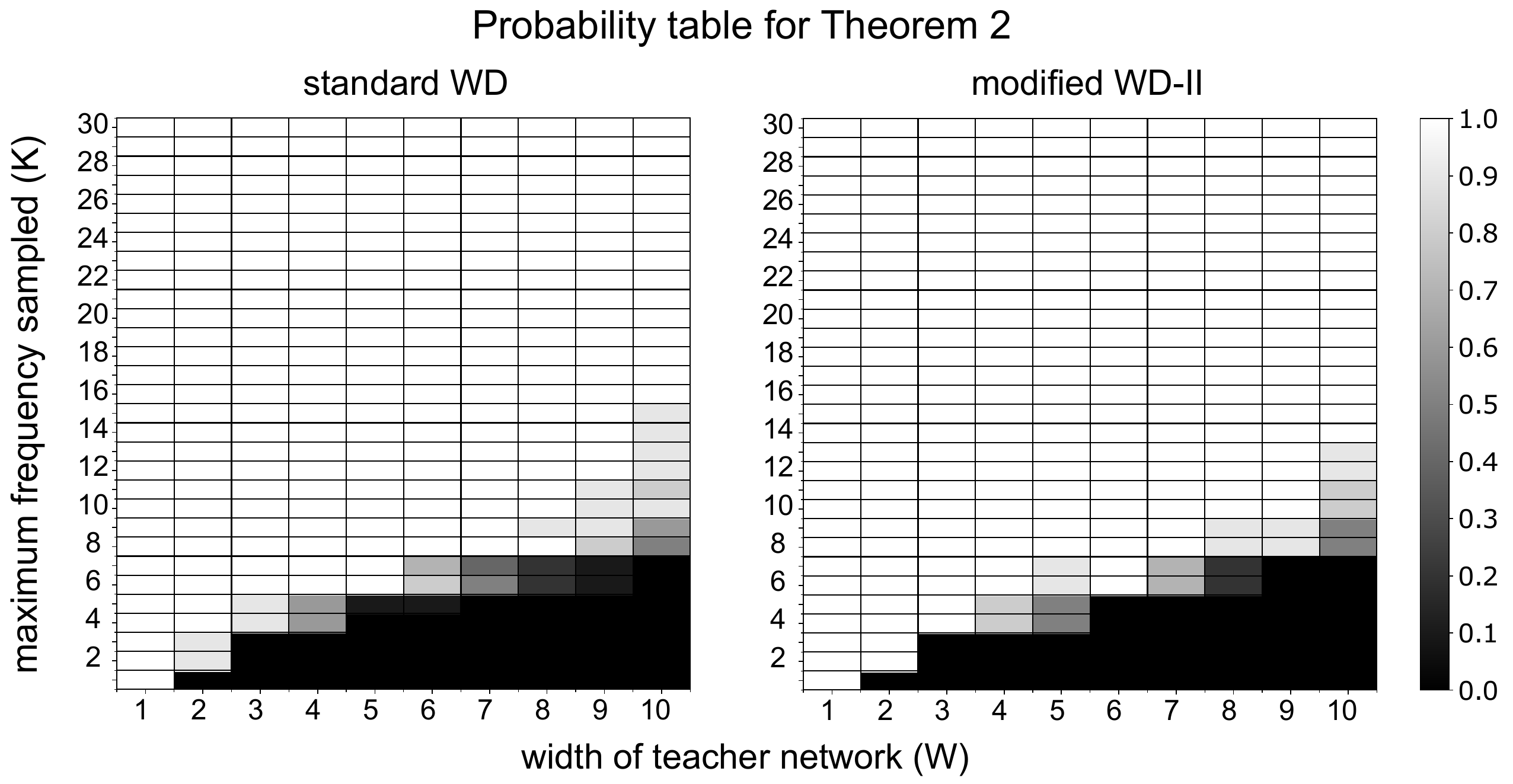}\\
    \caption{\footnotesize Empirical probability of exactly recovering an random image representable by a single hidden-layer INR of width-$W$ from low-pass Fourier samples with maximum sampling frequency $K$ by solving \eqref{eq:opt_param_space}. The left panel corresponds to the setting of \Cref{thm:main}, comparing standard weight decay (standard WD) with modified weight decay regularization proposed for  \Cref{thm:main} (modified WD-I) and the right panel corresponds to the setting of \Cref{thm:main2}, comparing standard weight decay (standard WD)  with the modified weight decay proposed for \Cref{thm:main2} (modified WD-II). Empirical probabilities are estimated from 10 random trials.}
    \label{fig:exact_rec}
\end{figure*}

In the probability tables corresponding to \Cref{thm:main2}, the linear transition is clear for both standard WD and modified WD-II regularization. This is consistent with \Cref{thm:main2}, which proves that $K = O(W)$ Fourier samples are sufficient for exact recovery. Also, in this setting, exact recovery is achieved for teacher widths greater than 5. We conjecture that this is because restricted Fourier features used in this setting result in a better conditioned problem.

This set of experiments suggests that standard weight decay regularization performs similar or better at achieving exact recovery of random phantom images relative to the modified forms of weight decay. However, there may be practical benefits to using modified weight decay regularization over standard weight decay regularization in achieving \emph{approximate} recovery, as we show in the next set of experiments with more realistic phantom images.

\subsubsection{Recovery of Continuous Domain Phantom Images}
In \Cref{fig:dot_shepp} and \Cref{fig:pwc_pws}, we illustrate the use of shallow INRs for the super-resolution recovery of four continuous domain phantoms whose Fourier coefficients are known exactly: (1) \textsf{DOT}, a ``dot phantom'' image, created for this study, that contains randomly generated disk-shaped features and is exactly realizable by a shallow INR architecture with a width of $W \geq 50$ and maximum Fourier features frequency $K_0 \geq 8$, (2) \textsf{SL}, the Shepp-Logan phantom  \cite{shepp1974fourier}, (3) \textsf{PWC-BRAIN} a piecewise constant MRI ``brain phantom'' introduced in \cite{guerquin2011realistic}, and (4) \textsf{PWS-BRAIN} a piecewise smooth version of the MRI brain phantom obtained by multiplying \textsf{PWC-BRAIN} by smooth function. Note that the  \textsf{SL},  \textsf{PWC-BRAIN}, and  \textsf{PWS-BRAIN} phantoms are all discontinuous functions, so they are not exactly realizable with the INR architecture considered in this work, whose output must be a continuous function. We also note all these four phantoms are not bandlimited, such that a simple reconstruction obtained by zero-filling in the Fourier domain and applying the IFFT results in significant ringing artifacts (see the second column in each panel of \Cref{fig:dot_shepp} and \Cref{fig:pwc_pws}).

For each phantom, the maximum sampling frequency $K$, maximum Fourier features frequency $K_0$, and the width of the INR $W$, are set as shown in \Cref{tab:sample_table}.
\begin{table}[ht!]
    \centering
    \begin{tabular}{|c|c|c|c|}
        \hline
        Phantom & $K$ & $K_0$ &  $W$ \\ 
        \hline
        \textsf{DOT}  & 32  & 10 & 100\\ 
        \hline
        \textsf{SL} &  48  & 10 & 100\\ 
         \hline
        \textsf{PWC-BRAIN} & 64 & 20 & 500\\ 
        \hline  
        \textsf{PWS-BRAIN} & 64 & 20 & 500\\ 
        \hline
    \end{tabular}
    \caption{Parameter settings for recovery of all phantoms by fitting a shallow INR}
    \label{tab:sample_table}
\end{table}
For this set of experiments, we fit the INR by minimizing the regularized least squares formulation \eqref{eq:lsfit} using the Adam optimizer for $40,000$ iterations with a learning rate of $1\times 10^{-3}$ followed by an additional $10,000$ iterations with a learning rate of $1\times 10^{-4}$. The weight decay hyperparameter $\lambda$ was selected via grid search to minimize the pixel-wise mean-squared error (MSE) between the ground truth image and INR rasterized onto a $1024\times 1024$ pixel grid.

Unlike the exact recovery experiments, where standard weight decay regularization seems to perform better in achieving exact recovery, for all phantoms we find that using a shallow INR trained with modified weight decay regularization results in lower image MSE and fewer artifacts in the reconstructions compared to standard weight decay regularization. In particular, the MSE is a magnitude lower on the \textsf{DOT} phantom with modified weight decay and the reconstruction has far fewer ringing artifacts. A more modest improvement is shown on the other phantoms, where training with modified weight decay yields a slight improvement in MSE, but some ringing artifacts are still visible. 

\begin{figure}[htbp!]
    \centering
    \includegraphics[width=0.85\textwidth]{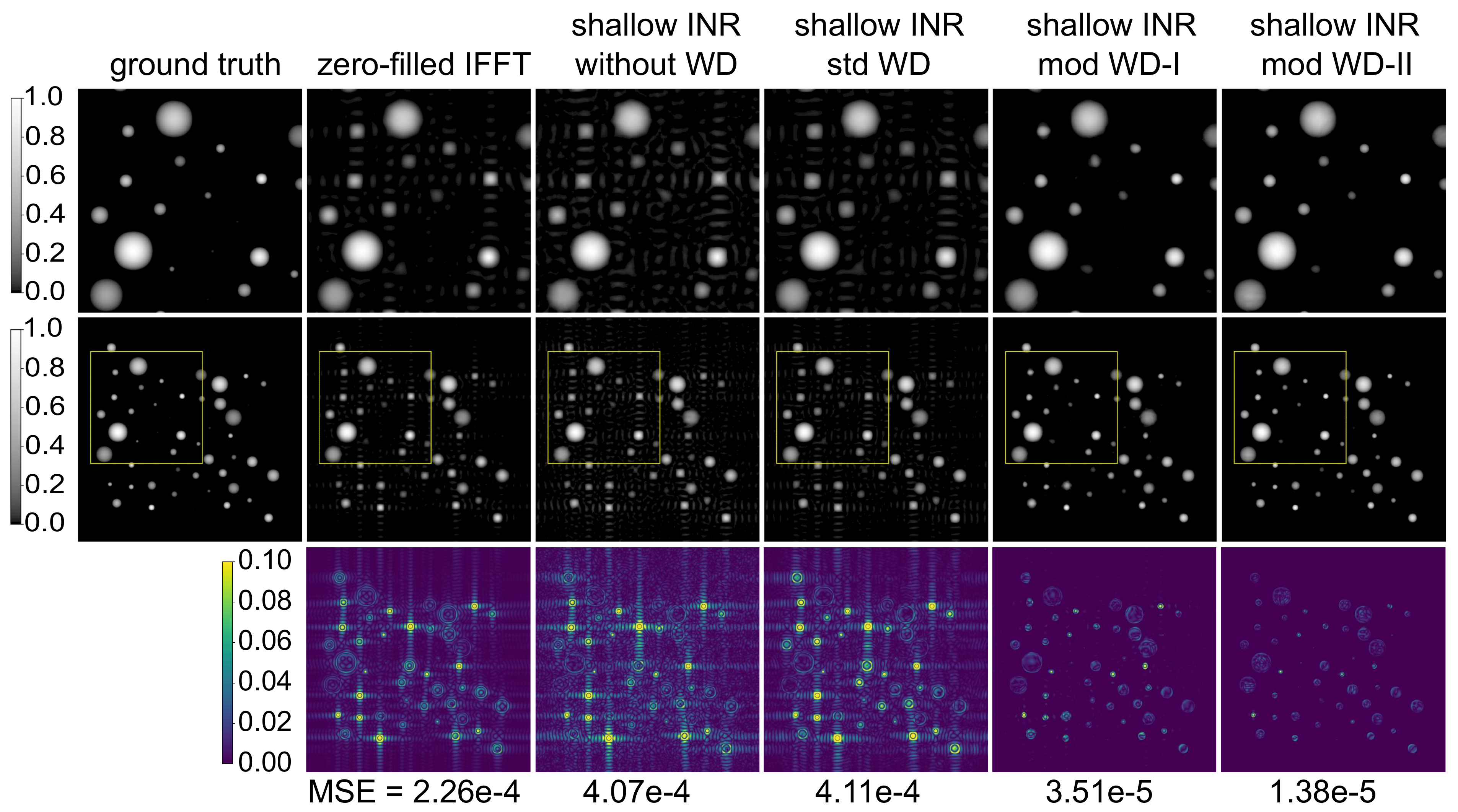}\vspace{1em}
    \includegraphics[width=0.85\textwidth]{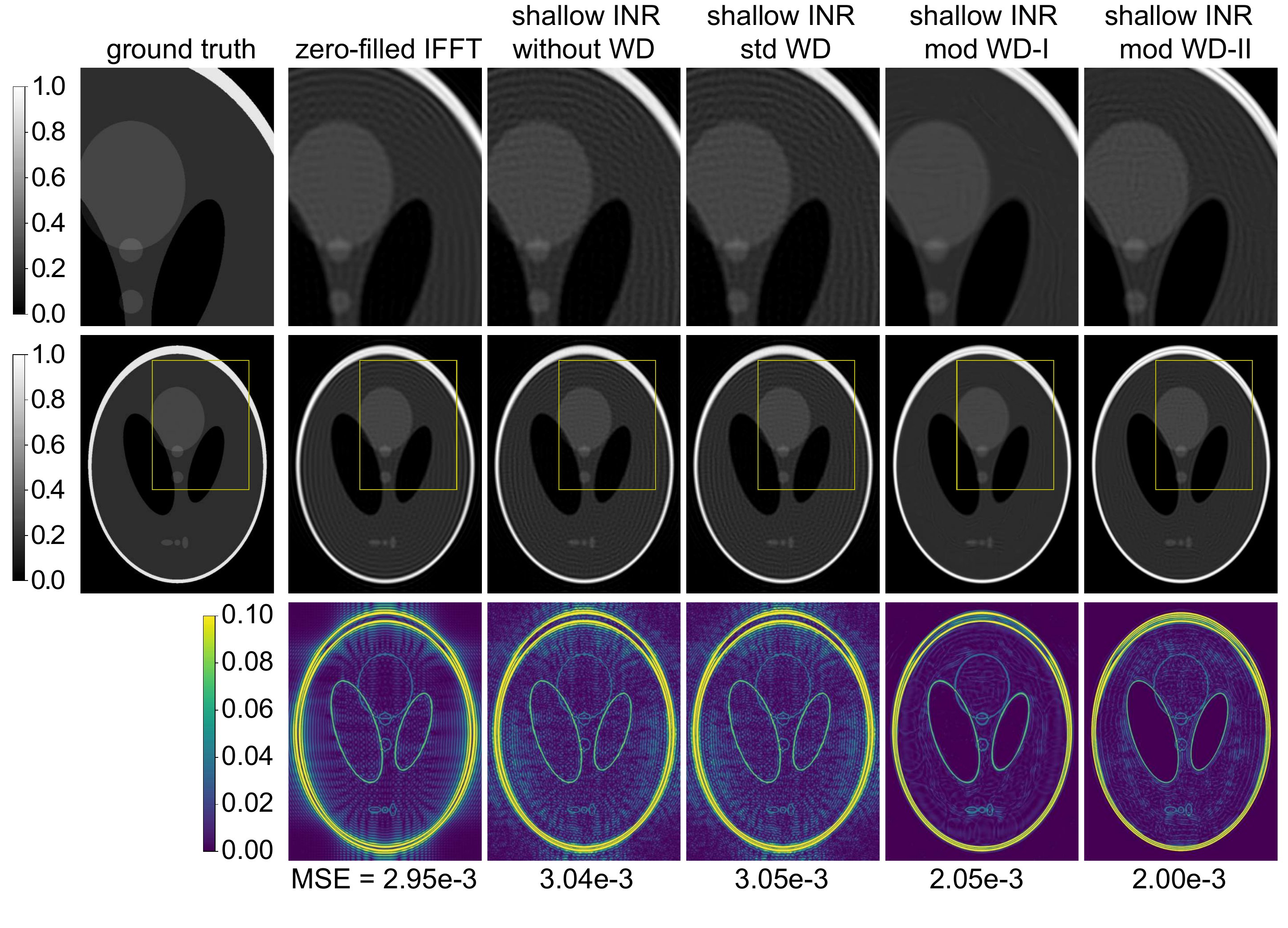}
    \caption{\footnotesize \textbf{Recovery of continuous domain \textsf{DOT} and \textsf{SL} phantoms from low-pass Fourier samples.} 
    We compare the zero-filled IFFT reconstruction with a shallow INR trained without regularization (without WD), with standard weight decay regularization (shallow INR, std WD), the modified weight decay regularization proposed for  \Cref{thm:main} (shallow, mod WD-I), and the modified weight decay regularization proposed for  \Cref{thm:main2} (shallow, mod WD-II). The bottom row shows the absolute value of the difference with the ground truth.}
    \label{fig:dot_shepp}
\end{figure}

\begin{figure}[htbp!]
    \centering
    \includegraphics[width=0.8\textwidth]{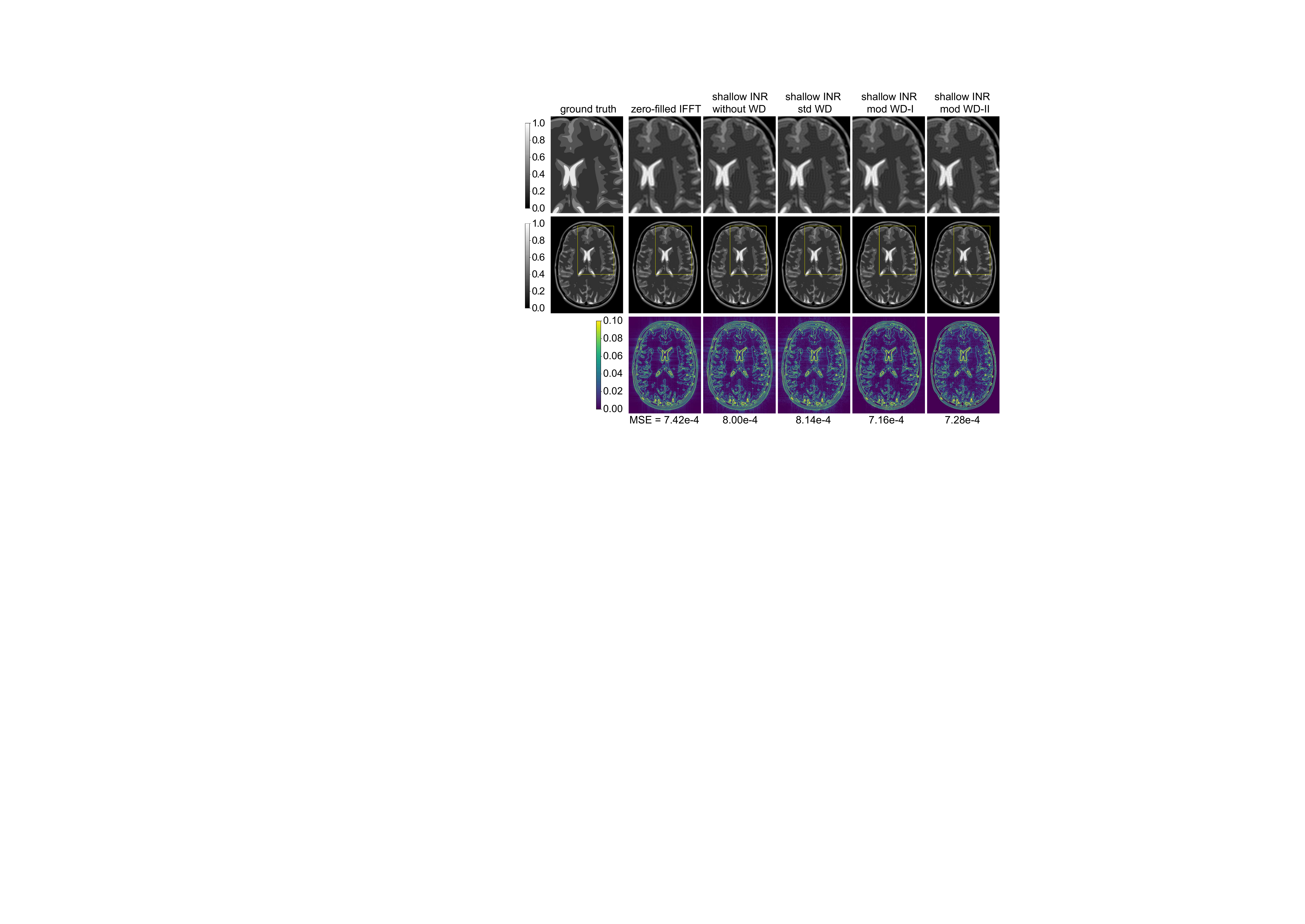}\vspace{1em}
    \includegraphics[width=0.8\textwidth]{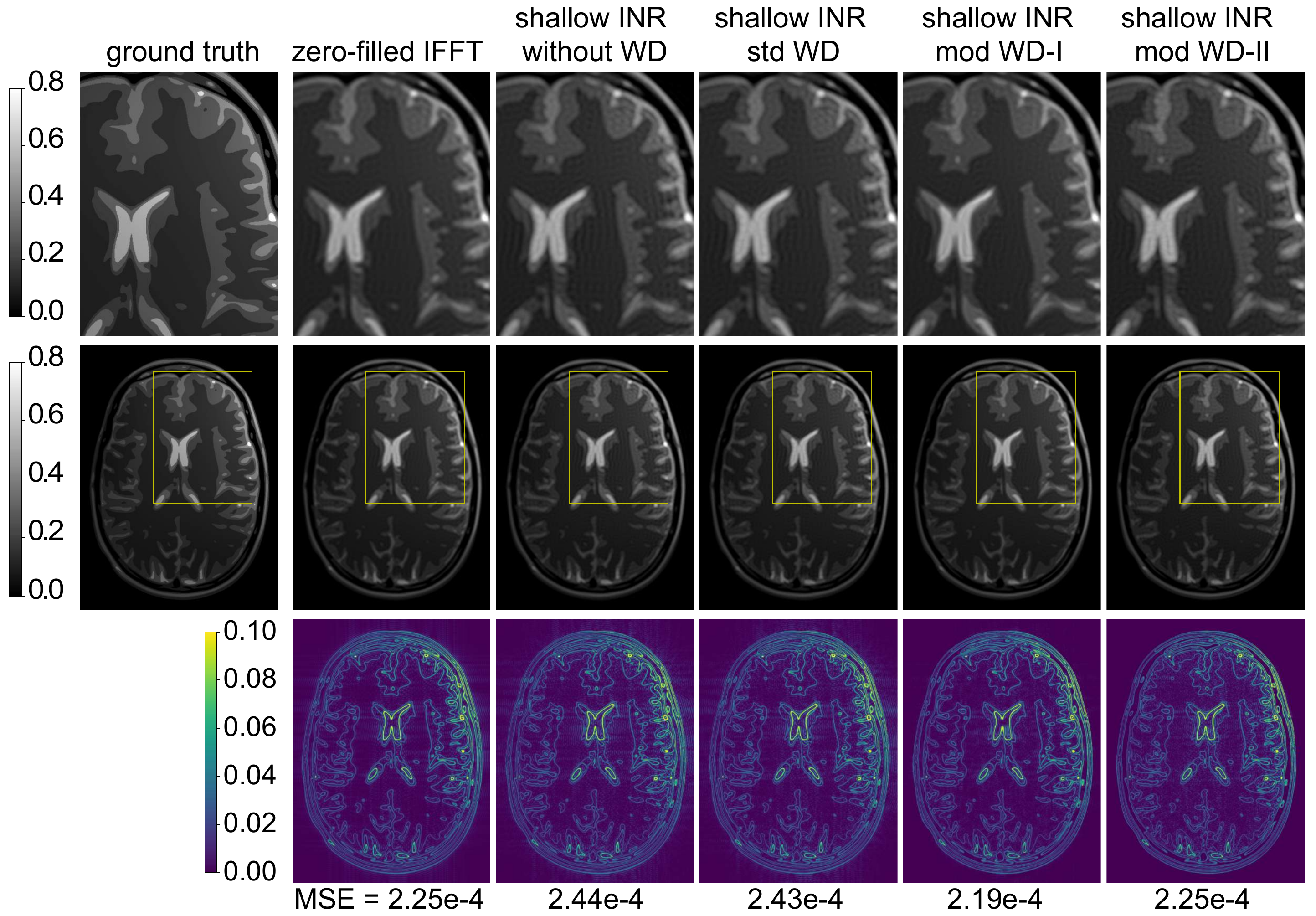}
    \caption{\footnotesize \textbf{Recovery of continuous domain \textsf{PWC-BRAIN} and \textsf{PWS-BRAIN} phantoms from low-pass Fourier samples.} 
    We compare the zero-filled IFFT reconstruction with a shallow INR trained without regularization (without WD), with standard weight decay regularization (shallow INR, std WD), the modified weight decay regularization proposed for  \Cref{thm:main} (shallow, mod WD-I), and the modified weight decay regularization proposed for  \Cref{thm:main2} (shallow, mod WD-II). The bottom row shows the absolute value of the difference with the ground truth.}
    \label{fig:pwc_pws}
\end{figure}

To illustrate the impact of varying the maximum Fourier features frequency, $K_0$, in the Supplementary Materials \ref{sec:supp:figures} we show the reconstruction of the \textsf{SL} phantom with $K_0 = 5, 10, 15, 20$, trained with standard WD, modified WD-I, and modified WD-II. Setting $K_0=10$ resulted in reconstructions with the lowest image MSE for both modified WD-I and modified WD-II.

Additionally, in \Cref{fig:std_mod_wd_10K}, we demonstrate the impact of weight decay regularization strength on the reconstructed \textsf{PWC-BRAIN} phantom. First, we see that including weight decay regularization (i.e., setting $\lambda>0$) yields reconstructions with lower MSE and fewer visual ringing artifacts than training without regularization ($\lambda = 0$). However, using too large of a $\lambda$ value yields an over-smoothed reconstruction. In \Cref{fig:std_mod_wd_10K}, we also investigate the impact of regularization strength on the sparsity of the trained INR weights. In particular, given a ReLU unit of the form $a[\vw^\T\bm\gamma(\cdot)]_+$, we define its unit size to be $|a|\eta(\vw)$, where $\eta$ is the weighting function used to define the generalized weight decay regularization term in \eqref{eq:param_space_reg}. By \Cref{lem:ell2_to_ell1}, the sum of the unit sizes should implicitly be minimized when using generalized weight decay. From the histograms of the unit sizes shown in \Cref{fig:std_mod_wd_10K}, we see a large spread in unit sizes when not using regularization ($\lambda =0$), while as we increase the regularization strength, the number of units with size close to zero increases markedly. The effect is especially pronounced for large $\lambda$ in the modified weight decay methods, where nearly all unit sizes are clustered close to zero except for a few large outliers, indicating the trained INR has very low effective width. 

\begin{figure*}[htbp!]
    \centering    \includegraphics[width=0.57\textwidth]{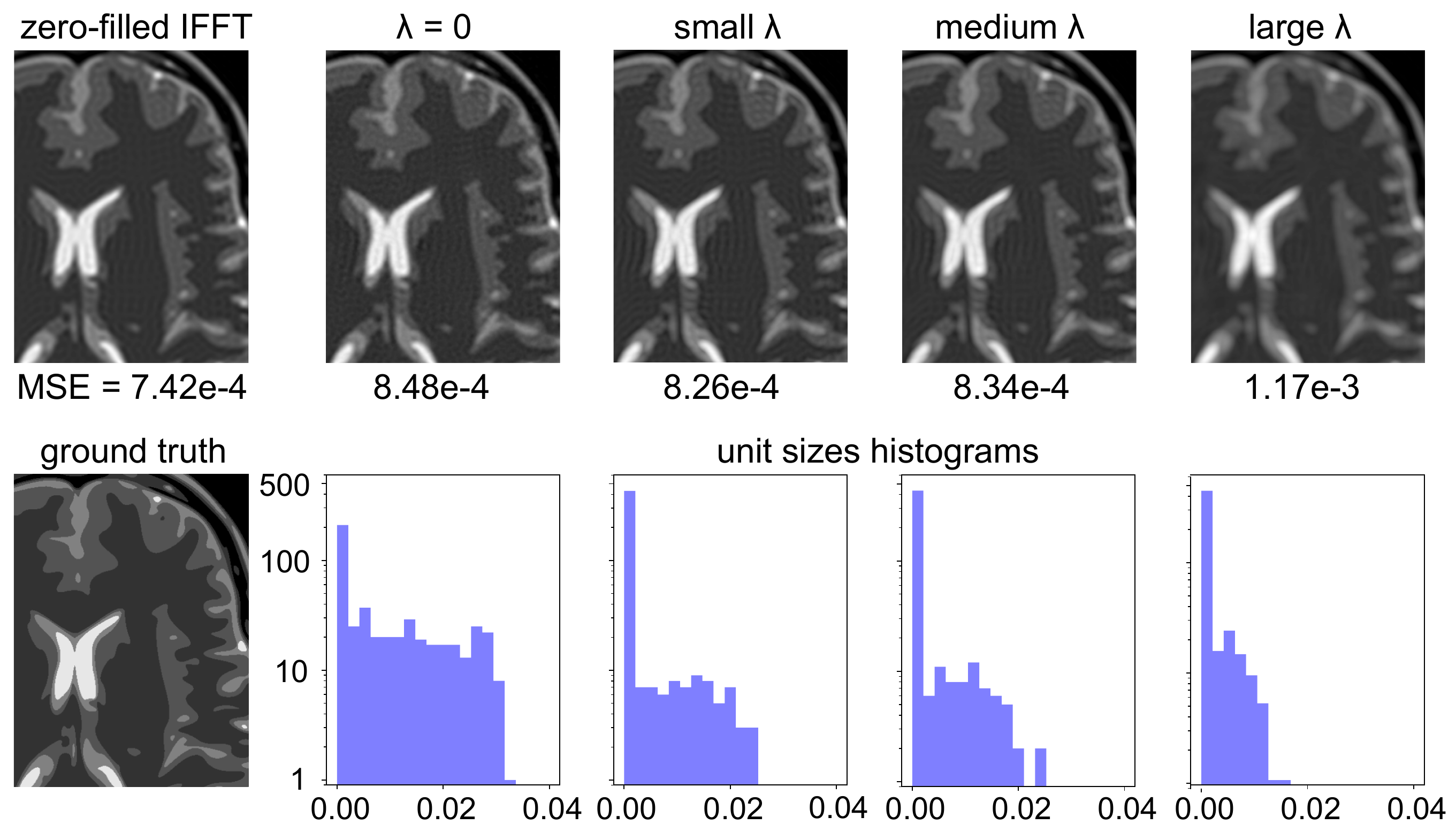} \\
    (a) standard weight decay regularization\\[1em]
\includegraphics[width=0.57\textwidth]{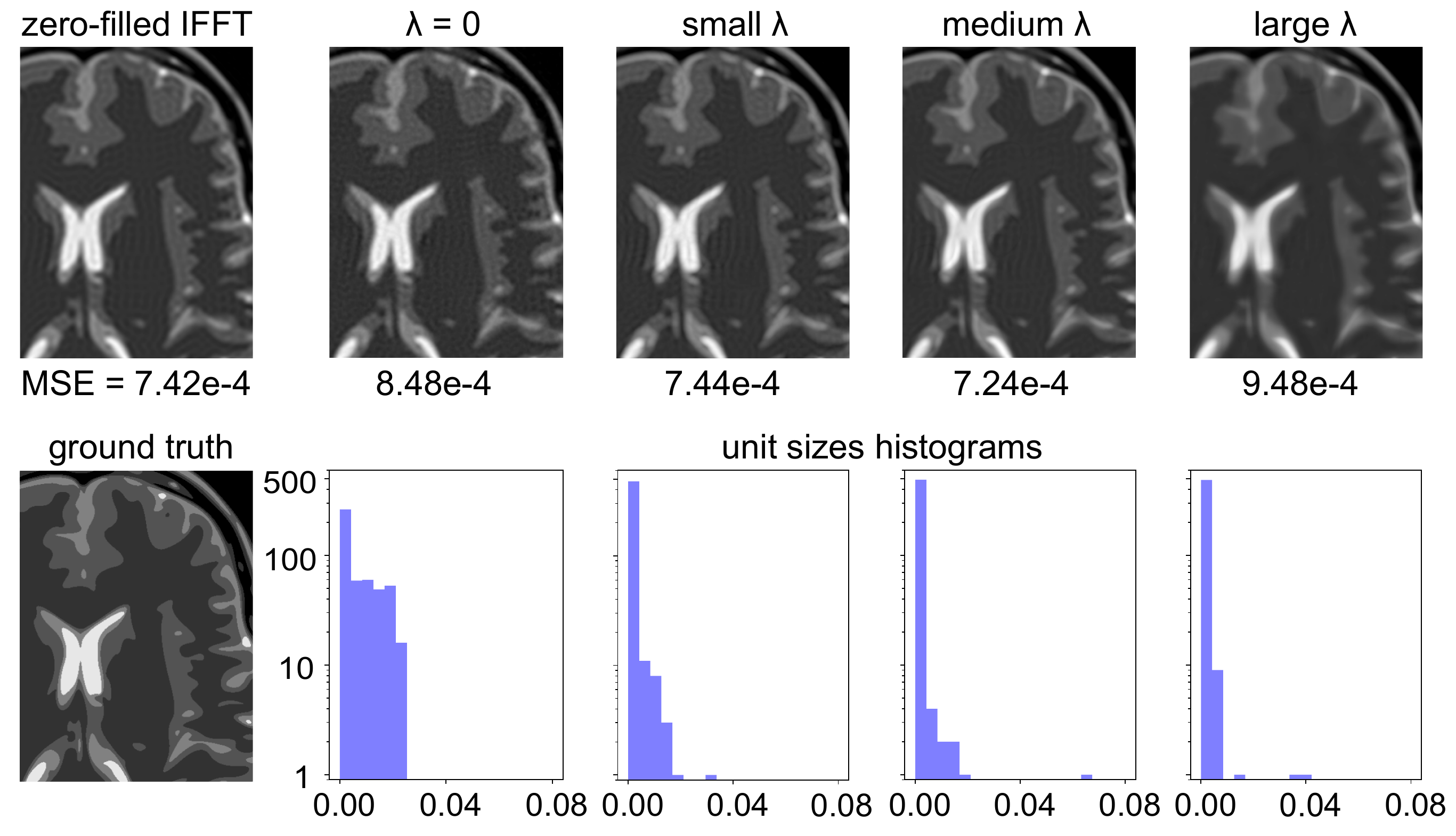}\\
   (b) modified weight decay regularization proposed in \Cref{thm:main}\\[1em]
   \includegraphics[width=0.57\textwidth]{figs/brain_2nd_mod_hist.pdf}\\
   (c) modified weight decay regularization proposed in  \Cref{thm:main2}
    \caption{\footnotesize \textbf{Effect of regularization on the \textsf{PWC BRAIN} phantom reconstruction and unit sizes of a trained depth 2 INR.} In the first row of each panel, we compare the zero-filled IFFT reconstruction with the reconstructions obtained with a shallow INR trained with standard/modified weight decay regularization for different values of regularization strength $\lambda$. The histograms in the second row of each panel represent the distribution of unit sizes $\{|a_i|\eta(\vw_i)\}_{i=1}^{500}$ of the trained INRs across different values of $\lambda$. As $\lambda$ increases, the unit sizes cluster at zero, indicating the trained INR is sparse in the sense that it has a few active units. (Note: histograms are plotted on a log scale.)}
    \label{fig:std_mod_wd_10K}
\end{figure*}

In \Cref{fig:brain_active_units_normalized}, we visualize the four largest active units from the unit size histograms presented in \Cref{fig:std_mod_wd_10K} for both standard and modified weight decay regularization applied to the \textsf{PWC-BRAIN} phantom. As predicted by our theory, applying stronger weight decay regularization (i.e., larger $\lambda$) results in a sparse INR, where only a few units contribute significantly to the output image. Consequently, these active units closely resemble the final reconstructed images. As the effect of regularization decreases (i.e., smaller $\lambda$), the number of active units increases. In particular, without regularization (i.e., $\lambda = 0$), the trained INR is no longer sparse, resulting in a larger number of active units, which appear random in structure and do not visually resemble the target phantom image. 
For the \textsf{DOT} phantom, similar figures representing the effect of regularization on the phantom reconstruction, unit sizes of a trained shallow INR, and the visualization of the four largest active units are provided in the Supplementary Materials \ref{sec:supp:figures}.
\begin{figure*}[htbp!]
    \centering    
\includegraphics[width=\textwidth]
{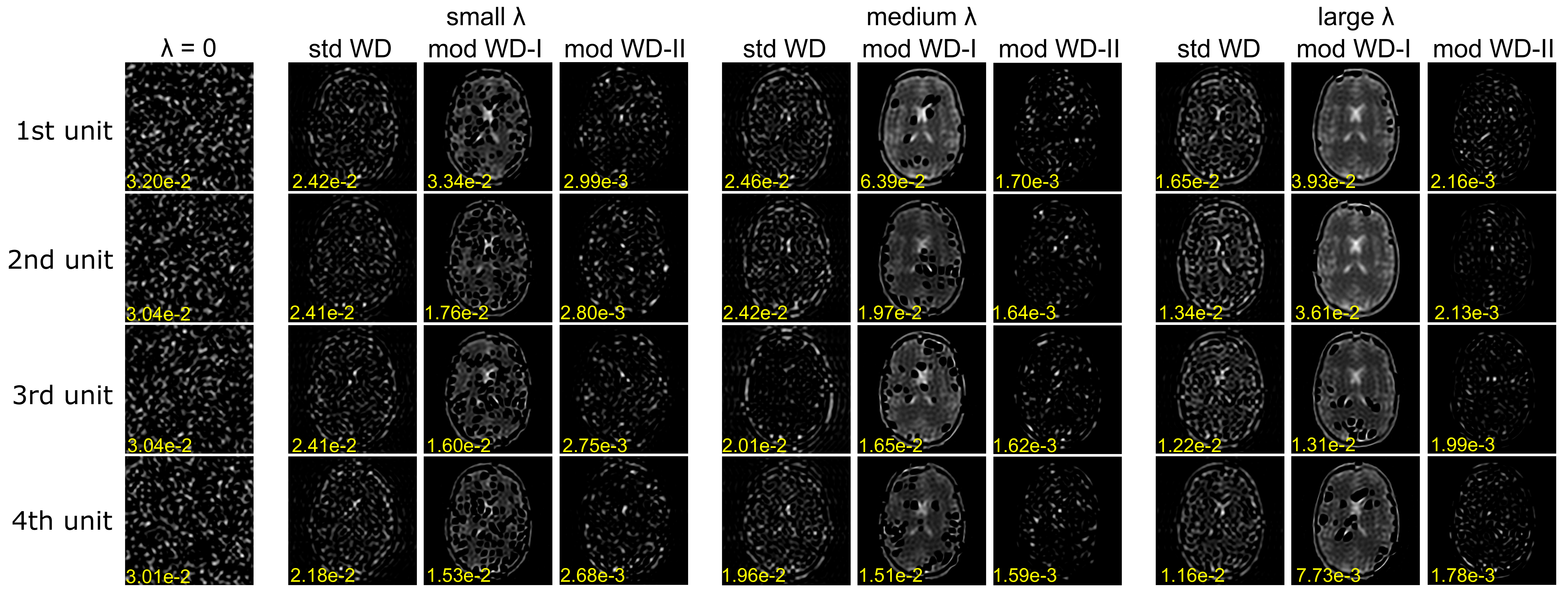}
    \caption{\footnotesize \textbf{Visualization of the four largest active units for the \textsf{PWC-BRAIN} phantom, obtained with a shallow INR.} The INR is trained with standard weight decay regularization (std WD), The modified weight decay regularization proposed in  \Cref{thm:main} (mod WD-I), and the modified weight decay regularization proposed in \Cref{thm:main2} (mod WD-II). All images are normalized to the [0,1] range. The corresponding ``unit size'' $|a_i|\eta(\vw_i)$ is displayed on the image in yellow.}
    \label{fig:brain_active_units_normalized}
\end{figure*}

\subsection{Super-Resolution with Deep INRs}
\subsubsection{Effect of Depth for Fourier Feature INRs}
\label{subsec:Effect-of-Depth}

Though our theory focuses on shallow INRs, deeper INR architectures are often used in practice. Here, we empirically explore the effect of INR depth on super-resolution recovery. Note that the generalized weight decay penalties investigated previously are only defined for shallow INRs, so for these experiments we focus on standard weight decay regularization only.

In \Cref{fig:PWC_SL_depth_reg_eff}, we show the results of fitting an INR of width 100 with depths of $2, 5, 10$, $15$, and $20$, with and without weight decay regularization for the recovery of the \textsf{SL} and \textsf{PWC BRAIN} phantoms.  
These results show that increasing INR depth results in lower MSE and fewer ringing artifacts, indicating that deeper INRs may be biased towards piecewise constant images. Additionally, for both phantoms, comparing reconstructions with and without regularization shows that, at any fixed depth, applying weight decay regularization leads to lower image MSE and better visual quality. This suggests that weight decay regularization can also help improve recovery in deep INR networks. For more details on the hyperparameter settings used for these results see Supplementary Materials \ref{depth}.

\begin{figure}[ht!]
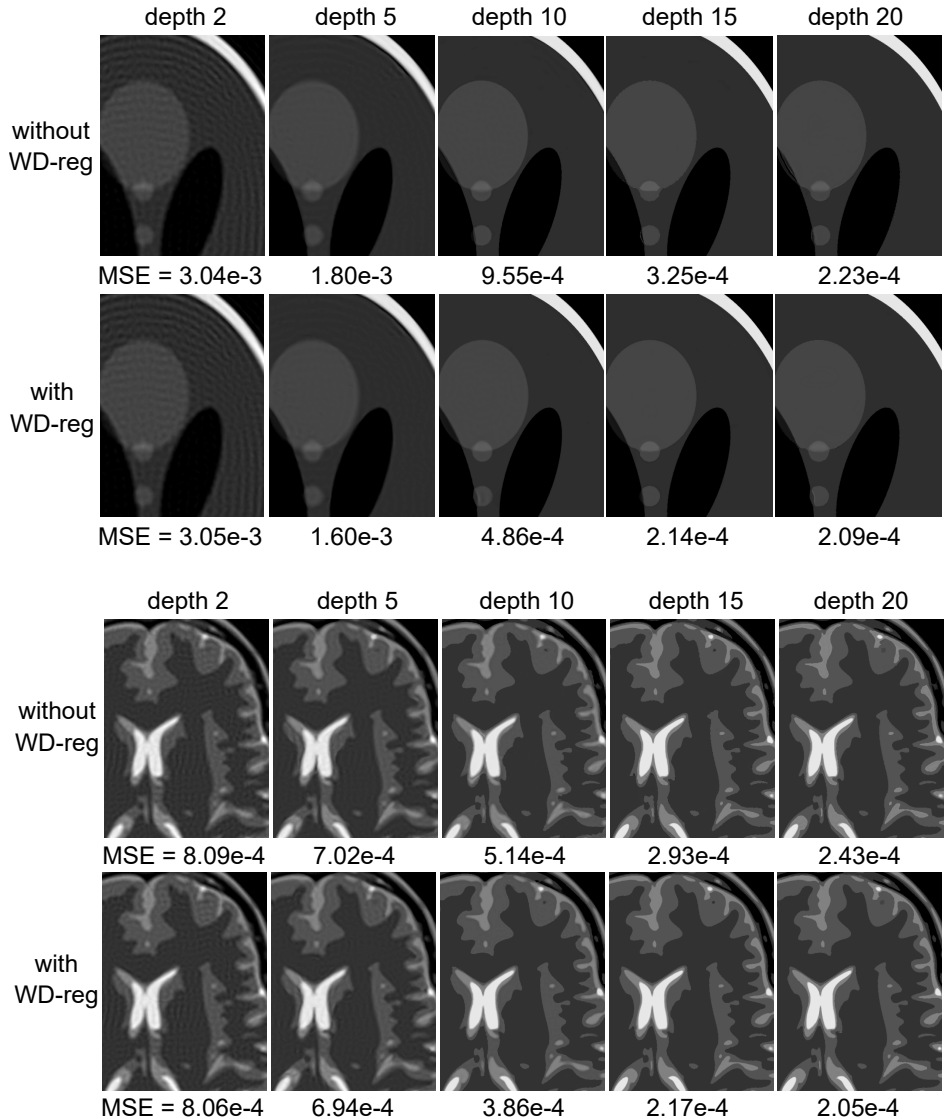

    \centering
    \includegraphics[width=0.76\textwidth]{figs/SL_depth_reg_W100.pdf}\\
    \vspace{1em}
    \includegraphics[width=0.76\textwidth]{figs/PWC_depth_reg_W100.pdf}
   \caption{\footnotesize\textbf{Effect of INR depth on \textsf{SL} and \textsf{PWC-BRAIN} phantom recovery.} The top row shows comparisons of reconstructions of an INR of width 100 across depths 2, 5, 10, 15, and 20 trained with no regularization (i.e., with the weight decay hyperparameter $\lambda = 0$). The bottom row shows comparisons of reconstructions of an INR with varying depths, trained with the standard weight decay regularization. All images are displayed on a scale of [0,1]. }
    \label{fig:PWC_SL_depth_reg_eff}
\end{figure}

\subsubsection{Comparison with Different INR Architectures}
\label{subsec:{Comparison with Different INR Architectures}}
In this work, we focus on one specific INR architecture inspired by \cite{tancik2020fourier} that combines a Fourier features embedding with a fully connected neural network with ReLU activation function. However, there are several alternative INR architectures that use different activation functions and/or feature embeddings. Here, we verify that the INR architecture investigated in this work gives comparable results to other popular INR architectures on the super-resolution recovery problem.

\Cref{fig:SL_diff_methods} and \Cref{fig:PWC_diff_methods} present the reconstructions of the \textsf{SL} and \textsf{PWC-BRAIN} phantoms using various INR architectures. Specifically, we compare with three established INR architectures: Hash Encoding networks \cite{muller2022instant}, SIREN \cite{sitzmann2020implicit}, and ReLU networks with \emph{random} Fourier features \cite{tancik2020fourier}, in contrast to the \emph{non-random} Fourier features used in previous experiments in this work. In addition, motivated by our earlier findings that deeper INR architectures tend to improve the recovery of piecewise constant phantoms, we evaluate SIREN at two depths: one with a depth of 5, denoted as ``SIREN'' and one with a depth of 14, denoted as ``deep SIREN''.

For a fair comparison, we conducted experiments with different training settings, adjusting the INR depth, the learning rate, weight decay regularization strength $\lambda$, and other architecture-specific hyperparameters. In \Cref{fig:SL_diff_methods} and \Cref{fig:PWC_diff_methods}, we display the reconstruction that achieved the lowest image MSE across all tuned hyperparameters. See Supplementary Materials \ref{diff_architectures_ex_details} for the hyperparameter settings used in these experiments.

As shown in these figures, all INR architectures result in lower image MSE values compared to zero-filled IFFT reconstruction. Reconstructions using Hash Encoding and SIREN architectures are comparable to reconstructions obtained using a shallow INR with Fourier features layer in terms of image MSE and visual quality. The deep INR architectures (deep SIREN, and deep ReLU networks with Fourier features) give a marked improvement, eliminating most visible ringing artifacts, and achieve a substantially lower image MSE than shallower architectures.  Additionally, a figure representing the \textsf{PWS-BRAIN} phantom recovery using different INR architectures is provided in the Supplementary Materials \ref{fig:PWS_diff_methods}. Here a similar improvement is seen with deeper INR architectures, though the relative change MSE and visual quality is smaller than with the piecewise constant phantoms.

\begin{figure}[htbp!]
    \centering
    \includegraphics[width=\textwidth]{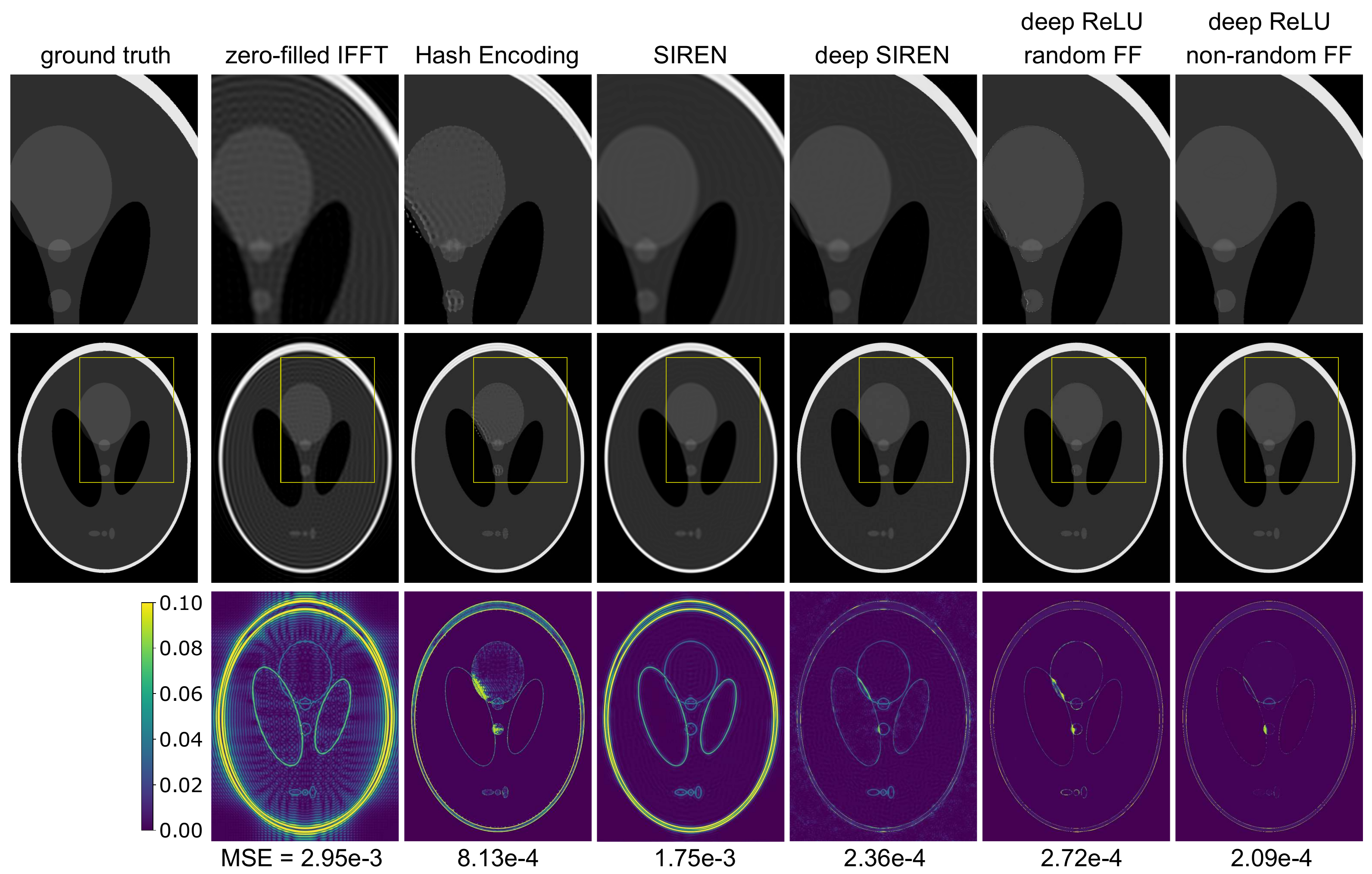}\\
    \caption{\footnotesize \textbf{ \textsf{SL} recovery using different INR architectures.} Top row: we compare the zero-filled IFFT reconstruction with the phantom recovery from a Hash Encoding INR, depth 5 SIREN (SIREN), depth 14 SIREN (deep SIREN), and depth 20 ReLU INR (deep ReLU) with \emph{random} Fourier features and \emph{non-random} Fourier features. The bottom row shows the absolute value of the difference with the ground truth. All grayscale images are shown on the scale [0, 1].}
    \label{fig:SL_diff_methods}
\end{figure}

\begin{figure*}[htbp!]
    \centering    
    \includegraphics[width=\textwidth]
    {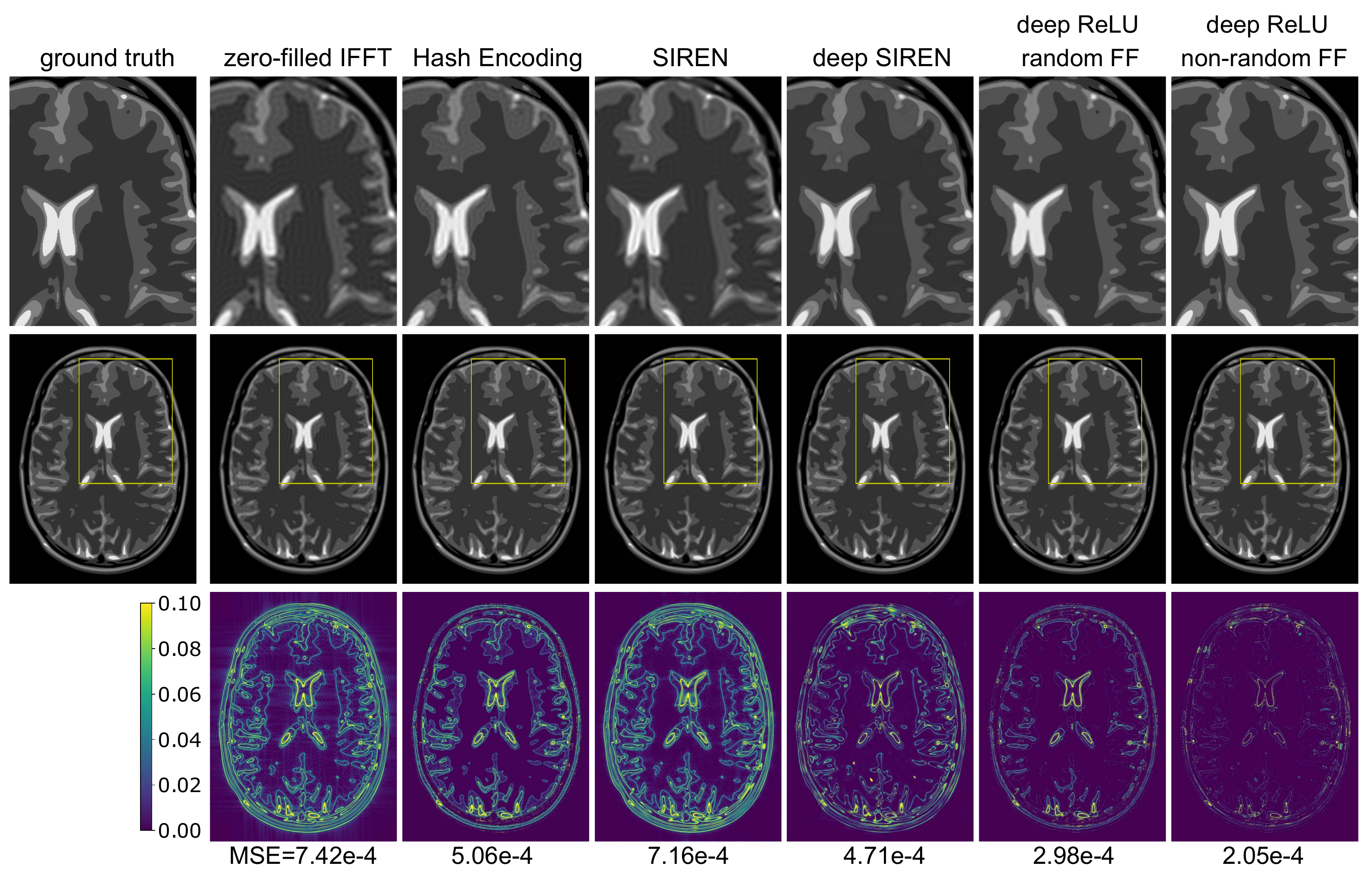}\\
    \caption{\footnotesize\textbf{\textsf{PWC-BRAIN} recovery using different INR architectures.} Top row: we compare the zero-filled IFFT reconstruction with the phantom recovery from a Hash Encoding INR, depth 5 SIREN (SIREN), depth 14 SIREN (deep SIREN), and depth 20 ReLU INR (deep ReLU) with \emph{random} Fourier features and \emph{non-random} Fourier features. The bottom row shows the absolute value of the difference with the ground truth. All grayscale images are shown on the scale [0, 1].}
    \label{fig:PWC_diff_methods}
\end{figure*}

\section{Conclusions}
\label{sec:conclusions}

In this work, we introduce a mathematical framework for studying exact recovery of continuous domain images from low-pass samples with a specific INR architecture. We show that fitting a shallow INR architecture having ReLU activation and an initial Fourier features layer with a generalized form of weight decay regularization is equivalent to a convex problem optimization defined over a space of measures. Using this perspective, we identify a sufficient number of Fourier samples to uniquely recover an image realizable as INR in two settings:  (i) the image is realizable as a width-1 shallow INR with a general Fourier features embedding defined in terms of integer frequencies, and (ii) the image is realizable as a width-$s$ shallow INR with a restricted Fourier features embedding defined in terms of two base frequencies and positive outer-layer weights. Furthermore, we conjecture that similar exact recovery results hold with standard weight decay regularization, which is verified empirically.

While this work focused on super-resolution imaging, the convex optimization framework developed here naturally extends to other linear inverse problems, such as MRI reconstruction from non-uniform Fourier samples or computed tomography reconstruction. An important next step is to investigate whether our proof techniques can yield exact recovery guarantees for INRs over a broader class of inverse problems.

Our analysis characterizes the global optima of an INR training problem, showing that its solutions coincide with those of a convex optimization problem. However, since the original INR training objective is non-convex, convergence to global optima via gradient-based methods is not guaranteed. Despite this, our empirical results demonstrate that standard INR training procedures can converge to global optima under certain conditions. Establishing theoretical convergence guarantees for practical gradient-based INR training algorithms is an interesting open problem.

Finally, the current analysis is specific to the INR architecture proposed in \cite{tancik2020fourier}. Extending our results to other popular INR architectures (as discussed in \Cref{sec:relatedwork}) is a natural direction for further exploration. Moreover, for mathematical tractability, we restricted our focus to shallow networks with a single hidden layer. Generalizing our analysis to deeper INR architectures and examining the role of weight decay regularization in such settings will be a focus of future work.

\appendix
\section{Proofs of Lemmas in \Cref{sec:Theory}} 
\subsection{Proof of \Cref{lem:ell2_to_ell1}}\label{sec:app:proof_ell2_to_ell1}
Given any parameter vector $\theta = ((a_i,\vw_i))_{i=1}^W \in \Theta_W$ we construct a new parameter vector $\theta' =  ((a'_i,\vw'_i))_{i=1}^W \in \Theta_W$ that satisfies $f_\theta = f_{\theta'}$ as follows:
For all $i$ such that $\eta(\vw_i) > 0$, define $c_i = \sqrt{|a_i|/\eta(\vw_i)}$, and set $a'_i = a_i/c_i$ and $\vw'_i =  c_i \vw_i$. For all $i$ such that $\eta(\vw_i) = 0$, set $\vw_i' = \bm 0$ and $a_i' = 0$. Then, by positive 1-homogeneity of the ReLU activation and of $\eta$, and since $\eta(\vw_i) = 0$ implies $[\vw_i^\T\bm\gamma(\cdot)]_+ = 0$, we see that $a'_i[(\vw_i')^\T \bm\gamma(\cdot)]_+ = a_i[\vw_i^\T \bm\gamma(\cdot)]_+$ for all $i$, and so $f_\theta = f_{\theta'}$. 

Now, by the inequality of arithmetic means and geometric means, we have
\[
R(\theta) = \frac{1}{2}\sum_{i=1}^W\left(|a_i|^2 + \eta(\vw_i)^2 \right) \geq \sum_{i=1}^W |a_i|\eta(\vw_i) = R(\theta').
\]
This shows that every parameter vector $\theta \in \Theta_W$ can be mapped to another parameter vector $\theta' \in \Theta_W$ that realizes the same function and achieves the modified cost $\sum_{i=1}^W |a_i|\eta(\vw_i)$ which is also a lower bound on $R(\theta)$. Therefore, we see that \eqref{eq:opt_param_space} has the same function space minimizers as
\[
\min_{\theta \in \Theta_W} \sum_{i=1}^W |a_i|\eta(\vw_i)~~s.t.~~\gF_\Omega f_\theta = \vy.
\]
Finally, since the quantity $|a_i|\eta(\vw_i)$ is invariant to the re-scaling $\vw_i \mapsto \vw_i/\|\vw_i\|$ and $a_i \mapsto \|\vw_i\|a_i$, and this re-scaling of parameters realizes the same function, we may additionally constrain the optimization to all parameter vectors satisfying $\|\vw_i\| = 1$ for all $i \in [W]$, i.e., to all $\theta \in \Theta_{W}'$, which gives the equivalence between function space minimizers of \eqref{eq:opt_param_space} and \eqref{eq:opt_finite_width_R}.

\subsection{Proof of \Cref{lem:low_width_minimizers}}\label{sec:app:proof_lem2}
To prove \Cref{lem:low_width_minimizers}, we apply the abstract representer theorem of \cite[Theorem 3.3]{bredies2020sparsity}. A simplified version of this result is stated as follows: if $Q$ is a norm on a locally convex space $\gX$, and $\gA:\gX\rightarrow H$ is a continuous linear operator mapping surjectively onto a finite dimensional Hilbert space $H$, then given any $y \in H$ the optimization problem
\[
\min_{x \in \gX}~Q(x)~~s.t.~~\gA x = y,
\]
has a minimizer of the form $x^* = \sum_{i=1}^s a_i u_i$ where $s \leq \dim H$, $a_i \in \R$, and $u_i$ are extremal points of the unit ball $\{x\in\gX : \phi(x) \leq 1\}$.

In the present setting, we consider the locally convex space $\gX = \gM(\Peta)$ of finite signed Radon measures over $\Peta$ equipped with the weak-$*$ topology. Here, we identify $\gM(\Peta)$ as the continuous dual of $\gC_0(\Peta)$, the Banach space of continuous functions on $\Peta$ vanishing on the boundary set $\partial\Peta \subseteq \{\vw \in \mathbb{S}^{D-1} : \eta(\vw) = 0\}$ equipped with the supremum norm. We also consider the norm $Q= \|\cdot\|_{TV,\eta}$ on $\gM(\Peta)$, and the linear operator $\gA = \gK_\Omega$ mapping into $H =  \Csym(\Omega)$, the space of complex-valued conjugate symmetric arrays indexed by $\Omega$, which is isomorphic to $\R^{|\Omega|}$.

By a similar argument to \cite[Theorem 4.2]{bredies2020sparsity}, it is straightforward to show that the extreme points of the unit ball $\{\mu \in \gM(\Peta) : \|\mu\|_{TV,\eta} \leq 1\}$ coincide with the set of scaled Dirac measures $\{\pm \eta(\vw)^{-1}\delta_\vw : \vw \in \Peta\}$. Hence, assuming $\gK_\Omega$ is weak-$*$ continuous and surjective, we see that \eqref{eq:opt_func_space} has a sparse solution $\mu^* = \sum_{i=1}^s a_i \delta_{\vw_i}$ where $a_i \in \R$ and $\vw_i \in \Peta$, and $s \leq \dim_\R \Csym(\Omega) = |\Omega|$. If $W \geq |\Omega|$, this implies $\mu^*$ is a minimizer of \eqref{eq:opt_discrete_measures} and so $p^*_W=p^*$, as claimed.

All that remains so show is that $\gK_\Omega:\gM(\Peta)\rightarrow \Csym(\Omega)$ is weak-$*$ continuous and surjective. First, to prove that $\gK_\Omega$ is weak-$*$ continuous, it suffices to show $\gK_\Omega$ is the adjoint of a bounded linear operator \cite[Prop.~1.3, pg.~167]{conway2019course}. In particular, consider the linear operator $\gL_\Omega: \Csym(\Omega) \rightarrow \gC_0(\Peta)$ defined for any $\vz \in \Csym(\Omega)$ by
\begin{equation}\label{eq:predualofKOmega}
[\gL_\Omega\vz](\vw) = \langle \gF_\Omega[\vw^\T\bm\gamma(\cdot)]_+, \vz\rangle~~\forall\vw \in \Peta.
\end{equation}
By the Cauchy-Schwarz inequality, we have
\[
|[\gL_\Omega\vz](\vw)| \leq \|\gF_\Omega[\vw^\T\bm\gamma(\cdot)]_+\|_2\|\vz\|_2,
\]
and Parseval's Theorem gives the bound
\begin{align*}
\|\gF_\Omega[\vw^\T\bm\gamma(\cdot)]_+\|_2 & = 
\|\gF_\Omega^*\gF_\Omega[\vw^\T\bm\gamma(\cdot)]_+\|_{L^2(\Td)}\\
& \leq \|[\vw^\T\bm\gamma(\cdot)]_+\|_{L^2(\Td)}\\ & \leq \|\vw^\T\bm\gamma(\cdot)\|_{L^2(\Td)}\\ & = \|\vw\|_2 = 1,
\end{align*}
where in the last equality we used the fact that the components $\{\gamma_i\}_{i=1}^D$ of the Fourier features embedding $\bm\gamma$ defined in \eqref{eq:ffmap} are orthonormal as functions in $L^2(\Td)$.
Therefore, we have shown $|[\gL_\Omega\vz](\vw)| \leq \|\vz\|_2$ for all $\vw \in \Peta$,
and so $\gL_\Omega$ is a bounded linear operator. Now we show $\gK_\Omega = \gL_\Omega^*$:  for any $\vz \in \Csym(\Omega)$ and $\mu \in \gM(\Peta)$, we have
\begin{align*}
\langle \mu, \gL_\Omega\vz\rangle & = \int_{\Peta} [\gL_\Omega\vz](\vw) d\mu(\vw)\\
& = \int_{\Peta} \langle \F_\Omega[\vw^\T\bm\gamma(\cdot)]_+, \vz\rangle d\mu(\vw)\\
& = \left\langle\int_{\Peta}  \F_\Omega[\vw^\T\bm\gamma(\cdot)]_+  d\mu(\vw), \vz\right\rangle\\
& = \langle \gK_\Omega\mu, \vz\rangle,
\end{align*}
which proves $\gK_\Omega = \gL_\Omega^*$, as claimed.

Finally, we show that $\gK_\Omega$ is surjective, i.e., $\mathrm{Im}(\gK_\Omega)= \Csym(\Omega)$, or equivalently, $\rank \gK_\Omega = |\Omega|$. For this, it suffices to show $\gL_\Omega$ has a trivial nullspace, since by the rank-nullity theorem this implies $\rank \gL_\Omega = |\Omega|$, and so $\rank \gL_\Omega^* = \rank \gK_\Omega  = |\Omega|$, as well. 

Towards this end, suppose $\vz \in \Csym(\Omega)$ is non-zero. We will show $\gL_\Omega\vz$ is not identically zero, i.e., there exists a vector $\vw_1' \in \Peta$ such that $[\gL_\Omega\vz](\vw_1') \neq 0$. First, observe that for all $\vw \in \Peta$ we have
\begin{align*}
[\gL_\Omega\vz](\vw) & = \langle \gF_\Omega [\vw^\T\bm\gamma(\cdot)]_+,\vz\rangle\\
& = 
\int_{\Td}[\vw^\T\bm\gamma(\vx)]_+[\gF_\Omega^*\vz](\vx)\,d\vx.
\end{align*}
Define $\Phi = \gF_\Omega^*\vz$, which is a non-zero real-valued trigonometric polynomial. Without loss of generality, assume the set of values of which $\Phi$ is positive $P = \{\vx \in \Td: \Phi(\vx) > 0\}$ is non-empty (if it is empty, the argument below can be repeated for the set over which $\Phi$ is negative, instead). By continuity of $\Phi$, $P$ is an open set, and so it contains an open ball $B(\vx_0,\epsilon)$ for some radius $\epsilon > 0$. Let 
\[\tau_0(\vx) = \frac{1}{2} + \frac{1}{2d}\sum_{j=1}^d \cos(2\pi x_j),\] which belongs to the range of the Fourier features embedding by assumption, i.e., there exists $\vw_0 \in \R^D$ such that $\tau_0(\vx) = \vw_0^\top \bm\gamma(\vx)$. Note that $0 \leq \tau_0(\vx) \leq 1$ for all $\vx \in \Td$, and it is easy to prove that $\vx = \bm 0$ is the unique global maximizer of $\tau_0$. Now, define $\tau_1(\vx) := \tau_0(\vx-\vx_0)-\alpha$ where we choose $0 < \alpha < 1$ to be such that $[\tau_1]_+$ is non-zero with support contained in $B(\vx_0,\epsilon)$. Since $\tau_1$ is simply a translation and constant shift of $\tau_0$ it has the same frequency support as $\tau_0$. Therefore, there exists a non-zero vector $\vw_1 \in \R^D$ such that $\tau_1(\vx) = \vw_1^\T \bm\gamma(\vx)$. Defining $\vw_1' = \vw_1/\|\vw_1\|_2$, we see that
\[
[\gL_\Omega\vz](\vw_1') = \|\vw_1\|_2^{-1}
\int_{\Td}[\tau_1(\vx)]_+\Phi(\vx)\,d\vx = \|\vw_1\|_2^{-1}\int_{P}[\tau_1(\vx)]_+\Phi(\vx)\,d\vx > 0.
\]
Finally, since $[(\vw_1')^\T\bm\gamma(\cdot)]_+ = \|\vw_1\|_2^{-1}[\tau_1(\cdot)]_+ \neq 0$, by the admissibility conditions it must be the case that $\eta(\vw_1') > 0$, which gives $\vw_1' \in \Peta$. This shows there exists $\vw_1' \in \Peta$ such that $[\gL_\Omega \vz](\vw_1') \neq 0$, as claimed.
\subsection{Description of the Dual Problem}\label{app:dual}
The optimization problem \eqref{eq:opt_func_space} is convex. We now describe the associated convex dual problem\footnote{Technically speaking, this is the pre-dual problem, since the Banach space $\gM(\Peta)$ is not reflexive. However, formally it has the same properties as a convex dual, so we call it the dual problem for simplicity.}.
\begin{myLem}
Let $p^*$ be the minimum of the primal problem \eqref{eq:opt_func_space}. Then we have $p^* = d^*$ where
\begin{equation}\label{eq:dual_problem}
   d^* = \sup_{\vz \in \Csym(\Omega)} \langle \vy,\vz\rangle~~s.t.~~|[\gL_\Omega\vz](\vw)| \leq \eta(\vw)~~\forall \vw \in \Peta,\tag{\ensuremath{D}}
\end{equation}
where $\gL_\Omega: \Csym(\Omega)\rightarrow \gC_0(\Peta )$ is the bounded linear operator defined in \eqref{eq:predualofKOmega}.
\end{myLem}
\begin{proof}
First, we may re-cast \eqref{eq:dual_problem} as the convex minimization problem:
\begin{equation*}
-d^* = \inf_{\vz \in \Csym(\Omega)} F(\vz) + G(\gL_\Omega\vz),
\end{equation*}
where $F(\vz) = -\langle \vy,\vz\rangle$ and $G$ is the indicator function of the convex set \[C = \{\phi \in \gC_0(\Peta) : |\phi(\vw)| \leq \eta(\vw), \forall \vw \in \Peta\}.\]
Since $\gL_\Omega$ is a bounded linear operator whose adjoint is $\gK_\Omega: \gM(\Peta)\rightarrow \Csym(\Omega)$, by Fenchel-Rockafellar duality \cite[Theorem II.4.1]{ekeland1999convex} we have
\begin{equation}\label{eq:fenchelstep}
-d^* \geq \sup_{\mu \in \gM(\Peta)} -G^*(\mu) -F^*(\gK_\Omega\mu),
\end{equation}
where $G^*$ and $F^*$ are the convex conjugates of $G$ and $F$, respectively. By the characterization of the total variation $\|\cdot\|_{TV}$ as a dual norm, we have:
\begin{align*}
\|\mu\|_{TV,\eta} = \|\eta\cdot \mu\|_{TV} & =  \sup_{\substack{\psi \in \gC_0(\mathbb{S}^{D-1})\\|\psi(\vw)|\leq 1,\forall \vw\in\mathbb{S}^{D-1}}} \int_{\S^{D-1}} \psi(\vw) \eta(\vw) d\mu(\vw)\\
& =  \sup_{\substack{\phi \in \gC_0(\Peta)\\|\phi(\vw)|\leq \eta(\vw),\forall \vw\in\Peta}} \int_{\Peta} \phi(\vw) d\mu(\vw) = \sup_{\phi \in C} \langle \mu,\phi\rangle = G^*(\mu).
\end{align*}
Additionally, it is straightforward to show $F^*$ is the indicator function of the singleton $\{\vy\}$. Therefore, \cref{eq:fenchelstep} is equivalent to 
\[
d^* \leq p^* := \inf_{\mu \in \gM(\Peta)}\|\mu\|_{TV,\eta}~~\mathrm{s.t.}~~\gK_\Omega\mu = \vy.
\]
Finally, since it is clear that $F$ and $G$ are proper convex lower semi-continuous functions, $F(\vz)$ is finite at $\vz = \bm 0$, and $G(\phi) $ is finite and continuous at $\phi = \gL_\Omega\bm 0 = 0$, by  \cite[Theorem II.4.1]{ekeland1999convex} strong duality holds, i.e., $p^* = d^*$, as claimed. Additionally, \cite[Theorem II.4.1]{ekeland1999convex} ensures the primal problem has a solution, i.e., there exists a $\mu^* \in \mathcal{M}(\Peta)$ attaining the infimum, so that the use of $\min$ in place of $\inf$ in \eqref{eq:opt_func_space} is justified.
\end{proof}

Now we show how the existence of a particular feasible variable for the dual problem \eqref{eq:dual_problem} ensures optimality of a sparse measure for the primal problem \eqref{eq:opt_func_space}.
Suppose $\vy = \gF_\Omega f$ where $f = \sum_{i=1}^s a_i [\vw_i^\top \bm\gamma(\cdot)]_+$ such that $a_i \neq 0$, $\vw_i \in \Peta$ for all $i\in[s]$, and $\vw_i \neq \vw_j$ for all $i\neq j$. 
Then for $\mu^* = \sum_{i=1}^s a_i\delta_{\vw_i}$ we have $\|\mu^*\|_{TV,\eta} = \sum_{i=1}^s |a_i|\eta(\vw_i)$, and by linearity of $\gF_\Omega$, we see that
\begin{align*}
\langle \vy, \vz \rangle & = \langle \gF_\Omega f,\vz\rangle = \sum_{i=1}^s a_i \langle \gF_\Omega[\vw_i^\T\bm\gamma(\cdot)]_+,\vz\rangle = \sum_{i=1}^s a_i [\gL_\Omega\vz](\vw_i).
\end{align*}
Therefore, if there exists a dual feasible $\vz \in \Csym(\Omega)$ such that $\Phi = \gL_\Omega\vz$ satisfies $\Phi(\vw_i) = \sign(a_i)\eta(\vw_i)$ for all $i\in [s]$, we see that 
\[
\langle \vy,\vz\rangle = \sum_{i=1}^s |a_i|\eta(\vw_i) = \|\mu^*\|_{TV,\eta},
\]
i.e., the dual objective matches the primal objective, which implies $\mu^*$ is a global minimizer of \eqref{eq:opt_func_space}, and so $f = f_{\mu^*}$ is a function space minimizer.

However, this argument does not determine whether $f_{\mu^*}$ is the unique function space minimizer of \eqref{eq:opt_func_space}. A sufficient condition that guarantees uniqueness is as follows:
\begin{myLem}\label{lem:dual_cert}
Suppose $\mu^* = \sum_{i=1}^s a_i \delta_{\vw_i}$ is feasible for \eqref{eq:opt_func_space} where the $\vw_i \in \Peta$ are distinct and $a_i \neq 0$ for all $i\in[s]$. Also, suppose there exists a finite set $\gW \subset \Peta$ such that $\{\vw_1,...,\vw_s\} \subseteq \gW$, and
for any measure $\nu_0 \in \gM(\Peta)$
with $\mathrm{supp}(\nu_0) \subseteq \gW$
we have $\gK_\Omega \nu_0 = 0$ implies $f_{\nu_0} = 0$. 
Furthermore, if there exists $\vz \in \Csym(\Omega)$ such that $\Phi = \gL_\Omega \vz$ satisfies
\begin{enumerate}
    \item[(i)] $|\Phi(\vw)| \leq \eta(\vw)$ for all $\vw \in \Peta$,
    \item[(ii)] $\Phi(\vw_i) = \sign(a_i)\eta(\vw_i)$ for all $i\in [s]$, and
    \item[(iii)] $|\Phi(\vw)| < \eta(\vw)$ for all $\vw \in \Peta/\gW$,
\end{enumerate}
then $f_{\mu^*} = \sum_{i=1}^s a_i[\vw_i^\T\bm\gamma(\cdot)]_+$ is the unique function space minimizer of \eqref{eq:opt_func_space}.
\end{myLem}
\begin{proof}
Suppose $\mu \in \gM(\Peta)$ is another measure feasible for \eqref{eq:opt_func_space} (i.e., $\gK_\Omega \mu = \vy$) such that $f_\mu \neq f_{\mu^*}$. Define $\nu = \mu-\mu^*$. By Lebesgue's Decomposition Theorem, there exist measures $\nu_0, \nu_1 \in \gM(\Peta)$ such that $\nu = \nu_0 + \nu_1$ where $\nu_0$ is absolutely continuous with respect to the measure $\delta_{\gW} := \sum_{\vw \in \gW} \delta_{\vw}$ and $\nu_1$ is mutually singular to $\delta_{\gW}$. In particular, we have $\mathrm{supp}(\nu_0) \subseteq \gW$, and so $\nu_1$ is also mutually singular to $\mu^*$.

First, we show $\nu_1 \neq 0$ by way of contradiction: if $\nu_1 = 0$ then
\[
\gK_\Omega \nu_0 = \gK_\Omega \nu = \gK_\Omega\mu-\gK_\Omega\mu^* = 0,
\]
and so $f_{\nu_0} = 0$ by assumption. This implies $f_{\nu} = 0$, or equivalently, $f_{ \mu} = f_{\mu^*}$, a contradiction. Therefore, $\nu_1\neq 0$, and in particular $\|\nu_1\|_{TV,\eta} > 0$ since $\|\cdot\|_{TV,\eta}$ is a norm.

Next, using the fact that $\nu_1$ is mutually singular to both $\mu^*$ and $\nu_0$, we have
\begin{align*}
\|\mu\|_{TV,\eta} & = \|\mu^* + \nu\|_{TV,\eta}\\
& = \|\mu^* + \nu_0\|_{TV,\eta} + \|\nu_1\|_{TV,\eta}\\
& >  \langle \mu^* + \nu_0, \Phi \rangle + \langle\nu_1,\Phi\rangle\\
& = \langle \mu^*,\Phi \rangle + \langle\nu,\Phi\rangle\\
& = \|\mu^*\|_{TV,\eta} + \langle \nu,\Phi\rangle\\
& = \|\mu^*\|_{TV,\eta}+ \langle \nu,\gL_\Omega\vz\rangle\\
& = \|\mu^*\|_{TV,\eta} + \langle \gK_\Omega\nu,\vz\rangle\\
& = \|\mu^*\|_{TV,\eta},
\end{align*}
where the strict inequality comes from the fact that  $\|\nu_1\|_{TV,\eta} > \langle\nu_1,\Phi\rangle$ since $\mathrm{supp}(\nu_1) \subset \Peta/\gW$ and the assumption that $|\Phi(\vw)| < \eta(\vw)$ for all $\vw \in \Peta/\gW$. 
Therefore, $\|\mu\|_{TV,\eta} > \|\mu^*\|_{TV,\eta}$ for all feasible $\mu$ with $f_{\mu}  \neq f_{\mu^*}$, which shows $f_{\mu^*} = \sum_{i=1}^s a_i[\vw_i^\T \bm\gamma(\cdot)]_+$ is the unique function space minimizer of \eqref{eq:opt_func_space}. 
\end{proof}

We call the function $\Phi$ appearing in \Cref{lem:dual_cert} a \emph{dual certificate} for the primal feasible measure $\mu^*$ associated with the support set $\gW$. The proofs of our exact recovery theorems, given below, rely on explicitly constructing dual certificates. 

\section{Proof of Sampling Theorems}
For our proofs in this section, it will be convenient to consider a complex Fourier basis for trigonometric polynomials. For any $\Gamma \subset \mathbb{Z}^d$ that is finite and symmetric (i.e., $\Gamma = -\Gamma$), we let $\TP(\Gamma) \subset L^2(\torus^d)$ denote the space of real-valued trigonometric polynomials $\tau:\torus^d\rightarrow\R$ with Fourier coefficients supported in $\Gamma$. Equivalently, we have
\[
\TP(\Gamma) = \left\{ \sum_{\vk \in \Gamma} c[\vk] e^{i2\pi \vk^\T \vx}  : \vc \in \Csym(\Gamma)\right\},
\]
where recall that we use $\Csym(\Gamma)$ to denote all complex-valued conjugate symmetric arrays indexed by $\Gamma$, i.e., $\vc = (c[\vk] : \vk \in \Gamma) \in \Csym(\Gamma)$ if and only if $c[\vk] = \overline{c[-\vk]}$ for all $\vk \in \Gamma$.
Note that $\TP(\Gamma)\simeq \Csym(\Gamma)$ by identifying any $\tau \in \TP(\Gamma)$ with its array of Fourier coefficients $\F_\Gamma\tau = (\hat{\tau}[\vk] : \vk \in \Gamma) \in \Csym(\Gamma)$.

In particular, suppose the Fourier features embedding $\bm\gamma$ given in \eqref{eq:ffmap} is defined in terms of frequencies $\Omega_0 = \{\vk_1,...,\vk_p\} \subset \Zd$ such that $\vk_i \neq \pm \vk_j$ for all $i\neq j$ and $\vk_i \neq \bm 0$ for all $i \in [p]$. Then we have
\[
    \left\{ \vw^\T \bm\gamma(\cdot) : \vw\in\R^{2p+1} \right\} = \TP(\Omega_0^*)
\]
where $\Omega_0^* = \Omega_0 \cup -\Omega_0 \cup \{\bm 0\}$. Furthermore, in this case the components of $\bm\gamma$ are orthonormal as functions in $L^2(\torus^d)$, and so if $\tau = \vw^\T \bm\gamma(\cdot)$ we have
$\|\tau\|_{L^2(\Td)} =\|\vw\|_2$.

Finally, we will frequently use the fact that if $\tau,\sigma \in \TP(\Gamma)$ where $\Gamma \subset \mathbb{Z}^d$ is finite and symmetric, then by the convolution theorem, their product $\tau\sigma \in \TP(2\Gamma)$ where $2\Gamma := \{\vk + \bm\ell : \vk,\bm\ell \in \Gamma\}$, i.e., the Minkowski sum of $\Gamma$ with itself. For example, if $\Gamma = \{\vk \in \Zd : \|\vk\|_\infty \leq K_0\}$ for some positive integer $K_0$, then $2\Gamma = \{\vk \in \Zd : \|\vk\|_\infty \leq 2K_0\}$. Likewise, if $\Gamma = \{\vk \in \Zd : \|\vk\|_1 \leq K_0\}$ then $2\Gamma = \{\vk \in \Zd : \|\vk\|_1 \leq 2K_0\}$.

\subsection{Proof of \Cref{thm:main}}\label{sec:app:thm1_proof}

Before giving the proof of \Cref{thm:main}, we prove the following key lemma, which shows that a rectified trigonometric polynomial is uniquely identifiable from sufficiently many of its low-pass Fourier samples:
\begin{myLem}\label{lem:injectivity}

Under the same assumptions on $\Omega_0$ and $\Omega$ in \Cref{thm:main},
the map $\vw \mapsto \gF_\Omega[\vw^\T\bm\gamma(\cdot)]_+$ is injective over the set $\{\vw \in \R^D : [\vw^\T\bm\gamma(\cdot)]_+ \neq 0 \}$.
\end{myLem}
The key idea behind the proof of \Cref{lem:injectivity} is to show that given $\vc = \gF_\Omega[\tau]_+$ where $\tau = \vw^\T\bm\gamma(\cdot)$ we can construct a linear operator $\gQ_\vc$ acting on trigonometric polynomials whose nullspace is one-dimensional and spanned by $\tau$. Below we first show how to construct $\gQ_\vc$, and then formally prove this claim, and finally give the proof \Cref{lem:injectivity}.

First, fix any $\tau \in \TP(\Omega_0^*)$ such that $[\tau]_+\neq 0$. Let $H:\R\rightarrow\R$ denote the Heaviside step function, i.e., $H(t) = 1$ if $t\geq 0$ and $H(t) = 0$ otherwise. Observe that $[\tau]_+ = \tau H(\tau)$, and furthermore
\[
\nabla [\tau]_+ = H(\tau)\nabla \tau~~a.e.
\]
Therefore,
\[
\tau\nabla [\tau]_+ = \tau H(\tau)\nabla \tau = [\tau]_+ \nabla \tau~~a.e.
\]
or, equivalently,
\begin{equation}\label{eq:annihilate}
\tau\nabla [\tau]_+-[\tau]_+ \nabla \tau = 0~~a.e.
\end{equation}
Passing to Fourier domain, by the convolution theorem, we have
\begin{equation}\label{eq:conv1}
\left(\widehat{\nabla[\tau]_+}\ast \widehat{\tau}\right)[\vk]-\left(\widehat{[\tau]_+}\ast \widehat{\nabla\tau}\right)[\vk] = \bm 0,~~\forall\vk\in\mathbb{Z}^d,
\end{equation}
where $\ast$ denotes the linear convolution of $d$-dimensional arrays.
Using the fact that $\tau \in \TP(\Omega_0^*)$, and for any weakly differentiable functions $g$ we have $\widehat{\nabla g}[\vk] = 2\pi i \vk \widehat{g}[\vk]$ for all $\vk \in \mathbb{Z}^d$, we see that \eqref{eq:conv1} is equivalent to
\begin{align}
\sum_{\bm\ell \in \Omega_0^*}(\bm\ell-\vk) &  \widehat{[\tau]_+}[\bm\ell-\vk]\widehat{\tau}[\bm\ell]- \sum_{\bm\ell \in \Omega_0^*}\bm\ell\,\widehat{[\tau]_+}[\bm\ell-\vk]\,\widehat{\tau}[\bm\ell] = \bm 0,~~\forall\vk \in \mathbb{Z}^d. \label{eq:conv2}
\end{align}
Note that for any fixed frequency $\vk \in \mathbb{Z}^d$ the equation \eqref{eq:conv2} is linear in terms of the Fourier coefficients $(\widehat{\tau}[\bm\ell] : \bm\ell \in \Omega_0^*)$. Also, given Fourier samples $\vc = \gF_\Omega[\tau]_+ = (\widehat{[\tau]_+}[\vk] : \vk\in\Omega)$, the equations in \eqref{eq:conv2} are realizable for all frequencies $\vk \in \mathbb{Z}^d$ belonging to the set $\Omega|\Omega_0^* := \{\vk \in \mathbb{Z}^d : \bm\ell-\vk \in \Omega,~\forall\bm\ell \in \Omega_0^* \}$. 

This motivates the following definition: 
for any $\vc \in \Csym(\Omega)$, define the linear operator $\gQ_\vc: \TP(\Omega_0^*) \rightarrow \Csym(\Omega)^d$ that acts on any $\varphi \in \TP(\Omega_0^*)$ by 
\[
(\gQ_\vc \varphi)[\vk] = \sum_{\bm\ell \in \Omega_0^*}(\bm\ell-\vk)  c[\bm\ell-\vk]\widehat{\varphi}[\bm\ell]- \sum_{\bm\ell \in \Omega_0^*}\bm\ell\,c[\bm\ell-\vk]\,\widehat{\varphi}[\bm\ell],~\forall \vk \in \Omega|\Omega_0^*.
\]
In particular, when $\vc = \gF_\Omega[\tau]_+$ working backwards from \eqref{eq:conv2}, we see that
\[
\gQ_\vc\varphi = \gF_{\Omega|\Omega_0^*}(\varphi \nabla[\tau]_+-[\tau]_+ \nabla \varphi),
\]
and by the identity given in \eqref{eq:annihilate}, we see that $\gQ_{\vc}\tau = 0$, i.e., $\tau$ belongs to the nullspace of $\gQ_{\vc}$. The next result shows that when $\vc = \gF_\Omega[\tau]_+$, and $\Omega$ is sufficiently large, the nullspace of $\gQ_{\vc}$ is one-dimensional and spanned by $\tau$.
\begin{myLem}\label{lem:unique_nullspace} 
    Suppose $\Omega|\Omega_0^* \supseteq 2\Omega_0^*$. Let $\vc = \gF_\Omega[\tau]_+$ where $\tau \in \TP(\Omega_0^*)$ with $[\tau]_+\neq 0$. Then $\ker \gQ_\vc = \mathrm{span}\{\tau\}$.
\end{myLem}
\begin{proof}
The condition $\varphi \in \ker \gQ_\vc$ is equivalent to 
\[
\gF_{\Omega|\Omega^*_0}(\varphi \nabla[\tau]_+-[\tau]_+ \nabla \varphi) = \bm 0.
\]
This implies for all vector-valued functions $\bm\rho \in \TP(\Omega|\Omega_0^*)^d$, i.e., $\bm\rho = (\rho_1,...,\rho_d)$ where $\rho_1,...,\rho_d \in \TP(\Omega|\Omega_0^*)$, we have
\begin{equation}\label{eq:nullcond}
\int_{\Td} (\varphi \nabla[\tau]_+-[\tau]_+ \nabla \varphi)^\T \bm\rho\, d\vx = 0.
\end{equation}
Recalling that $[\tau]_+ = H(\tau)\tau$ and $\nabla [\tau]_+ = H(\tau)\nabla \tau$ a.e., we see that equation \eqref{eq:nullcond} is equivalent to
\[
    \int_{\Td} H(\tau)(\varphi \nabla \tau -\tau \nabla \varphi)^\top \bm\rho\,d\vx = 0.
\]
Consider the specific choice $\bm\rho = \varphi \nabla \tau - \tau \nabla\varphi$. Then the components of $\bm\rho$ belong to  $\TP(2\Omega_0^*)$, and so $\bm\rho \in \TP(\Omega|\Omega_0^*)^d$ by the assumption that $\Omega|\Omega_0^* \supseteq 2\Omega_0^*$. Therefore, this implies
\begin{align}
\int_{\Td} H(\tau)
\|\varphi \nabla \tau -\tau \nabla \varphi\|_2^2 \,d\vx = 0,
\end{align}
which shows $\varphi(\vx) \nabla \tau(\vx) -\tau(\vx) \nabla \varphi(\vx) = \bm 0$ for all $\vx$ in the open set $P := \{\vx \in \Td : \tau(\vx) > 0\}$. Finally, we will show this implies $\varphi(\vx) = c\cdot \tau(\vx)$ for some constant $c \in \R$. 

Let $B$ be any closed ball contained in $P$. By continuity, $\tau$ must be bounded away from zero over $B$, hence $\frac{\varphi}{\tau}$ is smooth over $B$. Therefore, by the quotient rule 
\[
    \nabla\left(\frac{\varphi}{\tau}\right)(\vx) = \frac{\tau(\vx) \nabla \varphi(\vx) - \varphi(\vx) \nabla\tau(\vx)}{\tau(\vx)^2} = \bm 0~~\forall \vx \in B.
\]
This implies $\frac{\varphi}{\tau}$ is constant over $B$, and so there exists a constant $c \in \R$ such that $\varphi(\vx) = c \cdot \tau(\vx)$ for all $\vx \in B$, or equivalently, $\varphi(\vx) - c \cdot \tau(\vx) = 0$ for all $\vx \in B$. However, since a trigonometric polynomial is a real analytic function, it cannot vanish on a ball contained $\Td$ unless it is identically zero. Therefore, it must be the case that $\varphi(\vx) = c \cdot \tau(\vx)$ for all $\vx \in \Td$, as claimed.
\end{proof}

Now we prove \Cref{lem:injectivity}.
\begin{proof}[Proof of \Cref{lem:injectivity}]
Given $\vw_0,\vw_1 \in \Peta$, define trigonometric polynomials $\tau_0 = \vw_0^\T\bm\gamma(\cdot)$ and $\tau_1 = \vw_1^\T\bm\gamma(\cdot)$, and suppose for $\vc_0 = \gF_\Omega[\tau_0]_+$ and $\vc_1 = \gF_\Omega[\tau_1]_+$ we have $\vc_0 = \vc_1$. Then $\gQ_{\vc_0} = \gQ_{\vc_1}$ and so $\ker \gQ_{\vc_0} = \ker \gQ_{\vc_1}$. By \Cref{lem:unique_nullspace} this implies $\tau_0 = c\cdot\tau_1$ for some $c\in\R$. Since $\|\tau_0\|_{L^2(\Td)} = \|\vw_0\|_2 = 1$ and $\|\tau_1\|_{L^2(\Td)} = \|\vw_1\|_2 = 1$, we must have $c = \pm 1$. Suppose it were the case that $c = -1$, i.e., $\tau_1 = -\tau_0$, and so $\gF_\Omega[\tau_0]_+ = \gF_\Omega[-\tau_0]_+$. By the identity $[t]_+-[-t]_+ = t$ for all $t\in \R$, we see this implies $\gF_\Omega([\tau_0]_+-[-\tau_0]_+) = \gF_\Omega \tau_0 = \bm 0$. Since $\Omega_0 \subset \Omega$, this can only be the case if $\tau_0 = 0$, which is ruled out by our assumption that $[\tau_0]_+ \neq 0$. Therefore, $c=1$, and so $\tau_0 = \tau_1$ or equivalently, $\vw_0 = \vw_1$, which proves the claim.
\end{proof}

Finally, we prove \Cref{thm:main} by constructing a dual certificate and invoking \Cref{lem:dual_cert}.
\begin{proof}[Proof of \Cref{thm:main}]
Let $\eta(\vw) = \|\gF_\Omega[\vw^\T\bm\gamma(\cdot)]_+\|_2$ and $f = a_1[\vw_1^\T\bm\gamma(\cdot)]_+$. Define $\vw_0 = \vw_1/\|\vw_1\|$ and $a_0 = \|\vw_1\|a_1$. Then we have $f = a_0[\vw_0^\T \bm\gamma(\cdot)]_+$ with $\vw_0 \in S_\eta$.
We will show that
$\vz_0 := \sign(a_0) \eta(\vw_0)^{-1} \gF_\Omega[\vw_0^\T\bm\gamma(\cdot)]_+$
yields a valid dual certificate $\Phi = \gL_\Omega \vz_0$ for the measure $\mu = a_0\delta_{\vw_0}$ with respect to the support set $\gW = \{\vw_0\}$. 

First, observe that for all $\vw \in \Peta$ we have
\[
\Phi(\vw)  = \sign(a_0)\eta(\vw_0)^{-1}\langle \gF_\Omega[\vw^\T\bm\gamma(\cdot)]_+,\gF_\Omega[\vw_0^\T\bm\gamma(\cdot)]_+\rangle.
\]
Plugging in $\vw = \vw_0$, and since $\eta(\vw_0) = \|\gF_{\Omega}[\vw_0^\T\bm\gamma(\cdot)]_+\|_2$, we see that
\[
\Phi(\vw_0) = \sign(a_0)\eta(\vw_0)^{-1}\|\gF_{\Omega}[\vw_0^\T\bm\gamma(\cdot)]_+\|_2^2 = \sign(a_0) \eta(\vw_0).
\]
Also, by the Cauchy-Schwarz inequality, for all $\vw \in \Peta$ we have
\[
|\Phi(\vw)|  \leq \eta(\vw_0)^{-1}\|\gF_{\Omega}[\vw^\T\bm\gamma(\cdot)]_+\|_2\, \|\gF_{\Omega}[\vw_0^\T\bm\gamma(\cdot)]_+\|_2 = \eta(\vw),
\]
with strict inequality when $\gF_{\Omega}[\vw_0^\top\bm\gamma(\cdot)]_+$ is not co-linear with $\gF_{\Omega}[\vw^\T\bm\gamma(\cdot)]_+$. Therefore, to prove $\Phi$ is a valid dual certificate, it suffices to show $\gF_{\Omega}[\vw_0^\top\bm\gamma(\cdot)]_+$ is not co-linear with $\gF_{\Omega}[\vw^\T\bm\gamma(\cdot)]_+$ when $\vw \neq \vw_0$.

Towards this end, suppose $\gF_{\Omega}[\vw^\T\bm\gamma(\cdot)]_+ = c \cdot \gF_{\Omega}[\vw_0^\T\bm\gamma(\cdot)]_+$ for some $c \in \R$. 
First, note that for all $\vv \in \Peta$ we have $\widehat{[\vv^\T\bm\gamma(\cdot)]_+}[\bm 0] = \int_{\Td}[\vv^\T\bm\gamma(\cdot)]_+d\vx > 0$, since otherwise this would imply $[\vv^\T\bm\gamma(\cdot)]_+$ is identically zero, and so $\eta(\vv) = \|\gF_\Omega[\vv^\T\bm\gamma(\cdot)]_+\|_2 =  0$, contrary to assumption that $\vv \in \Peta$. This implies $c = \widehat{[\vw^\T\bm\gamma(\cdot)]_+}[\bm 0]/\widehat{[\vw_0^\T\bm\gamma(\cdot)]_+}[\bm 0] > 0$. Therefore, by linearity of $\gF_\Omega$ and positive 1-homogeneity of the ReLU activation, we have $\gF_\Omega[\vw^\T\bm\gamma(\cdot)]_+ = \gF_\Omega[(c\vw_0)^\T\bm\gamma(\cdot)]_+$. However, by \Cref{lem:injectivity}, we see this implies $\vw = c\vw_0$, and since we assume $\|\vw\|_2 = \|\vw_0\|_2 =1$, this can be the case only if $c = 1$, and so $\vw = \vw_0$. Therefore, if $\vw \in \Peta$ is such that $\vw \neq \vw_0$ then $\gF_{\Omega}[\vw^\top\bm\gamma(\cdot)]_+$ is not co-linear with $\gF_{\Omega}[\vw_0^\top\bm\gamma(\cdot)]_+$, and so $|\Phi(\vw)| < \eta(\vw)$. This proves $\Phi$ is a dual certificate for the measure $\mu = a\delta_{\vw_0}$, and so by \Cref{lem:dual_cert} (with $\gW = \{\vw_0\}$), $f = f_{\mu}$ is the unique function space minimizer of \eqref{eq:opt_func_space}. Finally, since $\mu$ is 1-sparse, $f$ is also the unique function space minimizer of \eqref{eq:opt_discrete_measures} for all $W\geq 1$, and hence also of \eqref{eq:opt_param_space}, since \eqref{eq:opt_param_space} has the same function space minimizers as \eqref{eq:opt_discrete_measures}.
\end{proof}

\subsection{Proof of \Cref{thm:main2}}\label{sec:app:thm2_proof}

The proof of \Cref{thm:main2} relies on some basic algebraic facts regarding the zero-sets of trigonometric polynomials as originally developed in \cite{ongie2016off}, which we summarize here. 

First, we say $\tau:\torus^2\rightarrow \C$ is a \emph{2D trigonometric polynomial} if it has the form
\begin{equation}\label{eq:2dtau}
\tau(x_1,x_2) = \sum_{(k_1,k_2)\in \Gamma} c_{(k_1,k_2)}e^{i2\pi k_1x_1}e^{i 2\pi k_2x_2},~~(x_1,x_2)\in\torus^2,
\end{equation}
for some constants $c_{(k_1,k_2)}\in\C$ indexed by a finite subset $\Gamma \subset \mathbb{Z}^2$, and use
$Z(\tau)$ denote its zero set
\[Z(\tau) = \{(x_1,x_2) \in \torus^2 : \tau(x_1,x_2) = 0\}.\] Following \cite{ongie2016off}, we say $C \subset \torus^2$ is a \emph{trigonometric curve} if $C = Z(\tau)$ for some non-zero 2D trigonometric polynomial $\tau$, and $C$ is an infinite set that has no isolated points. In \cite{ongie2016off} it is shown that every trigonometric curve is a union of finitely many piecewise smooth curves that are closed and connected as subsets of $\torus^2$.

Given any 2D trigonometric polynomial $\tau$ of the form \eqref{eq:2dtau} we may associate a unique complex polynomial in two variables $\gP[\tau] \in \C[z_1, z_2]$ by making the substitutions $e^{i2\pi x_1} \mapsto z_1$ and $e^{i2\pi x_2} \mapsto z_2$, and multiplying by the smallest powers of $z_1$ and $z_2$ so that the resulting function has no negative powers of $z_1$ or $z_2$. In particular, if $\tau$ has the form \eqref{eq:2dtau} then $p = \gP[\tau]$ has the form
\[
    p(z_1,z_2) = z_1^{\ell_1} z_2^{\ell_2}\left(\sum_{(k_1,k_2)\in\Gamma} c_{(k_1,k_2)} z_1^{k_1}z_2^{k_2}\right),~~(z_1,z_2)\in \C^2,
\]
for some non-negative integers $\ell_1$ and $\ell_2$. Note $\tau(x_1,x_2) = 0$ if and only if $p(e^{i2\pi x_1},e^{i2\pi x_2}) = 0$.

Following \cite{ongie2016off}, we say the 2D trigonometric polynomial $\tau$ is \emph{irreducible} if $p = \gP[\tau]$ is irreducible as a polynomial in $\C[z_1,z_2]$, i.e., $p$ cannot be written as a product of two or more non-constant polynomials in $\C[z_1,z_2]$.
Also, we say the 2D trigonometric polynomial $\tau_0$ is a \emph{minimal polynomial} for the trigonometric curve $C$ if $C = Z(\tau_0)$ and $p_0 = \gP[\tau_0]$ is a product of non-constant irreducible polynomials $p_1,...,p_n \in \C[z_1,z_2]$ that are not scalar multiples of each other and the zero set of each $p_i$ has infinite intersection with the complex unit torus $\C\torus^2 = \{(z_1,z_2) \in \C^2 : |z_1|=|z_2|=1\}$. In \cite{ongie2016off} it is shown there is a unique (up to scaling) real-valued minimal polynomial associated with every trigonometric curve. In particular, if $C = Z(\tau)$ is an infinite set and $\tau$ is real-valued and irreducible, then $\tau$ is a minimal polynomial for $C$.

Additionally, we will need the following facts regarding irreducible and minimal polynomials.

\begin{myLem}[Adapted from \cite {ongie2016off}]\label{lem:tau_prop}~~\\[-0.5em]
\begin{enumerate}
    \item[(a)] Suppose $\tau, \sigma $ are irreducible trigonometric polynomials and $\tau \neq c\cdot \sigma$ for any $c\in\R$. Then $Z(\tau) \cap Z(\sigma)$ is a finite set.
    \item[(b)] Suppose $\tau$ is a minimal polynomial for the trigonometric curve $C$. Then $\nabla \tau(\vx) = \bm 0$ for at most finitely  many $\vx \in C$.
    \item[(c)] Suppose $\tau$ is a minimal polynomial for the trigonometric curve $C$, and $\sigma$ is a trigonometric polynomial that vanishes on $C$. Then $\tau$ divides $\sigma$, i.e., there exists a trigonometric polynomial $\beta$ such that $\sigma = \tau \cdot \beta$.
\end{enumerate}
\end{myLem}
Part (a) is special case of \cite[Corollary A.2]{ongie2016off}, part (b) is \cite[Proposition A.4]{ongie2016off}, and part (c) is \cite[Proposition A.3]{ongie2016off}.

Let $\Gamma_1 = \{[\pm 1,0]^\T, [0,\pm 1]^\T\}$. Then $\TP(\Gamma_1)$ is the space of 2D trigonometric polynomials realized by the restricted Fourier features embedding $\bm\gamma$ defined in \eqref{eq:gamma_restricted}. In particular, $\tau \in \TP(\Gamma_1)$ if and only if
\begin{equation}\label{eq:tau_reduced_form}
\tau(x_1,x_2) = w_1 \cos(2\pi x_1) + w_2\sin(2\pi x_1) + w_3\cos(2\pi x_2) +  w_4 \sin(2\pi x_2) 
\end{equation}
for some $w_1,w_2,w_3,w_4\in\R$.

Our first result shows that every $\tau \in \TP(\Gamma_1)$ is such that $C = Z(\tau)$ is a trigonometric curve for which $\tau$ is a minimal polynomial, and characterizes when $\tau$ is irreducible. The proof of this result, which relies on elementary calculus and algebra, is provided in the Supplementary Materials  \ref{sec:supp:lem_proof}.
\begin{myLem}\label{lem:taus_are_minimal}
Suppose $\tau \in \TP(\Gamma_1)$ is non-zero. Then $C = Z(\tau)$ is a trigonometric curve, and $\tau$ is a minimal polynomial for $C$. Furthermore, $\tau$ is irreducible except in the following cases:
\begin{enumerate}
    \item $\tau(x_1,x_2) = \alpha \sin(2\pi(x_1-t))$ for some $\alpha \in \R$, $t \in \torus$.
    \item $\tau(x_1,x_2) = \alpha \sin(2\pi(x_2-t))$ for some $\alpha \in \R$, $t \in \torus$.
    \item $\tau(x_1,x_2) = \alpha (\sin(2\pi(x_1-t_1)) + \sin(2\pi(x_2-t_2)))$ for some $\alpha \in \R$, $t_1,t_2 \in \torus$.
\end{enumerate}
Together, these cases correspond to all $\tau$ of the form \eqref{eq:tau_reduced_form} where the weight vector $\vw = [w_1,w_2,w_3,w_4]^\T$ belongs to the set $V$ defined in \eqref{eq:Vset}.
\end{myLem}

Next, we prove a key lemma, which shows that given any $\tau \in \TP(\Gamma_1)$, a second-order differential operator applied to $[\tau]_+$ yields a distribution (i.e., generalized function \cite{friedlander1998introduction}) supported on $Z(\tau)$. Later, we use this property to construct a dual certificate.

Below, for a distribution $D$ and $\mathcal{C}^\infty$-smooth test function $\varphi:\torus^2\rightarrow\R$, we let $\langle D,\varphi\rangle$ denote the evaluation of $D$ at $\varphi$. In particular, when $D$ can be identified with a function $g \in L^1(\torus^2)$, we have $\langle D,\varphi\rangle = \int_{\torus^2} g \varphi \,d\vx$.

\begin{myLem}\label{lem:laplace_curve}
Let $\tau\in \TP(\Gamma_1)$, and let $\Delta = \frac{\partial^2}{\partial x_1^2} + \frac{\partial^2}{\partial x_2^2}$ denote the Laplacian. Then $(\Delta+4\pi^2)[\tau]_+$ is a distribution acting on $\mathcal{C}^\infty$-smooth test functions $\varphi:\torus^2\rightarrow \R$ by
\[
\langle(\Delta+4\pi^2)[\tau]_+,\varphi\rangle = \oint_{Z(\tau)} \varphi \|\nabla\tau\| ds,
\]
where $\oint_C(\cdot) ds$ denotes the line integral of a function over the curve $C$.
\end{myLem}
\begin{proof}
First, observe that $\Delta \tau = -4\pi^2 \tau$ and so $(\Delta+4\pi^2)\tau = 0$. Also, we have $\nabla [\tau]_+ = H(\tau) \nabla \tau$ almost everywhere. 
Now, for any $C^\infty$-smooth test function $\varphi:\torus^2\rightarrow \R$ we have
\begin{align}
\langle \Delta[\tau]_+, \varphi\rangle & = -\langle \nabla [\tau]_+, \nabla\varphi\rangle \nonumber\\
& = -\langle H(\tau)\nabla \tau, \nabla \varphi\rangle \nonumber\\
& = -\langle H(\tau) ,\nabla\tau\cdot \nabla \varphi\rangle \nonumber\\
& = -\langle H(\tau) , \nabla\cdot(\varphi\nabla \tau)-\varphi \Delta\tau\rangle \label{eq:weirdstep}\\
& = -\langle H(\tau) ,\nabla\cdot(\varphi\nabla \tau)\rangle + \langle H(\tau), \varphi \Delta\tau\rangle \nonumber\\
& = -\langle H(\tau) ,\nabla\cdot(\varphi\nabla \tau)\rangle - 4\pi^2 \langle H(\tau) ,  \varphi\tau\rangle \nonumber\\
& = -\langle H(\tau) ,\nabla\cdot(\varphi\nabla \tau)\rangle - 4\pi^2 \langle [\tau]_+,  \varphi\rangle.\nonumber
\end{align}
where in \eqref{eq:weirdstep} we used the identity $\nabla \cdot (\varphi\nabla \tau) = \nabla \tau \cdot \nabla \varphi + \varphi \Delta \tau$.
Rearranging terms above gives
\begin{align*}
\langle (\Delta + 4\pi^2)\tau,\varphi\rangle & = -\langle H(\tau), \nabla \cdot (\varphi \nabla \tau)\rangle\\
& = -\int_{\tau > 0}  \nabla \cdot (\varphi \nabla \tau)\, d\vx\\
& = -\oint_{Z(\tau)} \phi (\nabla \tau \cdot \vn)\, ds
\end{align*}
where $\vn$ is an outward unit normal to the curve $Z(\tau)$, and this last step follows from the divergence theorem. Since $\tau$ is a minimal polynomial for the trigonometric curve $Z(\tau)$, $\nabla \tau(\vx) = \bm 0$ for at most finitely many $\vx \in Z(\tau)$. Therefore, $\vn(\vx) = -\frac{\nabla \tau(\vx)}{\|\nabla\tau (\vx)\|}$ for all $\vx \in Z(\tau)$ except possibly finite many points, and so we have
\[
\langle (\Delta + 4\pi^2)\tau,\varphi\rangle  = \oint_{Z(\tau)} \varphi \|\nabla \tau\|\, ds,
\]
as claimed.
\end{proof}

Now we proceed with proof of \Cref{thm:main2}, which we split into two lemmas. First, we prove that, for sufficient large sampling set $\Omega$, the injectivity condition of \Cref{lem:dual_cert} is satisfied assuming the parameter vectors $\vw_i$ generate irreducible 2D trigonometric polynomials.

\begin{myLem}\label{lem:linind_of_rtps}
Let $\vw_1,...,\vw_s \in \mathbb{S}^3$ be such that $\vw_i \neq \pm \vw_j$ for all $i\neq j$, and $\tau_i = \vw_i^\top \bm\gamma(\cdot) \in \TP(\Gamma_1)$ is irreducible. Suppose $\Omega \supseteq \{ \vk \in \mathbb{Z}^2 : \|\vk\|_1 \leq 2(s-1)\}$. Then for any measure $\nu \in \gM(\mathbb{S}^3)$ supported on the set $\gW = \{\pm \vw_1,...,\pm \vw_s\}$, we have $\gK_\Omega \nu = 0$ implies $f_{\nu} = 0$.
\end{myLem}
\begin{proof}
Assume $\nu$ has the form
\[
\nu = \sum_{i=1}^s \left(a_i \delta_{\vw_i} +  b_i\delta_{-\vw_i}\right),
\]
for some scalars $a_i,b_i \in \R$, $i\in [s]$. Define $\tau_i = \vw_i^\T\bm\gamma(\cdot)$ for all $i \in [s]$. If  $\gK_\Omega \nu = 0$ then
\begin{equation}\label{eq:innernull0}
\sum_{i=1}^s \left(a_i \gF_\Omega[\tau_i]_+ + b_i \gF_\Omega[-\tau_i]_+\right) = \gF_\Omega\left(\sum_{i=1}^s a_i [\tau_i]_+ + b_i [-\tau_i]_+\right) = 0.
\end{equation}
This implies for all $\varphi \in \TP(\Omega)$ we have
\begin{equation}\label{eq:innernull}
\left\langle\sum_{i=1}^s a_i[\tau_i]_+ + b_i [-\tau_i]_+, \varphi \right\rangle = \sum_{i=1}^s a_i \langle [\tau_i]_+,\varphi \rangle + b_i \langle [-\tau_i]_+,\varphi \rangle  =  0.
\end{equation}
Now, fix any $j \in [s]$, and suppose we set $\varphi = (\Delta+4\pi^2)\psi_j$ where $\psi_j = \prod_{i\neq j} \tau_i^2$, which has frequency support contained in $\{ \vk \in \mathbb{Z}^2 : \|\vk\|_1 \leq 2(s-1)\}$, hence belongs to $\TP(\Omega)$ by assumption. Then, starting from equation \eqref{eq:innernull} we have
\begin{align}
    0 & = \sum_{i=1}^s a_i \langle [\tau_i]_+,(\Delta+4\pi^2)\psi_j \rangle + b_i \langle [-\tau_i]_+,(\Delta+4\pi^2)\psi_j \rangle \nonumber \\
    & = \sum_{i=1}^s a_i \langle (\Delta+4\pi^2)[\tau_i]_+,\psi_j\rangle + b_i\langle (\Delta+4\pi^2)[-\tau_i]_+,\psi_j\rangle \nonumber \\
    & = \sum_{i=1}^s (a_i+b_i) \oint_{Z(\tau_i)} \psi_j \|\nabla \tau_i\| ds \label{eq:innernull_2ndlast} \\
    & = (a_j+b_j) \oint_{Z(\tau_j)} \psi_j \|\nabla \tau_j\| ds, \label{eq:innernull_last}
\end{align}
where in \eqref{eq:innernull_2ndlast} we applied \Cref{lem:laplace_curve}, and in \eqref{eq:innernull_last} all but the $j$th term in the sum vanishes because $\psi_j$ contains a factor of $\tau_i$ for all $i\neq j$.
Finally, the integrand in \eqref{eq:innernull_last} is non-negative and not identically zero since $Z(\psi_j) = \cup_{i\neq j}Z(\tau_j)$, and by \Cref{lem:tau_prop}(a) each of the $\tau_i$ can have only finitely many zeros on $Z(\tau_j)$, and by \Cref{lem:tau_prop}(b), $\nabla\tau_j$ has at most finitely many zeros over $Z(\tau_j)$. Therefore, the integral in \eqref{eq:innernull_last} non-zero, and so $a_j + b_j = 0$, or equivalently, $b_j = -a_j$ for all $j \in [s]$. However, by the identity $\tau_i = [\tau_i]_+ - [-\tau_i]_+$, from equation \eqref{eq:innernull0}, this implies
\[
\gF_\Omega\left(\sum_{i=1}^s a_i \tau_i \right) = 0.
\]
But since $\Gamma_1 \subset \Omega$, this gives $\sum_{i=1}^s a_i \tau_i = 0$, and so
\[
f_\nu = \int_{\Peta}[\vw^\T\bm\gamma(\cdot)]_+ d\nu(\vw) =  \sum_{i=1}^s a_i \tau_i = 0,
\]
as claimed.
\end{proof}

Finally, to finish the proof of \Cref{thm:main2}, we construct a dual certificate assuming that $f$ is a positive weighted linear combination of rectified \emph{irreducible} trigonometric polynomials belonging to $\TP(\Gamma_1)$.
\begin{myLem}
Let $\vw_1,...,\vw_s \in \mathbb{S}^3$ be such that $\vw_i \neq \pm \vw_j$ for all $i\neq j$, and $\tau_i = \vw_i^\top \bm\gamma(\cdot)$ is irreducible for all $i\in[s]$. 
Suppose $f = \sum_{i=1}^s a_i [\tau_i]_+$ for some scalars $a_i > 0$, and $\vy = \gF_\Omega f$ where $\Omega \supseteq \{ \vk \in \mathbb{Z}^2 : \|\vk\|_1 \leq 2s\}$. Let the weighting function $\eta$ be given by $\eta(\vw) =\int_{\torus^2}[\vw^\T\bm\gamma(\vx)]_+ d \vx$. Then there exists a dual certificate for the measure $\mu = \sum_{i=1}^s a_i \delta_{\vw_i}$ with respect to the support set $\gW = \{\pm \vw_1,...,\pm \vw_s\}$.
\end{myLem}
\begin{proof}
Let $\rho = \prod_{i=1}^s \tau_i^2$. Then $\hat{\rho}[\vk]$ has support contained within $\{ \vk \in \mathbb{Z}^2 : \|\vk\|_1 \leq 2s\}$, and so $\rho \in \TP(\Omega)$ by assumption. Also, let $\varphi = 1 -\alpha(\Delta + 4\pi^2)\rho$, where $\alpha>0$ is a scale factor to be determined later. Since the Laplacian $\Delta$ is a Fourier multiplier, we also have $\varphi \in \TP(\Omega)$. Finally, let $\vz = \gF_\Omega\varphi$, and let $\Phi(\vw) =\gL_\Omega\vz = \langle \gF_\Omega[\vw^\T\bm\gamma(\cdot)]_+,\vz\rangle$ be the corresponding dual function. 

Now fix any $\vw\in \mathbb{S}^3$, and let $\tau = \vw^\T\bm\gamma(\cdot)$. Then, since $\gF_\Omega^*\vz = \gF_\Omega^*\gF_\Omega \varphi = \varphi$, we have
\begin{align}\label{eq:phi_description}
\Phi(\vw) & = \langle \gF_\Omega[\tau]_+,\vz\rangle  =  \langle [\tau]_+,\varphi \rangle = \eta(\vw)-\alpha\langle (\Delta + 4\pi^2)[\tau]_+, \rho \rangle,
\end{align}
where in the last equality we used the fact that $\langle [\tau]_+,1\rangle =  \int_{\torus^2} [\vw^\T\bm\gamma(\vx)]_+\,d\vx = \eta(\vw)$.

Next, by \Cref{lem:laplace_curve}, observe that
\begin{align*}
 \langle (\Delta + 4\pi^2)[\tau]_+, \rho \rangle
& = \oint_{Z(\tau)} \rho \|\nabla \tau\| ds \geq 0,
\end{align*}
since $\rho \geq 0$. Furthermore, we have the bound
\begin{align}
 \oint_{Z(\tau)} \rho \|\nabla \tau\|ds & \leq \|\rho\|_{L^\infty(\torus^2)}\oint_{Z(\tau)}\|\nabla \tau\| ds \nonumber \\
 & = \|\rho\|_{L^\infty(\torus^2)}\langle(\Delta + 4\pi^2)[\tau]_+,1\rangle \label{eq:used_lemma}\\
 & = \|\rho\|_{L^\infty(\torus^2)}\langle[\tau]_+,(\Delta + 4\pi^2)1\rangle \label{eq:self_adjoint}\\
 & = 4\pi^2\|\rho\|_{L^\infty(\torus^2)}\langle[\tau]_+,1\rangle \nonumber\\
 & = 4\pi^2\|\rho\|_{L^\infty(\torus^2)}\eta(\vw), \nonumber
\end{align}
where the equality in \eqref{eq:used_lemma} follows from  \Cref{lem:laplace_curve}, and the equality in \eqref{eq:self_adjoint} follows from the definition of distributional derivatives.
Therefore, we have shown
\[
0 \leq \langle (\Delta + 4\pi^2)[\tau]_+, \rho \rangle \leq 4\pi^2\|\rho\|_{L^\infty(\torus^2)}\eta(\vw).
\]
Choosing $\alpha$ to be any constant satisfying $0 < \alpha < (4\pi^2 \|\rho\|_{L^\infty(\torus^2)})^{-1}$, then \eqref{eq:phi_description} combined with the inequality above shows that $0 \leq \Phi(\vw) \leq \eta(\vw)$ for all $\vw\in\S^3$.

Finally, we prove $\Phi(\vw) = \eta(\vw)$ if and only $\vw = \pm \vw_i$ for some $i \in [s]$. From \eqref{eq:phi_description}, this is equivalent to showing
\begin{equation}\label{eq:vanishing_cond}
\langle (\Delta + 4\pi^2)[\tau]_+, \rho\rangle = \oint_{Z(\tau)} \left(\textstyle\prod_{i=1}^s \tau_i^2\right) \|\nabla \tau\| ds = 0
\end{equation} if and only if $\tau = \pm \tau_i$ for some $i\in[s]$. First, if $\tau = \pm \tau_i$ then clearly \eqref{eq:vanishing_cond} holds. Conversely, if \eqref{eq:vanishing_cond} holds, then the integrand $(\prod_{i=1}^s \tau_i^2)\|\nabla \tau\|$ must vanish over $Z(\tau)$. However, by \Cref{lem:tau_prop}(b), $\nabla \tau$ has at most finitely zeros over $Z(\tau)$, and so  $\prod_{i=1}^s \tau_i^2$ must vanish everywhere on $Z(\tau)$. Since $\tau$ is a minimal polynomial for $Z(\tau)$,  by \Cref{lem:tau_prop}(c) this implies $\tau$ divides $\prod_{i=1}^s \tau_i^2$. Furthermore, since the $\tau_i$ are irreducible, and $\tau$ must be square-free (by virtue of being a minimal polynomial), this is possible if and only if $\tau = c \prod_{i\in I}\tau_i$ for some constant $c \in \R$ and index set $I \subset [s]$. Finally, since the product of two or more $\tau_i$ has Fourier support extending outside $\Gamma_1$, the only possibility is $\tau = c\cdot\tau_i$ for some $i\in[s]$, which implies $\vw = c\vw_i$ for some $i$. Finally, since $\|\vw\|_2=1$, we have $c = \pm 1$, which proves the claim.

Therefore, we have shown
$
\Phi(\vw) = \eta(\vw)$
for all $\vw = \pm \vw_i$, $i \in [s]$, while
$
0 \leq  \Phi(\vw) < \eta(\vw)
$
for all $\vw$ such that $\vw \neq \pm\vw_i$. This shows $\Phi$ satisfies the conditions of \Cref{lem:dual_cert}, i.e., $\Phi$ is a dual certificate for the measure $\mu^* = \sum_{i=1}^s a_i \delta_{\vw_i}$ with respect to the support set $\gW=\{\pm \vw_1,...,\pm\vw_s\}$.
\end{proof}

\section{INR fitting with Augmented Lagrangian Method}\label{sec:app:AL} Here we describe an Augmented Lagrangian (AL) approach to solving the equality constrained INR fitting problem \eqref{eq:opt_param_space1} (using the approximation $\tilde{\gF}_\Omega$ in place of $\F_\Omega$).
For INR parameters $\theta \in \Theta_W$, we introduce a Lagrange multiplier vector $\vq \in \Im(\F_\Omega)$ and penalty parameter $\sigma > 0$, and define the AL loss function $AL(\theta,\vq,\sigma)$ associated with \eqref{eq:opt_param_space1} by
\[
AL(\theta,\vq,\sigma) := R(\theta) + \langle\vq, \tilde{\gF}_\Omega f_\theta-\vy\rangle + \frac{\sigma}{2}\|\tilde{\gF}_\Omega f_\theta-\vy\|_2^2.
\] 
Then starting from initializations $\theta_0$, $\vq_0$, and $\sigma_0$, for $n=0,1,2,...,N-1$, we perform the updates
\begin{align}
\theta_{n+1} & = \argmin_{\theta}AL(\theta,\vq_n,\sigma_n)\label{eq:AL_main}\\
\vq_{n+1} & = \vq_n+ \sigma_n (\tilde{\gF}_\Omega f_{\theta_{n+1}}-\vy)\nonumber\\
\sigma_{n+1} & = \gamma \sigma_n,\nonumber
\end{align}
where $\gamma > 1$ is a fixed constant. To approximate the solution of the subproblem \eqref{eq:AL_main} we use a standard gradient-based INR training algorithm, warm-starting from the parameters $\theta_n$ obtained in the previous outer-loop iteration. See the Supplementary Materials \ref{sec:sm:exact} for more details on the hyperparameter settings used in the experiments in \Cref{sec:exact_recovery}.

\section*{Acknowledgments}
This work was supported by NSF CRII award CCF-2153371.

\bibliographystyle{siamplain}

\bibliography{references}

\newpage
\section*{Supplementary Material}
\addcontentsline{toc}{section}{Supplementary Material}

\section{Additional Experimental Results and Figures}\label{sec:supp:figures}

In this section, we provide additional reconstruction figures.

\begin{figure*}[htbp!]
    \centering    
\includegraphics[width=\textwidth]
{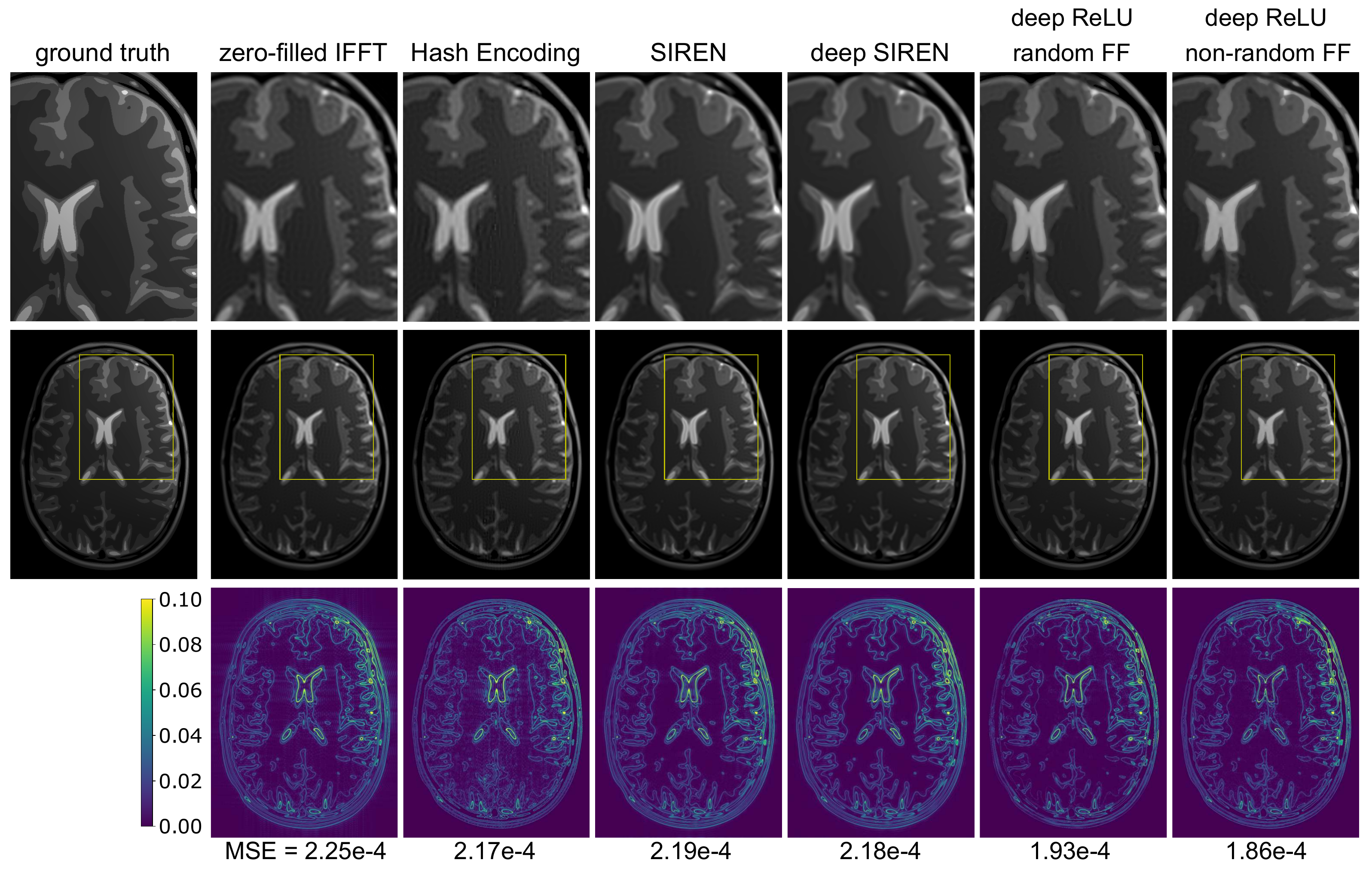}
    \caption{\footnotesize\textbf{\textsf{PWS-BRAIN} recovery using different INR architectures.} Top row: we compare the zero-filled IFFT reconstruction with the phantom recovery from Hash Encoding, depth 5 SIREN (SIREN), depth 14 SIREN (deep SIREN), and depth 15 INR (deep INR) with \emph{random} Fourier features and \emph{non-random} Fourier features using ReLU. The bottom row shows the absolute value of the difference with the ground truth. All grayscale images are shown on the scale [0, 1].}
    \label{fig:PWS_diff_methods}
\end{figure*}

\begin{figure*}[htbp!]
    \centering    \includegraphics[width=0.63\textwidth]{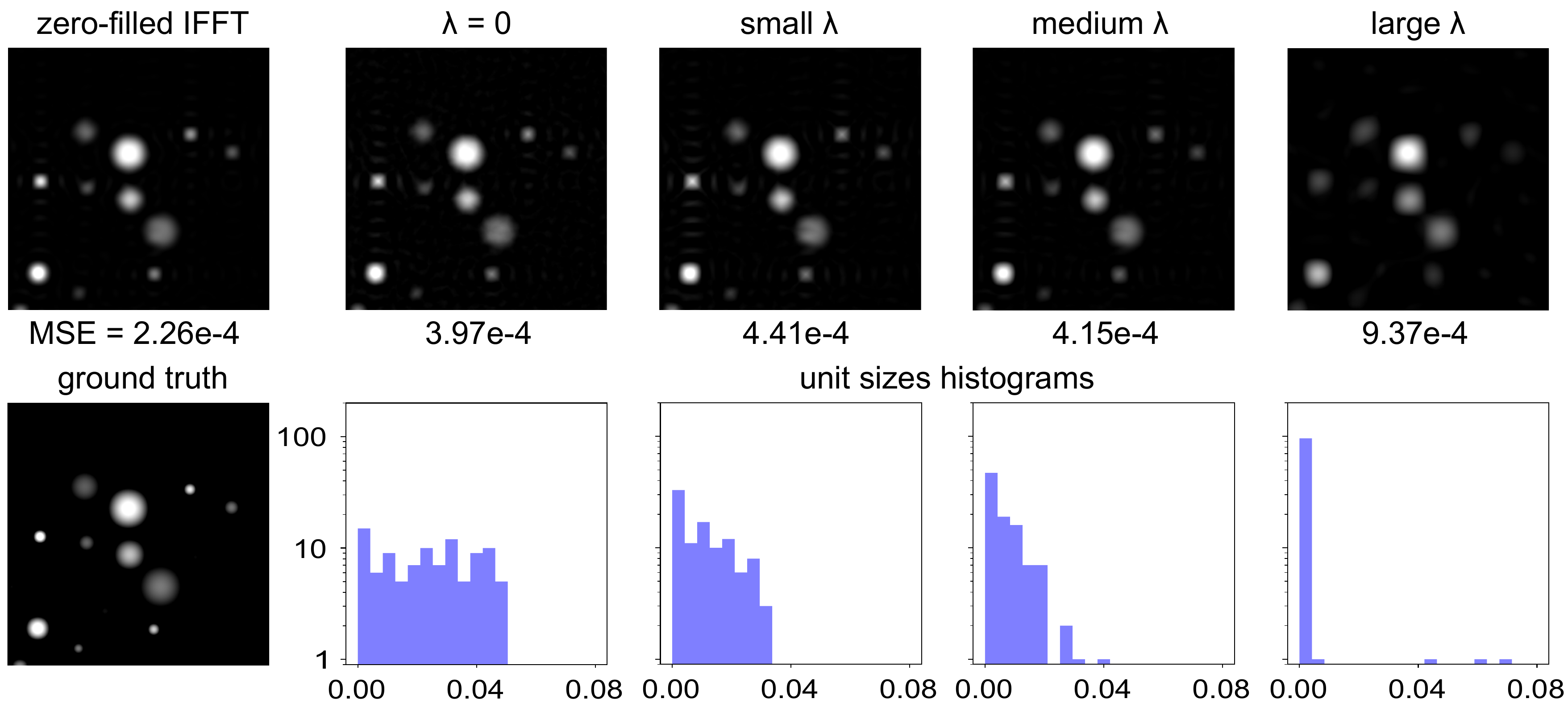} \\[0.4em]
    (a) standard weight decay regularization\\[0.4em]
\includegraphics[width=0.63\textwidth]{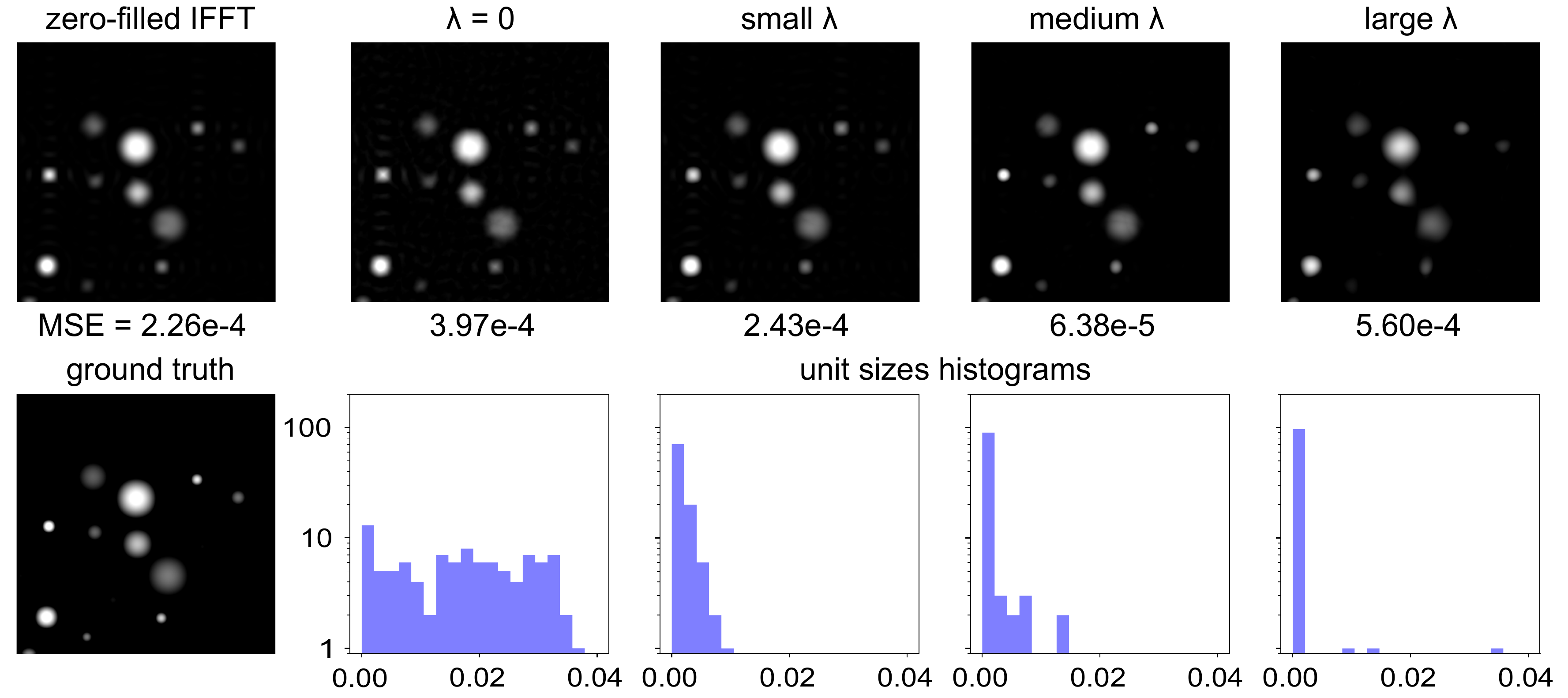}\\[0.4em]
   (b) modified weight decay regularization proposed in Theorem 1\\[0.4em]
   \includegraphics[width=0.63\textwidth]{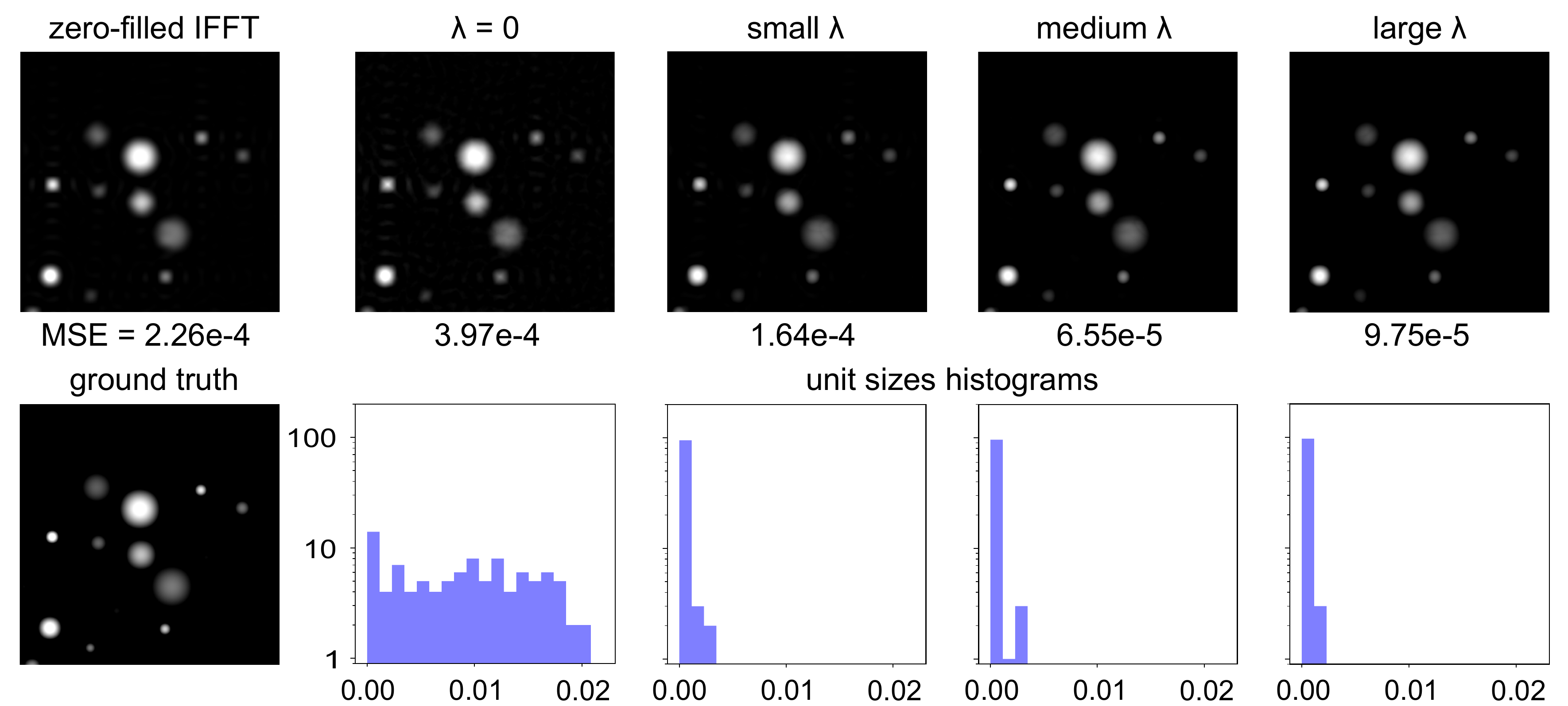}\\
   (b) modified weight decay regularization proposed in Theorem 2
     \caption{\footnotesize \textbf{Impact of regularization on the \textsf{DOT} phantom reconstruction and unit sizes of a trained shallow INR.} In the first row of each panel, we compare the zero-filled IFFT reconstruction with the reconstructions obtained with a shallow INR trained with standard/modified weight decay regularization for different values of regularization strength $\lambda$. The histograms in the second row of each panel represent the distribution of unit sizes $\{|a_i|\eta(\vw_i)\}_{i=1}^{100}$ of the trained INRs across different values of $\lambda$. As $\lambda$ increases, the unit sizes cluster at zero, indicating the trained INR is sparse in the sense that it has few active units. (Note: histograms are plotted on a log scale.)
}
    \label{fig:dot_std_mod_wd_10K}
\end{figure*}

\begin{figure*}[htbp!]
    \centering    
\includegraphics[width=\textwidth]{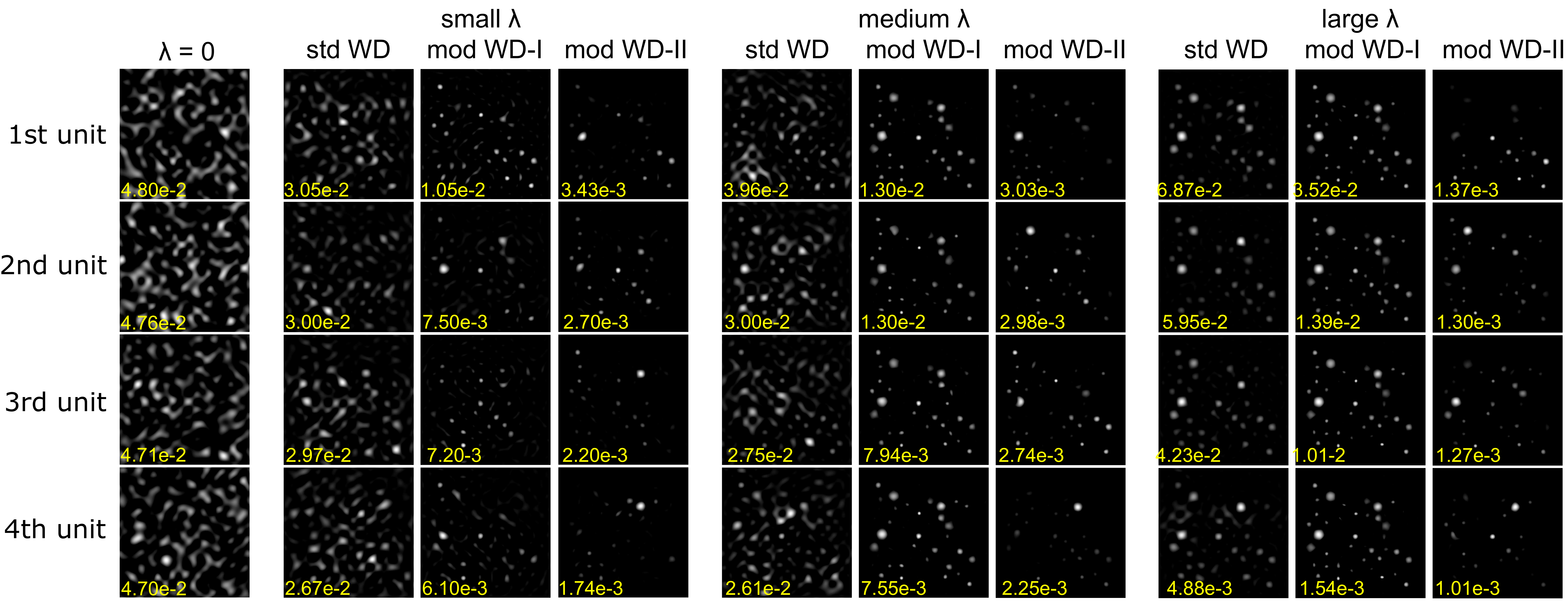}
     \caption{\footnotesize\textbf{Visualization of the 4 largest active units for the \textsf{DOT} phantom, obtained with a shallow INR.} The INR is trained with standard weight decay regularization (std WD), modified weight decay regularization proposed in Theorem 1 (mod WD-I), and the modified weight decay regularization proposed in Theorem 2 (mod WD-II). All images are normalized to the [0,1] range. The corresponding ``unit size'' $|a_i|\eta(\vw_i)$ is displayed on the image in yellow.}
    \label{fig:dot_active_units_normalized}
\end{figure*}

\begin{figure*}[htbp!]
    \centering    
    \includegraphics[width=\textwidth]{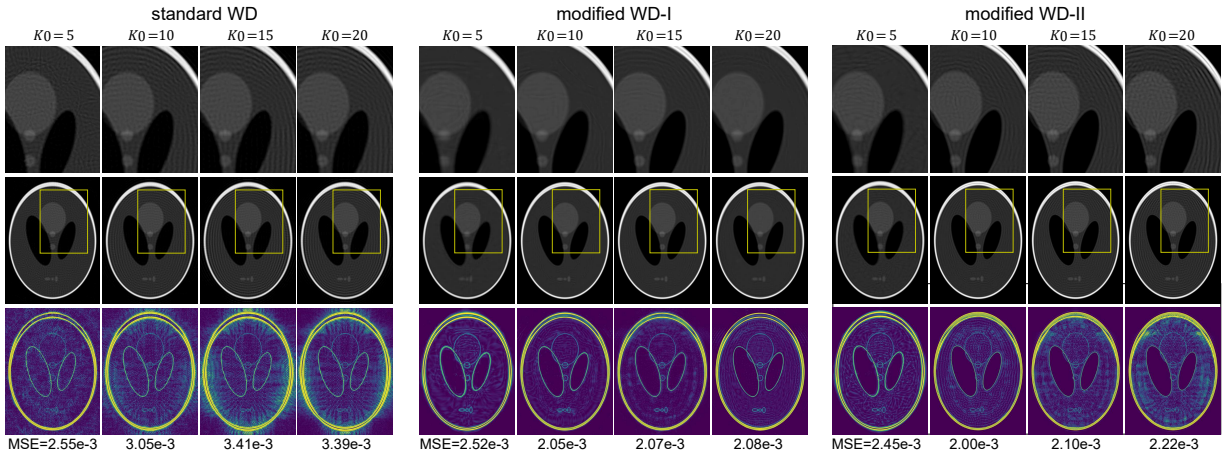}
    \caption{\footnotesize The effect of varying the maximum frequency for Fourier features $K_0$, on the \textsf{SL} phantom recovery trained using a shallow INR with standard WD, modified WD-I, and modified WD-II regularization. Observe that the lowest MSE is achieved for $K_0=10$    for modified WD-I and WD-II.}
    \label{fig:K0 test}
\end{figure*}

\vspace{2em}

\section{Experimental Details}\label{sec:supp:exp_details}

In this section, we provide additional details for some experiments included in the main text.

\subsection{Additional Details of  Experiments in Section \ref{sec:exact_recovery}}\label{sec:sm:exact}

In the setting of Theorem \ref{thm:main}, we use AL parameters $\sigma_0 = 10$, with $\gamma = 1.5$ and perform $N=60$ AL outer loop iterations. In the setting of Theorem \ref{thm:main2}, we use AL parameters $\sigma_0 = 10$, with $\gamma = 1.1$, and perform $N=100$ AL outer loop iterations. In both cases, we run $5,000$ iterations of the Adam optimizer with a learning rate of $1\times 10^{-3}$ to approximately minimize the AL inner loop subproblem \eqref{eq:AL_main}.

\subsection{Additional Details of Experiments in Section \ref{subsec:Effect-of-Depth}}
\label{depth}

This section presents the details of the experiments for reconstructing the \textsf{SL} and \textsf{PWC-BRAIN} phantoms using INR with different depths, with and without weight decay regularization, as shown in Figure \ref{fig:PWC_SL_depth_reg_eff}. 

For the \textsf{SL} phantom recovery, we set the maximum sampling frequency to $K=48$, and set the maximum sampling frequency of the Fourier features map to $K_0=10$. For the reconstructions with weight decay regularization, the regularization strength parameter $\lambda$ is set as follows: $ 1\times 10^{-6}$ for depth 2, $ 5\times 10^{-8}$ for depth 5, $ 5\times 10^{-7}$ for depth 10, $ 5\times 10^{-8}$ for depth 15 and $ 2\times 10^{-14}$ for depth 20. 

For the \textsf{PWC BRAIN} phantom recovery, we set $K=64$ and $K_0=10$. Additionally, $\lambda$ is set as follows: $ 1\times 10^{-4}$ for depth 2, $ 2\times 10^{-8}$ for depths 5 and 10, $ 1\times 10^{-7}$ for depth 15 and $ 1\times 10^{-8}$ for depth 20. In all experiments, both phantoms are recovered by training an INR of width 100 using the Adam optimizer for 40,000 iterations with a learning rate of $1\times 10^{-3}$ followed by 10,000 iterations with a learning of $ 1\times 10^{-4}$.

\subsection{Additional Details of Experiments in Section \ref{subsec:{Comparison with Different INR Architectures}}}
\label{diff_architectures_ex_details}

This section outlines the experimental details for reconstructing the \textsf{SL}, \textsf{PWC-BRAIN}, and \textsf{PWS-BRAIN} phantoms using different INR architectures, as shown in Figures \ref{fig:SL_diff_methods}, \ref{fig:PWC_diff_methods} and  \ref{fig:PWS_diff_methods}.

\paragraph{Hash Encoding}
For the hash encoding INR architecture, we use a fully-connected ReLU network with a depth of 3 and a width of 64 and the following parameters for the hash encoding layer: the number of levels $L=16$, number of feature dimensions per entry $F=2$, coarsest resolution $N_{min}=16$, finest resolution $N_{max}= 1024$. Additionally, details on training the INR using hash encoding for each phantom reconstruction are provided below.

\textsf{SL:} We set the hash table size $T= 2^{21}$, and train the INR using the Adam optimizer with a learning rate of $1 \times 10 ^{-2}$ for 100,000 iterations, and set the weight decay regularization strength  parameter to $\lambda = 1\times 10^{-4}$.
 
\textsf{PWC-BRAIN}:  We set the hash table size $T= 2^{19}$, and train the INR using the Adam optimizer with a learning rate of $1 \times 10 ^{-2}$ for 100,000 iterations, and set $\lambda=1\times 10^{-6}$.
 
\textsf{PWS-BRAIN}: We set the hash table size $T= 2^{19}$, and train the INR using the Adam optimizer with a learning rate of $1 \times 10 ^{-2}$ for 100,000 iterations, and set $\lambda=1\times 10^{-4}$.

\paragraph{SIREN}
For SIREN, we set the network's depth to 5 and for the deep SIREN, we set it to 14. For both SIREN and deep SIREN the hidden-layer widths is set to 256 and we set the bandlimit parameter $\omega_0 = 30$. Our experiments showed that increasing the depth from 5 to 14 can lead to lower image MSE. Additionally, in experiments with depth 5 SIREN networks, we did not apply weight decay regularization since it did not have a significant effect on the output. However, our experimental results indicate that with deep SIREN, using weight decay regularization improves the reconstructions and leads to lower image MSE values. Details on training the INR using SIREN for each phantom reconstruction are provided below.

\textsf{SL:} For SIREN, we train the INR using the Adam optimizer with a learning rate of $1 \times 10 ^{-5}$ for 50,000 iterations without weight decay regularization. For deep SIREN, we train the INR using the Adam optimizer with a learning rate of $1 \times 10 ^{-5}$ for 50,000 iterations and set $\lambda=5\times 10^{-3}$.
 
\textsf{PWC-BRAIN}: For SIREN, we train the network using the Adam optimizer with a learning rate of $1\times 10^{-6}$ for 50,000 iterations without weight decay regularization. For deep SIREN, we train the INR using the Adam optimizer for 10,000 iterations with a learning rate of $1\times 10^{-4}$ followed by 40,000 iterations with a learning rate of $1\times 10^{-5}$, and set $\lambda=5\times 10^{-2}$.

\textsf{PWS-BRAIN}: For SIREN, we train the network using the Adam optimizer with a learning rate of $1\times 10^{-6}$ for 50,000 iterations without weight decay regularization. For deep SIREN, we train the INR using the Adam optimizer for 50,000 iterations with a learning rate of $1\times 10^{-6}$, and set $\lambda$ to $5\times 10^{-3}$.

\paragraph{Deep INR with random Fourier features}
The INR width is set to 100, and the initial Fourier features layer is defined by sampling 256 random frequency vectors whose entries are i.i.d. Gaussian with mean-zero and standard deviation $10$, resulting in a total of $D=513$ features. Additionally, the network is trained using the Adam optimizer for 40,000 iterations with a learning rate of $1\times 10^{-3}$, followed by 10,000 iterations with a learning rate of $1\times 10^{-4}$. For each phantom, the INR's depth and weight decay hyperparameter are as below:

\textsf{SL:} The network's depth is set to 20 and $\lambda=1\times 10^{-11}$.

\textsf{PWC-BRAIN}: The network's depth is set to 20 and $\lambda=4\times 10^{-8}$.

\textsf{PWS-BRAIN}: The network's depth is set to 15 and $\lambda=5\times 10^{-7}$.

\paragraph{Deep INR with non-random Fourier features}

In these experiments, an INR of width 100 is trained using the Adam optimizer for 40,000 iterations with a learning rate of $1\times 10^{-3}$, followed by 10,000 iterations with a learning rate of $1\times 10^{-4}$. For each phantom, we define $K, K_0$, the INR's depth and weight decay hyperparameter $\lambda$ as below:

\textsf{SL:} The INR's depth is set to 20, $K=48$, $K_0=10$ and  $\lambda=2\times 10^{-14}$.

\textsf{PWC-BRAIN}: The INR's depth is set to 20, $K=64$, $K_0=10$, and  $\lambda=1\times 10^{-8}$.

\textsf{PWS-BRAIN}: The INR's depth is set to 20, $K=64$, $K_0=10$, and  $\lambda=5\times 10^{-8}$.

\subsection{Experimental Details of the Reconstructions Shown in Figures \ref{fig:std_mod_wd_10K}, \ref{fig:dot_std_mod_wd_10K}, \ref{fig:brain_active_units_normalized}, and \ref{fig:dot_active_units_normalized}}
\label{brain_dot_unit_sizes} This section describes the experimental setup for the effect of regularization on the \textsf{PWC BRAIN} and \textsf{DOT} phantoms reconstruction and
unit sizes of a trained shallow INR as shown in Figures \ref{fig:std_mod_wd_10K} and \ref{fig:dot_std_mod_wd_10K}. 
For the \textsf{PWC-BRAIN} phantom, we set $K=64$, $K_0=20$, width 500. For the \textsf{DOT} phantom, we set $K=32$, $K_0=10$, INR width to 100. For both phantoms, we train the INR networks using the Adam optimizer for 10,000 iterations with a learning rate of $1\times 10^{-3}$. The weight decay regularization hyperparameter $\lambda$ is selected by a grid search to optimize the image domain MSE in each experiment. Additionally, we use the same experimental settings and parameters to visualize the four largest active units shown in Figures \ref{fig:brain_active_units_normalized} and \ref{fig:dot_active_units_normalized}.

\subsection{Experimental Details of the Reconstructions Shown in Figure \ref{fig:K0 test}}
This section outlines the experimental details for reconstructing the \textsf{SL} phantom using a depth 2 INR with different $K_0$ as shown in Figure \ref{fig:K0 test}. 
In these experiments, we set $K=48$ and train a shallow INR of width 100 using the Adam optimizer for 40,000 iterations with a learning rate of $1\times 10^{-3}$, followed by 10,000 iterations with a learning rate of $1\times 10^{-4}$. We trained the INR using standard WD, modified WD-I, and modified WD-II regularization and tested varying values of $K_0$ for each. In each experiment, we selected the regularization parameter $\lambda$ using a grid search to optimize the image MSE. For the main results included in the paper, we set $K_0=10$ for the \textsf{SL} phantom recovery. When $K_0=10$, we have $n=220$ frequency vectors, plus one constant feature, which gives a total of $D = 2n + 1 = 441$ features. 

\section{Proof of Lemma \ref{lem:taus_are_minimal}}\label{sec:supp:lem_proof}
For convenience, we restate Lemma \ref{lem:taus_are_minimal} before giving its proof. Recall that $\TP(\Gamma_1)$ is the space of 2D trigonometric polynomials $\tau:\torus^2\rightarrow\R$ of the form
\[
\tau(x_1,x_2) = a_1 \cos(2\pi x_1) + b_1\sin(2\pi x_1) + a_2\cos(2\pi x_2) +  b_2 \sin(2\pi x_2)
\]
for some $a_1,a_2,b_1,b_2\in \R$.
\begin{myLemStar}[Restatement of \Cref{lem:taus_are_minimal}]
    Suppose $\tau \in \TP(\Gamma_1)$ is non-zero. Then $C = Z(\tau)$ is a trigonometric curve, and $\tau$ is a minimal polynomial for $C$. Furthermore, $\tau$ is irreducible except in the following cases:
    \begin{enumerate}
        \item $\tau(x_1,x_2) = \alpha \sin(2\pi(x_1-t))$ for some $\alpha \in \R$, $t \in \torus$.
        \item $\tau(x_1,x_2) = \alpha \sin(2\pi(x_2-t))$ for some $\alpha \in \R$, $t \in \torus$.
        \item $\tau(x_1,x_2) = \alpha (\sin(2\pi(x_1-t_1)) + \sin(2\pi(x_2-t_2)))$ for some $\alpha \in \R$, $t_1,t_2 \in \torus$.
    \end{enumerate}
\end{myLemStar}

\begin{proof}
    Let $\tau \in \TP(\Gamma_1)$ be non-zero, and $C = Z(\tau)$. First, we note that $C$ is nonempty: we have $\hat{\tau}[0,0] = \int\int_{\torus^2}\tau(x_1,x_2)dx_1 dx_2 = 0$, which implies $\tau$ must change sign over $\torus^2$, and so by the intermediate value theorem $\tau$ must have at least one zero.
    
    Next, we prove that $C$ has no isolated points (i.e., $\tau$ has no isolated zeros). Suppose, by way of contradiction, that $\tau$ has an isolated zero at a point $(x_1^*,x_2^*)$. Then, by continuity of $\tau$, it must have a strict local maximum or minimum at $(x_1^*,x_2^*)$. In particular, the Hessian matrix $\nabla^2 \tau(x_1^*,x_2^*)$ must be positive semi-definite or negative semi-definite. And since
    \[
    \tau(x_1,x_2) = a_1 \cos(2\pi x_1) + b_1\sin(2\pi x_1) + a_2\cos(2\pi x_2) +  b_2 \sin(2\pi x_2)
    \]
    for some $a_1,a_2,b_1,b_2\in\R$, we have
    \[
    \nabla^2 \tau(x_1^*,x_2^*) = -4\pi^2 \begin{bmatrix}
    a_1 \cos(2\pi x_1^*) + b_1 \sin(2\pi x_1^*) & 0 \\ 0 & a_2 \cos(2\pi x_2^*) + b_2 \sin(2\pi x_2^*)
    \end{bmatrix}.
    \]
    Also, $\tau(x_1^*,x_2^*) = 0$ implies $a_2 \cos(2\pi x_2^*) + b_2 \sin(2\pi x_2^*) = -(a_1 \cos(2\pi x_1^*) + b_1 \sin(2\pi x_1^*))$. Therefore, if $a_1 \cos(2\pi x_1^*) + b_1 \sin(2\pi x_1^*)$ were non-zero, this would imply the Hessian is indefinite, so the only possibility is
    \[a_1 \cos(2\pi x_1^*) + b_1 \sin(2\pi x_1^*) = a_2 \cos(2\pi x_2^*) + b_2 \sin(2\pi x_2^*) = 0.\]
    This implies
    \[
    a_1 = -\alpha \sin(2\pi x_1^*), b_1 = \alpha \cos(2\pi x_1^*),
    a_2 = -\beta \sin(2\pi x_2^*),
    b_2 = \beta \cos(2\pi x_2^*)
    \]
    for some scalars $\alpha,\beta \in \R$. Simple trigonometric identities show that this implies
    \[
    \tau(x_1,x_2) = \alpha \sin(2\pi(x_1-x_1^*)) + \beta \sin(2\pi(x_2-x_2^*)).
    \]
    Finally, since $(x_1^*,x_2^*)$ is a critical point of $\tau$, we have
    \[
    \nabla \tau (x_1^*,x_2^*) = 2\pi (\alpha,\beta) = (0,0)
    \]
    and so $\alpha = \beta = 0$, i.e., $\tau$ is identically zero, a contradiction. Therefore, $\tau$ cannot have any isolated zeros, and so $C = Z(\tau)$ is a trigonometric curve.
    
    Now we prove that $\tau$ is a minimal polynomial for $C$. To this end, we show that the conversion of $\tau$ into a complex polynomial in two variables is either irreducible or factors into a product of distinct irreducible factors. First, expanding $\tau$ as a sum of complex exponentials, we have
    \[
    \tau(x_1,x_2) = c_1 e^{i2\pi x_1} + \overline{c_1} e^{-i2\pi x_1} + c_2 e^{i2\pi x_2} + \overline{c_2} e^{-i2\pi x_2}
    \]
    where $c_1, c_2 \in \mathbb{C}$ are such that either $c_1 \neq 0$ or $c_2 \neq 0$. Let $p = \gP[\tau]$ be the conversion of $\tau$ to a complex polynomial in two variables. Then $p$ has the form
    \[
    p(z_1,z_2) = z_1^{q_1} z_2^{q_2} \left(c_1 z_1 + \overline{c_1} z_1^{-1} + c_2 z_2 + \overline{c_2} z_2^{-1}\right)
    \]
    where $q_1,q_2\in\{0,1\}$.
    We consider the following cases:
    
    Case 1: $c_1 = 0$ or $c_2 = 0$. Then $p$ has the form
    \[
    p(z_1,z_2) = z_i(c_iz_i^{-1} + \overline{c_i}z_i) = c_i + \overline{c_i}z_i^2
    \]
    for some $i=1,2$, which factors as
    \[
    p(z_1,z_2) = \overline{c_i}(z_i - r)(z_i+r)
    \]
    where $r = \sqrt{-c_i/\overline{c_i}}$ satisfies $|r|=1$. Note that the linear factors $z_i+r$  and $z_i-r$ are irreducible. Also, their zero-sets $Z(z_i-r) = \{(z_1,z_2) : z_i =  r\}$ and $Z(z_i+r) = \{(z_1,z_2) : z_i = -r\}$ are distinct, and have infinite intersection with the complex unit torus. Therefore, $\tau$ is a minimal polynomial for $Z(\tau)$. Converting $p$ back into a trigonometric polynomial gives $\tau(x_1,x_2) = \alpha \sin(2\pi (x_i-t))$ for some $i\in\{1,2\}$, $\alpha \in \R$ with $\alpha \neq 0$, and $t\in\torus$, which corresponds to the first two cases in the statement of the lemma.
    
    Case 2: $c_1\neq 0$ and $c_2\neq 0$. Then $p$ has the form
    \begin{equation}\label{eq:pform1}
    p(z_1,z_2) = z_1z_2\left(c_1 z_1 + \overline{c_1} z_1^{-1} + c_2 z_2 + \overline{c_2} z_2^{-1}\right) = c_1 z_1^2z_2 + \overline{c_1} z_2 +  c_2 z_2^2z_1 + \overline{c_2} z_1.
    \end{equation}
    If $p$ is irreducible, there is nothing to show since then $\tau$ is a minimal polynomial by definition. Instead, assume $p$ is reducible, i.e., $p$ factors as $p = p_1p_2$ for some non-constant polynomials $p_1,p_2 \in \mathbb{C}[z_1,z_2]$. Since $p$ has total degree 3, the only possibility is that one of the factors has total degree 1 and the other factor has total degree 2, i.e., $p$ has a factorization of the form
    \begin{align}
    p(z_1,z_2) & = (a_0 + a_1 z_1 + a_2 z_2)(b_0 + b_1 z_1 + b_2 z_2 + b_3 z_1z_2 + b_4 z_1^2 + b_5 z_2^2) \nonumber \\
    & = a_{1} b_{4} z_{1}^{3} + a_{2} b_{5} z_{2}^{3} + (a_{1} b_{3} + a_{2} b_{4}) z_{1}^{2}z_2 + (a_{1} b_{5} + a_2b_3) z_2^2z_1 \label{eq:pform2} \\ 
    &\quad  + (a_{0} b_{4} + a_{1} b_{1})z_1^2  + \left(a_{0} b_{5} + a_{2} b_{2}\right) z_{2}^{2} 
     + (a_1b_2 + a_2 b_1 + a_0b_3) z_1z_2 \nonumber \\ 
    &\quad + (a_2b_0 + a_0b_2)z_{2} + (a_{0} b_{1} + a_{1} b_{0})z_1 + a_{0} b_{0} \nonumber
    \end{align}
    Comparing coefficient between \eqref{eq:pform1} and \eqref{eq:pform2}, we see that:
    \begin{gather*}
    a_0b_0 = a_1b_4 = a_2b_5 = a_0b_4 + a_1b_1 = a_{0} b_{5} + a_{2} b_{2} = a_1b_2 + a_2 b_1 + a_0b_3 = 0,\\
    c_1 = a_1b_2 + a_2b_4 = \overline{a_2b_0 + a_0b_2} \neq 0,\\
    c_2 = a_1b_5 + a_2b_3 = \overline{a_0b_1+a_1b_0} \neq 0.
    \end{gather*}
    From the equation $a_0b_0 = 0$, either $a_0 = 0$ or $b_0 = 0$ (or both).  First, suppose $b_0 = 0$. From the equations above, this implies $a_0b_2 \neq 0$ and $a_0b_1 \neq 0$, and so $a_0,b_1,b_2$ are all nonzero. 
    If $a_1 \neq 0$ then the equation $a_1b_4 = 0$ implies $b_4 = 0$. Together with the equation $a_0b_4 +a_1b_1 = 0$ 
    this implies $b_1 = 0$, a contradiction. Therefore, $a_1=0$. By the same reasoning, if $a_2 \neq 0$ then this would imply $b_2 = 0$, a contradiction, and so $a_2=0$. But if $a_1 = a_2 = 0$ then the first factor $a_0 + a_1z_1 + a_2 z_2$ is a constant polynomial, and so $p$ is irreducible, contradicting our assumption that $p$ was reducible. 
    
    Therefore, if $p$ is reducible, then $b_0 \neq 0$ and $a_0 = 0$. From the equations above this implies $a_2 b_0 \neq 0$ and $a_1b_0 \neq 0$, i.e., $b_0,a_1,a_2$ are all non-zero. This in turn implies $b_1=b_2=b_4=b_5 = 0$. Therefore, the only possible factorization of $p$ is
    \[
    p(z_1,z_2) = (a_1 z_1 + a_2 z_2)(b_0 + b_3 z_1z_2)
    \]
    where $a_1,a_2,b_0,b_3 \in \mathbb{C}$ are all non-zero, and satisfy $a_1 b_3 = \overline{a_2b_0}$ and $a_2b_3 = \overline{a_1b_0}$. In particular, from these relations imply $|b_0| = |b_3|$ and $|a_1| = |a_2|$. Note that $p_1(z_1,z_2) = a_1 z_1 + a_2 z_2$ is irreducible since it is degree 1. Also, the zero-set of $p_1$ has an infinite intersection with the complex unit torus: substituting $z_i = e^{i2\pi x_i}$, $i=1,2$, the condition $a_1 z_1 + a_2 z_2 = 0$ is equivalent to $e^{j2\pi (x_1-x_2)} = a$ where $a = -a_2/a_1$ is such that $|a|=1$. This has infinitely many solutions $\{(x_1,x_2) \in \mathbb{T}^2 : x_1 - x_2 = t\}$ where $t \in \mathbb{T}$ is such that $e^{i2\pi t} = a$.
    
    Also, the degree 2 polynomial $p_2(z_1,z_2) = b_0 + b_3 z_1z_2$ is irreducible since it is easy to show it cannot be written as a product of two degree 1 polynomials. Furthermore, the zero-set of $p_2$ has infinite intersection with the complex unit torus: substituting $z_i = e^{i2\pi x_i}$, $i=1,2$, the condition $b_0 + b_3 z_1z_2 = 0$ is equivalent to $e^{i2\pi(x_1+x_2)} = b$ where $b = -b_3/b_0$ is such that $|b|=1$. This has infinitely many solutions $\{(x_1,x_2) \in \mathbb{T}^2 : x_1 + x_2 = t\}$ where $t \in \mathbb{T}$ is such that $e^{i2\pi t} = b$.
    
    Therefore, we have shown that if $p$ is reducible, it is a product of distinct irreducible factors, each of which vanish on an infinite subset of the complex unit torus, which proves that $\tau$ is a minimal polynomial. Finally, converting $p$ back into a trigonometric polynomial, it is straightforward to show that in this case
    \[
    \tau(x_1,x_2) = \alpha (\sin(2\pi(x_1-t_1)) + \sin(2\pi(x_2-t_2)))
    \]
    for some $\alpha \in \R$ with $\alpha \neq 0$ and some $t_1,t_2 \in \torus$, which corresponds to the third case in the statement of the lemma.
    \end{proof}

\end{document}